\documentclass[12pt]{article}

\usepackage[titlenumbered,ruled]{algorithm2e}
\usepackage{bm}
\usepackage{amsmath}
\usepackage{graphicx,psfrag,epsf}
\usepackage{enumerate}
\usepackage{natbib}
\usepackage{url} % not crucial - just used below for the URL 
\usepackage{enumitem}
\usepackage{subcaption} % Required for subfigures
\usepackage{mathtools}
\usepackage{booktabs}
\usepackage{multirow}
\usepackage{array}

\usepackage{algpseudocode}
\usepackage{float} % Allows use of [h] float specifier
\usepackage{color}
\usepackage{amsmath}
\usepackage{amssymb}
\usepackage{amsthm}

\usepackage[colorlinks=true,
            linkcolor=blue,    % 章节、公式等的颜色
            citecolor=blue,    % 文献引用的颜色
            urlcolor=blue,     % 超链接的颜色
            anchorcolor=blue]{hyperref}

\newcommand{\blind}{1}

\DeclareMathOperator*{\argmin}{arg\,min}

\def\sign{\operatorname{sign}}
\def\T{\top}

\def\E{\mathrm{E}}
\def\cov{\mathrm{cov}}
\def\var{\mathrm{var}}

\def\I{\mathrm{I}}
\def\II{\mathrm{II}}
\def\III{\mathrm{III}}

\def\P{\Pr}
\def\tr{\mathrm{tr}}

\def\mo{\mathrm{o}}
\def\mO{\mathrm{O}}

\def\mc{\mathrm{c}}
\def\mC{\mathrm{C}}
\def\md{\mathrm{d}}
\def\me{\mathrm{e}}

\def\mr{\mathrm{r}}

\def\RR{\mathbb{R}}

\def\tilde{\widetilde}
\def\hat{\widehat}

\def\polyLog{\operatorname{polyLog}}
\def\prox{\operatorname{prox}}

\def\bone{\mathbf{1}}
\def\bzero{\mathbf{0}}

\def\bfa{\boldsymbol{a}}
\def\bfA{\boldsymbol{A}}
\def\bfb{\boldsymbol{b}}
\def\bfB{\boldsymbol{B}}

\def\bfb{\boldsymbol{b}}

\def\bfD{\boldsymbol{D}}

\def\bfM{\boldsymbol{M}}

\def\bfq{\boldsymbol{q}}

\def\bfS{\boldsymbol{S}}

\def\bfu{\boldsymbol{u}}

\def\bfv{\boldsymbol{v}}
\def\bfV{\boldsymbol{V}}
\def\bfw{\boldsymbol{w}}
\def\bfx{\boldsymbol{x}}
\def\bfX{\boldsymbol{X}}

\def\bfz{\boldsymbol{z}}

\def\bbeta{\boldsymbol{\beta}}
\def\bdelta{\boldsymbol{\delta}}
\def\bepsilon{\boldsymbol{\epsilon}}
\def\boldeta{\boldsymbol{\eta}}
\def\bgamma{\boldsymbol{\gamma}}

\def\bSigma{\boldsymbol{\Sigma}}

\def\cR{\mathcal{R}}
\def\cV{\mathcal{V}}
\def\cX{\mathcal{X}}

\begin{document}

\def\spacingset#1{\renewcommand{\baselinestretch}%
{#1}\small\normalsize} \spacingset{1}

\newtheorem{assumption}{Assumption}
\newtheorem{proposition}{Proposition}
\newtheorem{theorem}{Theorem}
\newtheorem{remark}{Remark}
\newtheorem{lemma}{Lemma}
\newtheorem{corollary}{Corollary}
\newtheorem{definition}{Definition}
\newtheorem{example}{Example}
\newtheorem{assumptionalt}{Assumption}[assumption]
\newenvironment{assumptionp}[1]{
  \renewcommand\theassumptionalt{#1}
  \assumptionalt
}{\endassumptionalt}

%%%%%%%%%%%%%%%%%%%%%%%%%%%%%%%%%%%%%%%%%%%%%%%%%%%%%%%%%%%%%%%%%%%%%%%%%%%%%%

\if1\blind
{
  \title{\bf Transfer Learning for Moderate-Dimensional Ridge-Regularized Robust Linear Regression}
  \author{Lingfeng Lyu, Xiao Guo\thanks{Corresponding author: xiaoguo@ustc.edu.cn} \ and Zongqi Liu\bigskip \\
    \small 
    Department of Statistics and Finance, University of Science and Technology of China}
    \date{ }
  \maketitle
} \fi

\if0\blind
{
  \bigskip
  \bigskip
  \bigskip
  \begin{center}
    {\LARGE\bf Transfer Learning for Moderate-Dimensional Ridge-Regularized Robust Linear Regression}
\end{center}
  \medskip
} \fi

\bigskip

\begin{abstract}
This paper studies transfer learning for ridge-regularized robust linear regression in the moderate-dimensional regime, where the number of predictors is of the same order as the sample size and the regression coefficients are not assumed to be sparse.
		We propose Trans-RR, which combines a robust ridge estimator from a source study with a robust ridge correction based on the target study.
		Under mild assumptions, we characterize the asymptotic estimation error of the proposed estimator and show that leveraging source data can substantially improve estimation accuracy relative to the traditional single-study ridge-regularized robust estimator.
		Simulation results and a real-data analysis support the theory and illustrate both positive and negative transfer as the discrepancy between the source and target studies varies.
\end{abstract}


\vfill

\newpage
\spacingset{1.45} % DON'T change the spacing!
\section{Introduction}
\label{sec:intro}

Modern statistical analyses often involve several related datasets collected from different studies, populations, or experiments.
A basic question is how to use these datasets together to improve prediction and estimation for a study of interest.
Transfer learning addresses this question by borrowing useful information from related tasks.
It is now standard in machine learning and has been successful in applications such as natural language processing, remote sensing, and computer vision \cite{chen2022analysis,MA2024113924,gopalakrishnan2017deep}.
In statistics, it has also become an important tool for improving performance in multi-study problems.

In many contemporary applications, regression problems arise in a moderate-dimensional regime, where the number of predictors is of the same order as the sample size and sparsity is often not a reasonable structural assumption.
At the same time, heavy-tailed errors or outlying observations may substantially affect estimation accuracy.
This setting arises naturally in multisite metabolomics studies, where many metabolites are measured simultaneously, between-cohort heterogeneity is often present, and outliers or other data contamination can be an important practical concern \citep{liu2024modeling}.
These features make transfer learning particularly challenging: related source studies may contain useful information, but effective borrowing requires methods that are robust to contamination and suitable for moderate-dimensional, non-sparse settings.

Existing theoretical work on transfer learning covers several important settings.
For linear regression, \cite{chen2015data} studied data-enriched regression in a fixed-dimensional setting, and \cite{pmlr-v139-tripuraneni21a} analyzed linear models with a shared low-dimensional representation across tasks.
In the high-dimensional sparse regime, \cite{bastani2021predicting} considered transfer learning with proxy data, while \cite{li2022transfer} established prediction and estimation guarantees for sparse linear regression.
For high-dimensional generalized linear models, \cite{Tian2023Transfer} and \cite{Li2024Estimation} developed transfer learning methods with theoretical guarantees.
Transfer learning has also been studied for nonparametric classification \citep{10.1214/20-AOS1949}, nonparametric regression \citep{cai2024transfer}, and settings with unreliable source data \citep{fan2023robust}.
However, these works do not apply to the moderate-dimensional robust setting considered here.

Robust regression has, in contrast, been extensively studied in the single-study setting.
For classical M-estimation, a substantial body of work has established asymptotic results when $p/n \to 0$ while $p \to \infty$; see, for example, \cite{huber1973robust,portnoy1984asymptotic,portnoy1985asymptotic,portnoy1986asymptotic,portnoy1987central,mammen1989asymptotics}.
When $p/n \to \kappa \in (0,1)$, robust regression has a qualitatively different asymptotic behavior \citep{el2013robust}.
When $p/n \to \kappa > 0$, \cite{el2018impact} proposed a ridge-regularized robust estimator to address the nonexistence of the ordinary robust estimator.
However, these results do not directly extend to the transfer learning framework when related source data are available.

Motivated by these gaps, we study transfer learning under a moderate-dimensional robust linear model with one target study and one related source study.
We allow the predictor dimension to be of the same order as the target and source sample sizes, do not impose sparsity assumptions on the regression coefficients, and permit heavy-tailed errors.
Within this framework, our goal is to leverage information from the source study to improve the estimation performance of traditional single-task approaches.

Our main contribution is to develop and analyze a transfer learning method for moderate-dimensional ridge-regularized robust linear regression.
First, we propose Trans-RR, a transfer learning procedure for ridge-regularized robust linear regression.
It combines a robust ridge fit on the source study with a robust ridge correction on the target study and is designed for non-sparse coefficients.
Second, we derive an asymptotic characterization of the $\ell_2$ risk of the resulting estimator under mild assumptions on the design and error distributions.
The theory shows how auxiliary source data can improve estimation accuracy relative to the single-study ridge-regularized robust estimator, while also clarifying the possibility of negative transfer.
Third, we conduct simulation studies and a real-data analysis to examine the practical performance of the proposed method.

The rest of the paper is organized as follows.
Section \ref{sec_method} introduces the model setup and the proposed algorithm.
Section \ref{sec_theory} presents the technical assumptions and theoretical results.
Section \ref{sec_simulation} presents simulation studies to evaluate the performance of the proposed method.
Section \ref{sec_realdata} applies the proposed method to a real-data example.
The proofs of the main theoretical result as well as the lemmas are included in the Supplementary Material.

\textbf{Notation.}
Denote by $\I_m$ the $m \times m$ identity matrix.
Let $\bzero_m \in \RR^m$ and $\bone_m \in \RR^m$ be the vectors of zeros and ones, respectively.
For a vector $\bfv=(v_1, \ldots, v_m)^\T$, the $\ell_2$ norms are $\|\bfv\|=(\sum_{i=1}^m v_i^2)^{1 / 2}$, whereas $\|v\|_{\infty}=\max _{1 \leq k \leq p}|v_k|$.
For an $m \times m$ matrix $\bfA=\{a_{i j}\}_{1 \leq i, j \leq m}$, denote by $\lambda_{\max }(\bfA)$ and $\lambda_{\min }(\bfA)$ the maximum and minimum eigenvalues of $\bfA$, respectively.
The $L_2$ norm of $\bfA$ is defined as $\|\bfA\|=\{\lambda_{\max }(\bfA^{\T} \bfA)\}^{1 / 2}$.

\section{Methodology} \label{sec_method}

\subsection{Problem setup}

We consider a transfer learning problem with one target study and one related source study.
In the target study, we observe $n$ samples $\boldsymbol{x}_i \in \mathbb{R}^p$ and $y_i \in \RR$, $i=1, \ldots, n$, generated from
\begin{equation}
    \label{eq:target_model}
    y_i=\bfx_i^{\T} \bbeta_0 + \epsilon_i,
\end{equation}
where $\epsilon_i, i=1, \ldots, n$, are independently distributed errors and $\bbeta_0 \in \RR^p$ is the unknown regression parameter of interest.

In addition to the target data, we observe $n_1$ samples $(\bfx_i^{(1)}, y_i^{(1)})$, $i = 1,\ldots, n_1$, from the source study satisfying
\begin{equation}
    \label{eq:source_model}
    y_i^{(1)}=(\bfx_i^{(1)})^{\T} \bfw_0 + \epsilon_i^{(1)},
\end{equation}
where $\bfw_0 \in \mathbb{R}^p$ is the regression parameter for the source study and $\epsilon_i^{(1)}, i=1, \ldots, n_1$, are independently distributed errors.
Throughout, both $\epsilon_i$ and $\epsilon_i^{(1)}$ are allowed to be heavy-tailed.

We work in the moderate-dimensional regime, where $p$ is of the same order as both $n$ and $n_1$, with $p/n \to \kappa \in (0,\infty)$ and $p/n_1 \to \kappa_1 \in (0,\infty)$.
We also do not impose sparsity assumptions on $\bbeta_0$ or $\bfw_0$.
Let $\bdelta_0 = \bbeta_0 - \bfw_0$ denote the source--target discrepancy and let $h = \|\bdelta_0\|$ measure its size.
Smaller values of $h$ correspond to a more informative source study and hence a greater potential for useful transfer.

\subsection{Trans-RR algorithm}

Based on this setup, we now introduce the proposed transfer learning algorithm, referred to as Trans-RR.
Following the general two-stage strategy used in \cite{bastani2021predicting}, \cite{li2022transfer}, and \cite{Tian2023Transfer}, the core of the procedure consists of two estimation steps.
The first step estimates the source coefficient vector $\bfw_0$ from the source data.
The second step estimates the source--target discrepancy $\bdelta_0 = \bbeta_0 - \bfw_0$ from the target data, after which the two estimates are combined.
Algorithm \ref{algo1} summarizes the procedure.

\vspace{0.1in}\begin{algorithm}
%\NoCaptionOfAlgo
 \SetKwInOut{Input}{Input}
    \SetKwInOut{Output}{Output}
\SetAlgoLined
 \Input{target data $\{ (\bfx_i, y_i) \}_{i=1}^{n}$ and source data $\{ (\bfx_i^{(1)}, y_i^{(1)}) \}_{i=1}^{n_1}$}
 \Output{the estimated coefficient vector $\hat{\bbeta}$}

\underline{Step 1}. Compute
\begin{equation}
\label{eq:step_1}
\hat \bfw = \argmin_{\bfw \in \RR^p} \Big[ \frac{1}{n_1} \sum_{i=1}^{n_1} \tilde \rho \big\{ y^{(1)}_i-(\bfx^{(1)}_i)^{\T} \bfw \big\} + \frac{\tau_1}{2} \|\bfw\|^2 \Big]
\end{equation}
for some constant $\tau_1$.

\underline{Step 2}. Compute
\begin{equation}
    \label{eq:step_2}
    \hat \bdelta = \argmin_{\bdelta \in \mathbb{R}^p} \Big[\frac{1}{n} \sum_{i=1}^{n} \rho \big\{y_i-\bfx_i^{\T} (\hat \bfw + \bdelta) \big\} + \frac{\tau}{2} \|\bdelta\|^2 \Big]
\end{equation}
for some constant $\tau$.

\underline{Step 3}. Let 
\begin{equation}
    \label{eq:opt_ori}
    \hat{\bbeta}=\hat{\bfw}+\hat{\bdelta}.
\end{equation}

\underline{Step 4}. Output $\hat{\bbeta}$.

 \caption{\textbf{Trans-RR algorithm}}
 \label{algo1}
\end{algorithm}

The idea behind the construction is straightforward.
Step 1 computes a ridge-regularized robust estimator from the source study.
Step 2 then estimates the discrepancy relative to the source-stage fit by solving a second ridge-regularized robust regression problem on the target study.
The final estimator is obtained by combining these two pieces, namely $\hat \bbeta=\hat \bfw+\hat \bdelta$.
The main difference between our procedure and those of \cite{bastani2021predicting,li2022transfer,Tian2023Transfer} is that we use ridge/$\ell_2$ regularization in both steps, whereas they use $\ell_1$ regularization.
This choice is motivated by our diffuse-coefficient setting: the regression parameters $\bbeta_0$ and $\bfw_0$ have many small coordinates and are not well approximated by sparse vectors.
In this setting, lasso-based methods are not well suited to the problem, whereas ridge penalization is natural.
Another difference is that we use robust loss functions rather than the quadratic loss, which makes the procedure less sensitive to outliers and heavy-tailed errors.

\begin{remark}
When robustness to heavy-tailed errors or outliers is needed, Huber-type loss functions are natural candidates for $\rho$ and $\tilde \rho$.
Specific choices of $\rho$ and $\tilde \rho$ under our theoretical framework are discussed in Section \ref{sec_theory}.
The regularization parameters $\tau$ and $\tau_1$ may be selected by standard data-driven tuning methods such as cross-validation.
\end{remark}

% In step 1, $\hat \bfw$ is studied in \cite{el2018impact}.
% We will elaborate it in the next section.

\section{Theoretical Results} \label{sec_theory}

This section introduces the assumptions for the analysis and then presents the main asymptotic error results for Trans-RR.

\subsection{Technical assumptions}

We study the estimation error of the estimator in Algorithm \ref{algo1} under the following assumptions.
We state the assumptions separately for the target study (Assumption \ref{ass:target}) and the source study (Assumption \ref{ass:source}), since the two stages are based on different samples.
The two sets of conditions are largely parallel.

\begin{assumptionp}{A} \label{ass:target}
    \
      \begin{itemize}
        \item \textbf{A1:} $p/n \to \kappa \in (0,\infty)$.
        \item \textbf{A2:}
        Suppose $\rho$ is an even and convex function.
        Assume that $\psi = \rho'$ is bounded and $\psi'$ is Lipschitz and bounded.
        Moreover, we assume that $\sign(\psi(x)) = \sign(x)$ and that $\rho(x) \geq \rho(0) = 0$ for all $x \in \RR$.
        \item \textbf{A3:}
        Assume that there exist independent variables $\lambda_i$'s and $\cX_i$'s such that $\bfx_i=\lambda_i \cX_i$.
        Suppose that $\cX_i$'s are i.i.d. with independent entries, and they have mean $\bzero_p$ and $\cov(\cX_i)=\I_p$.
        Suppose there exist $\mc_n$ and $C_n$ that vary with $n$, where $1/\mc_n = \mO(\polyLog(n))$ and $C_n$ is bounded in $n$, such that for any convex $1$-Lipschitz function $G$ of $\cX_i$, $P(|G(\cX_i) - m_G| > t) \leq C_n \exp(-\mc_n t^2)$ holds for all $t > 0$, where $m_G$ is the median of $G(\cX_i)$.
        We require the same assumption to hold for the columns of the $n \times p$ design matrix $\cX$.
        Additionally, we assume that the coordinates of $\cX_i$ have moments of all orders, and the $k$-th moment of the entries of $\cX_i$ is assumed to be uniformly bounded independently of $n$ and $p$ for all $k$.

        \item \textbf{A4:}
        Suppose that $\lambda_i$'s are independent, with $\E(\lambda_i^2) = 1$, $\E(\lambda_i^4)$ being bounded, and $\sup_{1 \leq i \leq n} |\lambda_i|$ growing at most like $C_\lambda(\log n)^k$ for some $k$.
        $\lambda_i$'s may have finitely many possible distributions.
        \item \textbf{A5:}
        Suppose that $\epsilon_i$'s are independent and are also independent of $\cX_i$'s and $\lambda_i$'s.
        They may have finitely many possible distributions, each with a density that is differentiable, symmetric, and unimodal.
        If $f_i$ is the density of one such distribution, we assume that $\lim _{x \to \infty} x f_i(x) = 0$.
        \item \textbf{A6:}
        The fraction of occurrences for each possible combination of distributions for $(\epsilon_i, \lambda_i)$ has a limit as $n \rightarrow \infty$.
        \item \textbf{A7:}
        There exist constants $C_{\bbeta}$ and $\me > 1/3$ such that $\|\bbeta_0\| \leq C_{\bbeta}$ and $\|\bbeta_0\|_{\infty}=\mO(n^{-\me})$.
      \end{itemize}
\end{assumptionp}

\begin{assumptionp}{B}
    \label{ass:source}
    \
      \begin{itemize}
        \item \textbf{B1:}
        $p/n_1 \to \kappa_1 \in (0,\infty)$.
        \item \textbf{B2:}
        $\tilde \rho$ and $\tilde \psi$ satisfy Assumption \ref{ass:target}2.
        \item \textbf{B3:}
        $\bfx_i^{(1)}$, $\cX_i^{(1)}$'s and $\lambda_i^{(1)}$'s satisfy Assumption \ref{ass:target}3.
        \item \textbf{B4:}
        $\lambda_i^{(1)}$'s satisfy Assumption \ref{ass:target}4.
        \item \textbf{B5:}
        $\epsilon_i^{(1)}$'s, $\cX_i^{(1)}$'s and $\lambda_i^{(1)}$'s satisfy Assumption \ref{ass:target}5.
        \item \textbf{B6:}
        $\lambda_i^{(1)}$'s and $\epsilon_i^{(1)}$'s satisfy Assumption \ref{ass:target}6.
        \item \textbf{B7:}
        $\|\bfw_0\|_2$ remains bounded.
        Furthermore, $\|\bfw_0\|_{\infty}=\mO(n_1^{-\me})$, where $\me > 1 / 3$.
      \end{itemize}
\end{assumptionp}

For Assumptions \ref{ass:target}2 and \ref{ass:source}2, it is quite common in robust statistics to require $\psi$ to be bounded.
For example, the Huber loss
\begin{equation*}
    \label{eq:rho_Huber}
    \rho_H(x)= \begin{dcases}
        \frac{x^2}{2} & \text { if }  |x| \leq \delta, \\
        \delta \cdot \big(|x| - \frac{1}{2} \delta \big)  & \text { otherwise} .
    \end{dcases}
\end{equation*}
is chosen to grow linearly to infinity, which reduces the influence of outliers on the resulting regression estimator.
Although the Huber loss does not fully satisfy the assumptions because it is not differentiable at $|x| = \delta$, these assumptions hold for a smoothed approximation such as
\begin{equation}
    \label{eq:rho_eta}
    \rho_\eta(x)= \begin{dcases}
        \frac{x^2}{2} & \text { if }  |x| \leq \delta-\eta, \\
        \Big(\delta - \frac{\eta}{2}\Big) \cdot |x| + \frac{(\delta - |x|)^3}{6\eta} + C_\rho & \text { if } |x| \in(\delta-\eta, \delta), \\
        \Big(\delta - \frac{\eta}{2}\Big) \cdot |x| + C_\rho & \text { if } |x| \geq \delta,
    \end{dcases}
\end{equation}
where $C_\rho = - \eta^2/6 + \eta \delta / 2 - \delta^2/2$.
The corresponding $\psi_\eta$ is given by
\begin{equation}
    \psi_\eta(x)= \begin{dcases}
        x & \text { if }  |x| \leq \delta-\eta, \\
        \sign(x) \cdot \Big\{\delta - \frac{\eta}{2} - \frac{(\delta-|x|)^2}{2 \eta} \Big\} & \text { if } |x| \in(\delta-\eta, \delta), \\
        \sign(x) \cdot \Big(\delta-\frac{\eta}{2} \Big) & \text { if } |x| \geq \delta.
    \end{dcases}
\end{equation}
Another example that fully satisfies the assumptions is the Pseudo-Huber loss function \citep{Charbonnier1997,hartley2003multiple}, defined by
\begin{equation}
    \label{eq:rho_pseudo}
    \rho_{P}(x)=\delta^2 \Big(\sqrt{1+\frac{x^2}{\delta^2}}-1 \Big).
\end{equation}

% For Assumptions \ref{ass:target}3, \ref{ass:target}4, \ref{ass:target}5, \ref{ass:source}3, \ref{ass:source}4 and \ref{ass:source}5,
The assumptions on $\bfx_i$'s and $\bfx_i^{(1)}$'s, in particular that they have mean $\bzero_p$ and covariance matrix $\I_p$, are common in the study of M-estimators for linear models.
These assumptions have been used in the low-dimensional regime $p/n \to 0$ studied in \cite{huber1973robust,yohai1979asymptotic,portnoy1984asymptotic}, in the moderate-dimensional regime $p/n \to \kappa \in (0,1)$ considered in \cite{el2013robust}, and in the regime $p/n \to \kappa > 0$ analyzed in \cite{karoui2013asymptotic,el2018impact}.

The concentration assumption on $\cX_i$'s and $\cX_i^{(1)}$'s is weaker than the Gaussian assumptions often imposed in robust statistics.
This assumption has also been studied in \cite{Karoui2009Concentration,el2018impact} and holds for a broad class of distributions.
Corollary 4.10 in \cite{ledoux2001concentration} demonstrates that our assumptions are satisfied if $\cX_i$ has independent entries bounded by $1 / (2 \sqrt{c_1})$ for some $c_1 > 0$.
Additionally, Theorem 2.7 in \cite{ledoux2001concentration} shows that the assumptions hold when $\cX_i$ has independent entries with density $f_k$, $1 \leq k \leq p$, such that $f_k(x) = \exp(-u_k(x))$ and $u_k''(x) \geq \sqrt{c_2}$ for some $c_2 > 0$.
In particular, this condition holds when $\cX_i$ has i.i.d. $N(0,1)$ entries, where $c_2 = 1$.
As will be seen in the proof, the functions $G$ that arise in our analysis are either linear functions or square roots of quadratic forms.
A similar discussion applies to the $\cX_i^{(1)}$'s.

The introduction of $\lambda_i$ and $\lambda_i^{(1)}$, as also considered in \cite{el2013robust,el2018impact}, is used to induce a nonspherical geometry on the predictors.
Although the assumption $\E(\lambda_i^2) = 1$ can be relaxed to the requirement that $\E(\lambda_i^2)$ be uniformly bounded, it remains statistically important because it guarantees that $\cov(\bfx_i) = \I_p$ in all the models we consider.
This construction shows that many models can share the same covariance $\cov(\bfx_i)$ while exhibiting substantially different estimation errors for $\hat\bbeta$.
This contrasts with the low-dimensional setting studied in \cite{huber2011robust}, where $\cov(\bfx_i)$ is the key quantity for robust regression.
A similar discussion applies to $\lambda_i^{(1)}$'s and $\bfx_i^{(1)}$'s.

In Assumptions \ref{ass:target}5 and \ref{ass:source}5, no moment restriction is imposed on the $\epsilon_i$'s and $\epsilon_i^{(1)}$'s.
For instance, smooth symmetric log-concave densities satisfy all of these assumptions; see \cite{karlin1968total,ibragimov1956composition}.
Furthermore, the Cauchy distribution also satisfies these conditions; see Theorem 1.6 in \cite{dharmadhikari1988unimodality}.
This makes the framework compatible with heavy-tailed errors, which are of particular interest in robust regression.

Assumptions \ref{ass:target}7 and \ref{ass:source}7 impose a non-sparse structure on $\bbeta_0$ and $\bfw_0$, meaning that these vectors cannot be well approximated by sparse vectors in $\ell_2$ norm.
This setting is common in moderate-dimensional statistics and contrasts with the sparse regime, where only a small fraction of coefficients are substantial.
Under these assumptions, the proposed Trans-RR estimator may outperform lasso-based methods.

We now turn to the target-stage result and the resulting error characterization for Trans-RR.

\subsection{Asymptotic characterization of estimation error}

Our main theorem characterizes the asymptotic $\ell_2$ error of the Trans-RR estimator.
Recall that $\hat \bbeta$ is defined in \eqref{eq:opt_ori}, and let $\tau > 0$ be fixed as $n$ and $p$ vary.
To state the result, let $\prox(c\rho)$ denote the proximal mapping of the function $c \rho$; see \cite{moreau1965proximite}.
It is given by
$$
\prox(c\rho)(x)=\argmin_{y \in \RR}(c \rho(y)+\frac{1}{2}(x-y)^2).
$$
Under Assumptions \ref{ass:target} and \ref{ass:source}, the estimation error admits the following limit.

\begin{theorem}
    \label{thm:step_2}
    Under Assumption \ref{ass:target} and Assumption \ref{ass:source}, conditional on the source-stage estimator $\hat \bfw$, which is independent of the target sample, we have $\|\hat \bbeta - \bbeta_0\| \to r_\rho(\kappa)$ in probability, where $r_\rho(\kappa)$ is deterministic for the given value of $\hat \bfw$.
    Let $W_i = \epsilon_i + r_\rho(\kappa) \lambda_i Z_i$, where $Z_i$ is a standard normal random variable independent of $\epsilon_i$ and $\lambda_i$.
    Then there exists a constant $c_\rho(\kappa)$ such that
    \begin{equation}
        \label{eq:thm_step_2}
        \begin{dcases}
            \lim _{n \rightarrow \infty} \frac{1}{n} \sum_{i=1}^n \E \big\{[\prox\{c_\rho(\kappa) \lambda_i^2 \rho\}]' (W_i) \big\} &= 1-\kappa+\tau c_\rho(\kappa), \\
            \kappa \Big[ \lim _{n \rightarrow \infty} \frac{1}{n} \sum_{i=1}^n \E \Big(\frac{[W_i-\prox\{c_\rho(\kappa) \lambda_i^2 \rho\}(W_i)]^2}{\lambda_i^2} \Big) \Big] +\tau^2\|\bbeta_0 - \hat \bfw\|^2 c_\rho^2(\kappa) &=\kappa^2 r_\rho^2(\kappa) .\end{dcases}
    \end{equation}
\end{theorem}

The proof of Theorem \ref{thm:step_2} is given in the Supplementary Material.
Here and below, the dependence of $r_\rho(\kappa)$ and $c_\rho(\kappa)$ on $\hat \bfw$ is suppressed for notational simplicity.
The limits on the left-hand side of \eqref{eq:thm_step_2} exist because Assumption \ref{ass:target}6 guarantees convergence of the proportions associated with each pair $(\mathcal{L}(\epsilon_i), \mathcal{L}(\lambda_i))$, where $\mathcal{L}(\epsilon_i)$ and $\mathcal{L}(\lambda_i)$ denote the laws of $\epsilon_i$ and $\lambda_i$.
For the second equation in \eqref{eq:thm_step_2}, the ratio can be interpreted through the identity
$$
\frac{[x-\prox\{c_\rho(\kappa) \lambda^2 \rho\}(x)]^2}{\lambda^2}=c_\rho^2(\kappa) \lambda^2 \psi^2(\prox\{c_\rho(\kappa) \lambda^2 \rho\}(x))
$$
which is well defined when $\lambda = 0$.
Equivalently, \eqref{eq:thm_step_2} can be written as
\begin{equation*}
    \begin{dcases}
        \lim _{n \rightarrow \infty} \frac{1}{n} \sum_{i=1}^n \E \big\{[\prox\{c_\rho(\kappa) \lambda_i^2 \rho\}]' (W_i) \big\} &= 1-\kappa+\tau c_\rho(\kappa), \\
        \kappa \Big[ \lim _{n \rightarrow \infty} \frac{1}{n} \sum_{i=1}^n \E \big\{ c_\rho^2(\kappa) \lambda_i^2 \psi^2(\prox\{c_\rho(\kappa) \lambda_i^2 \rho\}(W_i)) \big\} \Big] +\tau^2\|\bbeta_0 - \hat \bfw\|^2 c_\rho^2(\kappa) &=\kappa^2 r_\rho^2(\kappa) .\end{dcases}
\end{equation*}

This representation shows that the expectation in \eqref{eq:thm_step_2} is well defined, and Assumption \ref{ass:target}6 ensures that the relevant limits exist.
Unlike several recent transfer learning analyses, such as \cite{bastani2021predicting,li2022transfer,Tian2023Transfer}, our theory does not impose direct structural restrictions on $\bdelta_0$, the difference between the target and source coefficients.
Theorem \ref{thm:step_2} also shows that the performance of $\hat \bbeta$ depends on the distribution of the $\lambda_i$'s in the representation $\bfx_i = \lambda_i \cX_i$ from Assumption \ref{ass:target}3.
Thus, in the moderate-dimensional regime, the geometry of the target predictors encoded by $\lambda_i$ materially affects the estimation error.
This again contrasts with low-dimensional robust regression, where $\cov(\bfx_i)$ is the dominant quantity.

When $\lambda_i^2 = 1$ for all $i$ and the errors $\epsilon_i$ are i.i.d., Theorem \ref{thm:step_2} simplifies as follows.

\begin{corollary}
    \label{thm:step_2_s}
    Under the same assumptions as in Theorem \ref{thm:step_2}, if $\lambda_i^2 = 1$ for all $i$ and the errors $\epsilon_i$ are i.i.d., then, conditional on $\hat \bfw$, we have $\|\hat \bbeta - \bbeta_0\| \to r_\rho(\kappa)$ in probability, where $r_\rho(\kappa)$ is deterministic for the given value of $\hat \bfw$.
    Let $w = \epsilon + r_\rho(\kappa) z$, where $\epsilon$ has the same distribution as the $\epsilon_i$'s and $z$ is a standard normal random variable independent of $\epsilon$.
    Then there exists a constant $c_\rho(\kappa)$ such that
    \begin{equation*}
        \begin{dcases}
        \E \big\{[\prox\{c_\rho(\kappa) \rho\}]' (w) \big\} &= 1-\kappa+\tau c_\rho(\kappa), \\
        \kappa \E \big([w-\prox\{c_\rho(\kappa) \rho\}(w)]^2 \big) +\tau^2\|\bbeta_0 - \hat \bfw\|^2 c_\rho^2(\kappa) &=\kappa^2 r_\rho^2(\kappa) .\end{dcases}
    \end{equation*}
\end{corollary}

Corollary \ref{thm:step_2_s} shows that, under a homogeneous target design, the general characterization in Theorem \ref{thm:step_2} reduces to a simpler scalar system.
This special case is useful for interpretation and will also serve as a convenient benchmark in the simulation study.

To discuss when transfer can be beneficial, we also recall the asymptotic characterization of the source-stage estimator from \cite{el2018impact}.

\begin{corollary}[{{\citet[Theorem 2.1]{el2018impact}}}]
    \label{thm:step_1}
    Under Assumption \ref{ass:source}, we have $\|\hat \bfw - \bfw_0\| \to r_{\tilde \rho}(\kappa_1)$ in probability, where $r_{\tilde \rho}(\kappa_1)$ is deterministic.
    Let $\tilde W_i = \epsilon_i^{(1)} + r_{\tilde \rho}(\kappa_1) \lambda_i^{(1)} \tilde Z_i$, where $\tilde Z_i$ is a standard normal random variable independent of $\epsilon_i^{(1)}$ and $\lambda_i^{(1)}$.
    Then there exists a constant $c_{\tilde \rho}(\kappa_1)$ such that
    \begin{equation*}
        \begin{dcases}
            \lim _{n_1 \to \infty} \frac{1}{n_1} \sum_{i=1}^{n_1} \E \big\{[\prox\{c_{\tilde \rho}(\kappa_1) (\lambda_i^{(1)})^2 \tilde \rho\}]' (\tilde W_i) \big\}  &= 1-\kappa_1+\tau_1 c_{\tilde \rho}(\kappa_1), \\
            \kappa_1 \Big[\lim_{n_1 \to \infty} \frac{1}{n_1} \sum_{i=1}^{n_1} \E \Big(\frac{[\tilde W_i-\prox\{c_{\tilde \rho}(\kappa_1) (\lambda_i^{(1)})^2 \tilde \rho\}(\tilde W_i)]^2}{(\lambda_i^{(1)})^2} \Big) \Big]+\tau_1^2\|\bfw_0 \|^2 c_{\tilde \rho}^2(\kappa_1)  &=\kappa_1^2 r_{\tilde \rho}^2(\kappa_1) .\end{dcases}
    \end{equation*}
\end{corollary}

Corollary \ref{thm:step_1} shows that the source-stage estimator admits a deterministic asymptotic error characterization.

A central question in transfer learning is how the source study affects the comparison between Trans-RR and the single-task ridge-regularized estimator
\begin{equation}
    \label{eq:single-task}
    \hat \bbeta_{st} = \argmin_{\bbeta \in \mathbb{R}^p} \Big[\frac{1}{n} \sum_{i=1}^{n} \rho (y_i-\bfx_i^{\T} \bbeta ) + \frac{\tau}{2} \|\bbeta\|^2 \Big].
\end{equation}
Theorem \ref{thm:step_2} shows that, conditional on $\hat \bfw$, this comparison is governed by the quantity $\|\bbeta_0-\hat \bfw\|$.
To better understand this dependence, Section \ref{sec:sim_theoretical_curves} numerically solves the scalar system in Corollary \ref{thm:step_2_s} under the smoothed Huber loss \eqref{eq:rho_eta} and plots $r_\rho$ as a function of $\|\bbeta_0-\hat \bfw\|$.
In the three cases considered there, the resulting curves are increasing over the displayed range.
Combined with Corollary \ref{thm:step_1}, this numerical evidence suggests that transfer is more favorable when the source coefficient $\bfw_0$ is close to the target coefficient $\bbeta_0$ and the source-stage estimation error is small, so that $\hat \bfw$ provides a more accurate approximation to $\bbeta_0$.
By contrast, if either the source--target discrepancy or the source-stage estimation error is large, the advantage of transfer may diminish or disappear.

\section{Simulation} \label{sec_simulation}

In this section, we conduct numerical studies to support the theoretical results.
We set the dimension of both target and source data to be $p \in \{200, 400, 800\}$.
We set $n = p, p/4$ and $n_1 = 2p, p/2$, corresponding to moderate-dimensional settings with $\kappa = 1, 4$ and $\kappa_1 = 1/2, 2$.
To generate data, we set $\bfx_i=\lambda_i \cX_i$ and $\bfx_i^{(1)}=\lambda_i^{(1)} \cX_i^{(1)}$, where $\cX_i$ and $\cX_i^{(1)}$ have i.i.d. $N(0,1)$ entries.
We consider three different cases for the choices of $\lambda_i$'s, $\epsilon_i$'s, $\lambda_i^{(1)}$'s and $\epsilon_i^{(1)}$'s:
\begin{itemize}
    \item \textbf{Case $\I$:} $\lambda_i = 1$ for $i=1,\ldots,n$ and $\lambda_j^{(1)} = 1$ for $j=1,\ldots,n_1$. The target errors $\epsilon_i$ are i.i.d. $N(0,1)$, and the source errors $\epsilon_j^{(1)}$ are i.i.d. $N(0,2^2)$.
    \item \textbf{Case $\II$:} The variables $\lambda_i$ and $\lambda_j^{(1)}$ are i.i.d. Unif$(0,\sqrt{3})$, while $\epsilon_i$ and $\epsilon_j^{(1)}$ are i.i.d. $Cauchy(0,1)$ and $Cauchy(0,2)$, respectively.
    \item \textbf{Case $\III$:} In both the target and source studies, half of the observations are generated as in Case $\I$ and the other half are generated as in Case $\II$.
\end{itemize}
Case $\I$ is a standard Gaussian setup for linear regression.
Case $\II$ features a non-Gaussian design and heavy-tailed errors.
Case $\III$ is a mixture of the two cases and is designed to test the effectiveness of our theoretical results under non-identical $\bfx_i$'s and $\epsilon_i$'s.

\subsection{Validity of theoretical results}
We first evaluate the validity of the proposed scalar $r_\rho$ in Theorem \ref{thm:step_2}.
For each setting, we generate $\bbeta^*$ and $\bfw^*$ once with i.i.d. Unif$(0,1)$ entries and set $\bbeta_0=\bbeta^*/\sqrt{n}$ and $\bfw_0=\bfw^*/\sqrt{n}$.
This construction yields diffuse coefficients whose Euclidean norms remain bounded as $n$ grows.
These coefficient vectors are fixed across the $1000$ replications, while the target and source samples are regenerated in each replicate.
In each replicate, we first compute $\hat \bfw$ from the source sample and then obtain $\hat \bbeta$ by applying Algorithm \ref{algo1}, using the smoothed Huber loss \eqref{eq:rho_eta} for both $\tilde \rho$ and $\rho$ with parameters $\delta = 1.35$ and $\eta = 0.1$.
We fix $\tau = \tau_1 = 1$ and repeat each setup $1000$ times.

Figure \ref{fig1} presents boxplots of the estimation error $\| \hat\bbeta - \bbeta_0 \|^2$ for cases $\I$--$\III$ and $\kappa = 1, 4$.
The red point in each boxplot marks the theoretical value $r_\rho^2$, obtained by numerically solving the system in Theorem \ref{thm:step_2} under the corresponding simulation specification.
We observe that the empirical distribution of $\| \hat\bbeta - \bbeta_0\|^2$ is centered close to this value, and its dispersion decreases as $n$ and $p$ become larger.
Table \ref{table1} shows the mean and standard deviation (SD) of $\| \hat\bbeta - \bbeta_0\|^2$ (denoted as $\hat r^2$) and the corresponding $r_\rho^2$ for each setup.
As dimensionality increases, that is, as $\kappa$ increases from $1$ to $4$, both the mean error and its variability increase, indicating that estimation becomes more difficult in more challenging moderate-dimensional regimes.
The average estimation error also grows with heavier-tailed errors, highlighting the difficulty of estimation under such conditions.
Results under case $\III$ demonstrate that Theorem \ref{thm:step_2} is effective in handling non-identical $\bfx_i$'s and $\epsilon_i$'s.
Overall, the findings in Figure \ref{fig1} and Table \ref{table1} align well with the theoretical predictions of Theorem \ref{thm:step_2}.

\begin{figure}[H]
    \centering
    \includegraphics[width=\textwidth]{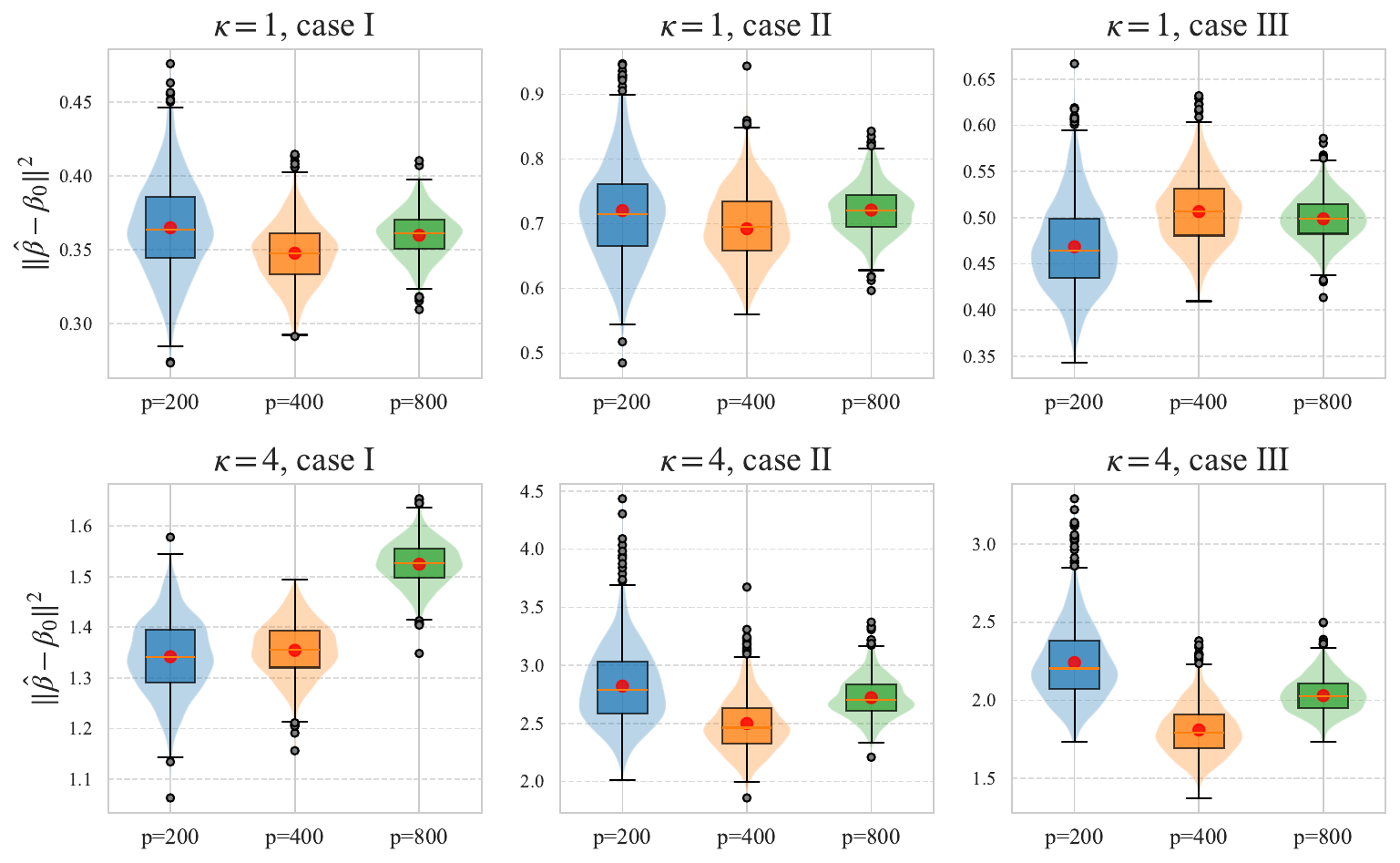}
    \caption{Boxplot of $\| \hat\bbeta - \bbeta_0\|^2$ over $1000$ simulations.
    The red point in each boxplot represents the theoretical value $r_\rho^2$ from Theorem \ref{thm:step_2}.
    Panels from top to bottom are for $\kappa = 1, 4$, respectively, while panels from left to right are for cases $\I, \II, \III$,  respectively.}
    \label{fig1}
\end{figure}

\begin{table}[H] % Use [htbp] for suggested float placement
    \centering
    % Adjust \tabcolsep as needed. The template's 20pt might be too wide for 7 columns.
    % 6pt is default, 8-10pt might be a good start here.
    \setlength\tabcolsep{7pt} 
    \caption{Mean and SD (in parentheses) of $\| \hat{\bm{\beta}} - \bm{\beta}_0\|^2$ (denoted as $\hat r^2$) and the corresponding $r_\rho^2$ over $1000$ simulations. Rows are indexed by $(p,n)$, with $n_1=2p$ when $\kappa=1$ and $n_1=p/2$ when $\kappa=4$.}
    \begin{tabular}{@{}lcccccc@{}} % Use @{} to remove extra horizontal space at the ends of the table
                                   % l for the first column, c for the others
        \toprule
        \multirow{2}{*}{\textbf{($p, n$)}} & \multicolumn{2}{c}{\textbf{Case $\mathrm{I}$}} & \multicolumn{2}{c}{\textbf{Case $\mathrm{II}$}} & \multicolumn{2}{c}{\textbf{Case $\mathrm{III}$}} \\
        \cmidrule(lr){2-3} \cmidrule(lr){4-5} \cmidrule(lr){6-7} % Mid-rules for grouped columns
        & \textbf{$\hat r^2$} & \textbf{$r^2_\rho$} & \textbf{$\hat r^2$} & \textbf{$r^2_\rho$} & \textbf{$\hat r^2$} & \textbf{$r^2_\rho$} \\
        \midrule
        \multicolumn{7}{@{}c@{}}{\textbf{$\kappa = 1$}} \\ % Centered section header for kappa = 1
        \midrule % Rule to separate kappa header from its data
        (200, 200) & 0.3653(0.0318) & 0.3649 & 0.7163(0.0738) & 0.7204 & 0.4683(0.0478) & 0.4685 \\
        (400, 400) & 0.3472(0.0208) & 0.3477 & 0.6970(0.0549) & 0.6923 & 0.5076(0.0368) & 0.5065 \\
        (800, 800) & 0.3603(0.0151) & 0.3598 & 0.7206(0.0374) & 0.7212 & 0.4989(0.0238) & 0.4986 \\
        \midrule
        \multicolumn{7}{@{}c@{}}{\textbf{$\kappa = 4$}} \\ % Centered section header for kappa = 4
        \midrule % Rule to separate kappa header from its data
        (200, 50)  & 1.3415(0.0734) & 1.3419 & 2.8222(0.3311) & 2.8216 & 2.2410(0.2407) & 2.2415 \\
        (400, 100) & 1.3565(0.0531) & 1.3544 & 2.4896(0.2363) & 2.4996 & 1.8087(0.1590) & 1.8102 \\
        (800, 200) & 1.5261(0.0427) & 1.5247 & 2.7226(0.1693) & 2.7219 & 2.0342(0.1153) & 2.0301 \\
        \bottomrule
    \end{tabular}
    \label{table1} % Changed the label to avoid potential conflicts
\end{table}

\subsection{Theoretical estimation error curves}\label{sec:sim_theoretical_curves}

To further illustrate the theoretical characterization in Corollary \ref{thm:step_2_s}, we numerically solve the associated scalar system while varying the discrepancy term $\|\bbeta_0-\hat \bfw\|$.
We again use the smoothed Huber loss \eqref{eq:rho_eta} with $(\delta,\eta)=(1.35,0.1)$ and fix $\kappa=1$.
Figure \ref{fig:theoretical_curves} displays the resulting curves of $r_\rho$ for five values of $\tau$ under cases $\I$--$\III$.
In all three cases, the numerical curves are increasing over the displayed range, indicating that larger values of $\|\bbeta_0-\hat \bfw\|$ are associated with larger asymptotic estimation errors in these examples.
This experiment complements Figure \ref{fig1} by showing how the theoretical prediction varies with the discrepancy term that enters the second equation of Theorem \ref{thm:step_2}.

\begin{figure}[H]
    \centering
    \includegraphics[width=\textwidth]{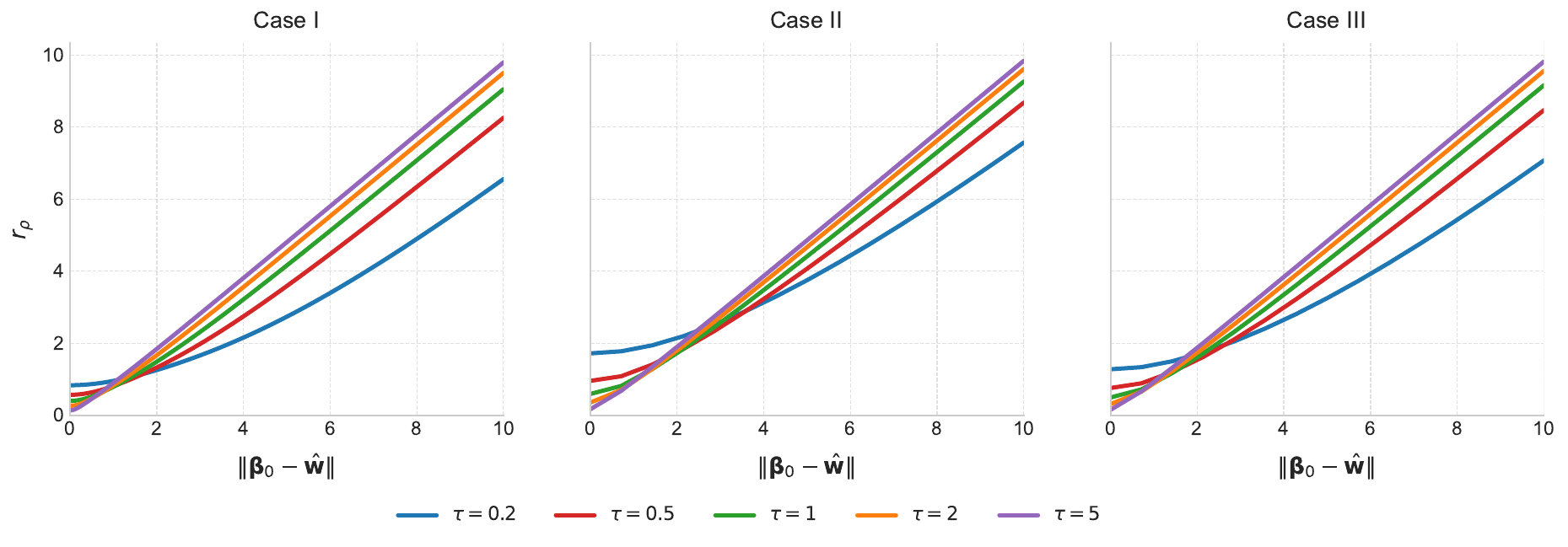}
    \caption{Theoretical curves of $r_\rho$ as a function of $\|\bbeta_0-\hat \bfw\|$ for five values of $\tau$ under cases $\I$--$\III$, obtained by numerically solving Corollary \ref{thm:step_2_s}. The three panels correspond to cases $\I$, $\II$, and $\III$, respectively.}
    \label{fig:theoretical_curves}
\end{figure}

\subsection{Comparison with existing methods}

To evaluate when transfer is beneficial, we compare our method with several competing procedures across the three scenarios described above.
We set $p = 400$, $n = p$ and $n_1 = 2p$.
We generate $\bbeta_0=\bbeta^*/\|\bbeta^*\|$, where $\bbeta^*=(\beta_1^*, \ldots, \beta_{p}^*)^{\top}$ has i.i.d. Unif$(0,1)$ entries.
To control the transfer strength $h = \|\bdelta_0\|$, we set
\[
\bdelta_0 = \exp(c_d) \cdot \mathbf{1}_p / \sqrt{p}, \qquad c_d \in \{-2.0,-1.5,-1.0,-0.5,0,0.5,1.0\},
\]
and define the source coefficient by $\bfw_0 = \bbeta_0 - \bdelta_0$.
By varying $c_d$ from $-2.0$ to $1.0$, we obtain $h$ ranging from approximately $0.135$ to $2.718$, providing a comprehensive evaluation across different levels of source-target similarity.
For each value of $c_d$, the pair $(\bbeta_0,\bfw_0)$ is fixed across replications and only the data are regenerated.

The methods under comparison are Single-RR, Trans-RR, Pooled-RR, Single-Lasso, and a transfer-lasso baseline denoted by Trans-Lasso.
Performance is summarized by the relative estimation error $\|\hat \bbeta-\bbeta_0\|^2/\|\bbeta_0\|^2$.
For the ridge-type methods, we use the smoothed Huber loss \eqref{eq:rho_eta} with $(\delta,\eta)=(1.35,0.1)$.
Each ridge tuning parameter is selected by $5$-fold cross-validation over a common grid, using the validation mean absolute error as the criterion.
Single-Lasso is implemented by $5$-fold cross-validated Lasso.
Trans-Lasso is implemented by the analogous two-stage procedure, with a source-stage cross-validated Lasso estimator followed by a target-stage cross-validated Lasso fit to the offset-adjusted response.
For each setup, we repeat the simulation $1000$ times.

\begin{figure}[ht]
    \centering
    \includegraphics[width=\textwidth]{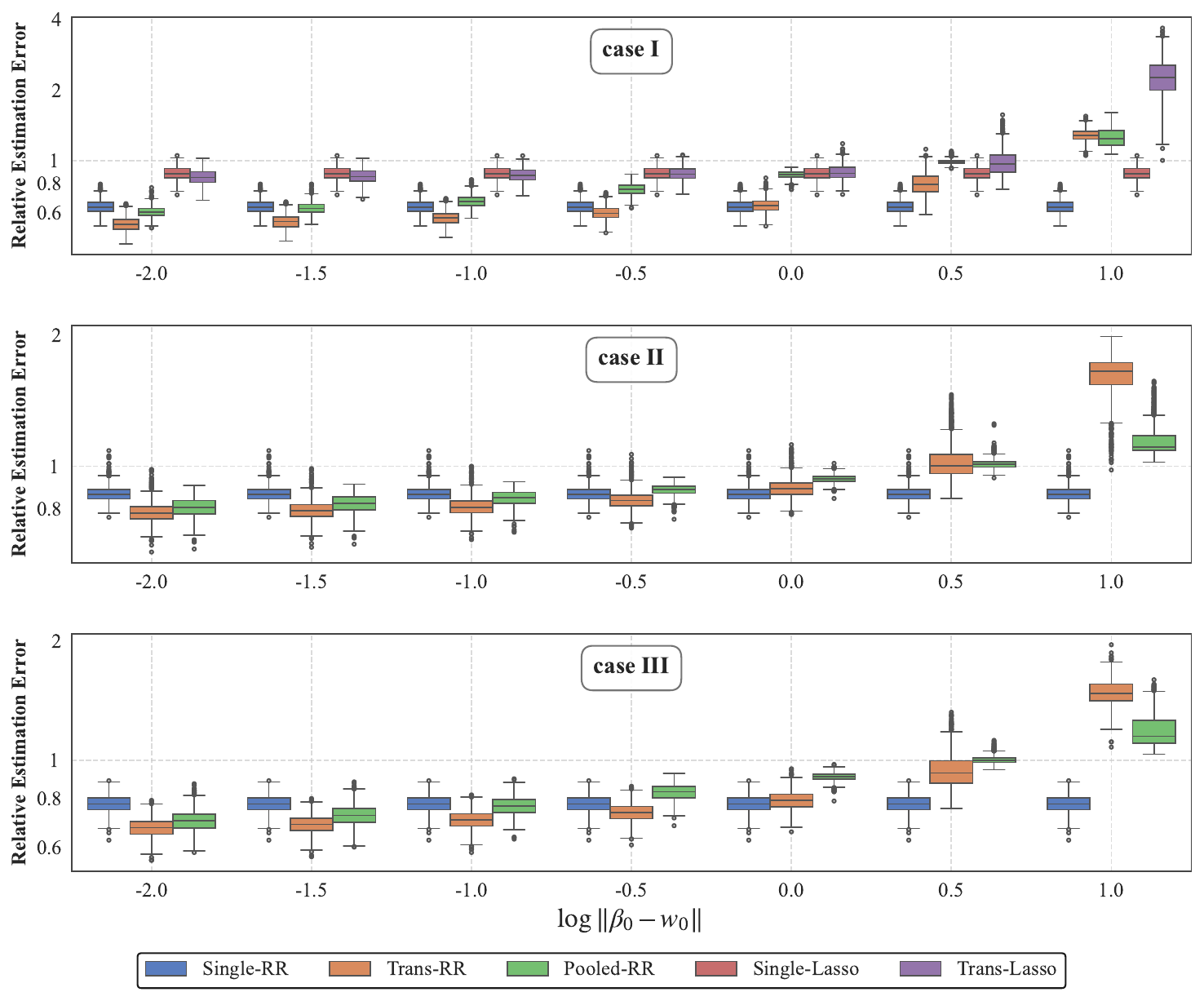}
    \caption{Boxplots of relative estimation errors (log scale) across $1000$ replications for varying $\|\bbeta_0 - \bfw_0\|$ under cases $\I$--$\III$, with $p=400$ and $n=400$. Case $\I$ includes all five methods, while cases $\II$ and $\III$ include only the three ridge-type procedures.
    Panels from top to bottom are for cases $\I, \II, \III$, respectively.
    }
    \label{figmethods}
\end{figure}

Figure \ref{figmethods} presents boxplots of the relative estimation errors on a logarithmic scale for varying values of $h$ across cases $\mathrm{I}$--$\mathrm{III}$.
 
Because our main focus in cases $\mathrm{II}$ and $\mathrm{III}$ is robustness under heavy-tailed settings, we restrict those panels to ridge-type procedures.
Under case $\mathrm{I}$, Trans-RR compares favorably with both lasso baselines, which is consistent with the nonsparse structure assumed throughout the paper.

Across all scenarios, Trans-RR consistently outperforms Pooled-RR.
This superiority can be attributed to the fact that Pooled-RR suffers from two sources of error: differences between $\bbeta_0$ and $\bfw_0$, and distributional differences between the predictors in the source and target domains.
As $h$ decreases, the performance gap between Pooled-RR and our method narrows, which aligns with our theoretical expectations since smaller $h$ indicates greater similarity between source and target domains.

More importantly, the comparison with Single-RR reveals the transition between positive and negative transfer.
When $h$ is small, Trans-RR achieves the lowest relative error.
As $h$ increases, its advantage shrinks and eventually reverses; in our experiments, this turnover occurs around $h = 1$.
This is consistent with the numerical evidence in Section \ref{sec_theory}: transfer tends to be more favorable when the source-stage estimator provides a better approximation to the target coefficient.

The same qualitative pattern appears in cases $\mathrm{II}$ and $\mathrm{III}$, showing that the method remains competitive under heavy-tailed and heterogeneous settings while still exhibiting the expected risk of negative transfer when the source is not sufficiently similar.

\section{Real Data Analysis}\label{sec_realdata}

We evaluate the proposed transfer procedure on the near-infrared (NIR) spectral dataset from the 2002 International Diffuse Reflectance Conference (IDRC) ``Shootout'' competition. 
The data consist of pharmaceutical tablet measurements collected by two spectrometers, with \texttt{ASSAY} as the response. 
For each instrument, the dataset contains a training sample of size $460$ and an external test sample of size $155$. 
Each spectrum is recorded at $650$ wavelengths, yielding a moderate-dimensional regression problem.

Let $X$ and $X_1$ denote the spectra from the two instruments. 
We consider two transfer directions: in Direction A, $X$ is the target domain and $X_1$ is the source domain; in Direction B, the roles are reversed. 
In each repetition, we randomly split the training sample of the target instrument into two parts of sizes $160$ and $300$. 
The subset of size $160$ is used as the target training sample, while the remaining $300$ tablets are matched with their measurements from the other instrument to form the source training sample. 
Thus, when $X$ is the target domain, the target fit is constructed from $160$ observations in $X$, and the source fit is constructed from the corresponding $300$ observations in $X_1$; the same scheme is applied symmetrically in the reverse direction. 
The external test sample of the target instrument is used for evaluation. 
To reduce Monte Carlo variability, we repeat this procedure $20$ times.

The following preprocessing procedure is used in both transfer directions.
Since adjacent wavelengths are highly correlated, we retain every fourth wavelength, resulting in $163$ predictors.
For each domain, predictors are whitened using the corresponding training sample,
\[
\tilde X=(X-\hat{\mu})\hat{\Sigma}^{-1/2},
\]
where $\hat{\mu}$ and $\hat{\Sigma}$ are estimated from that training sample, and the same transformation is applied to the associated test sample. 
The response is centered using the mean of the corresponding training response, and the same centering constant is used for the associated test response. 
Thus, all preprocessing parameters are estimated using training data only and then carried over to the corresponding test set, thereby avoiding information leakage.

We compare five methods: Single-RR, Trans-RR, Pooled-RR, Single-Lasso, and Trans-Lasso. 
For the ridge-type methods, we use the smoothed Huber loss in \eqref{eq:rho_eta} with $(\delta,\eta)=(1.35,0.1)$. 
All ridge and lasso tuning parameters are selected by cross-validation from the same grid
\[
\mathcal{G}=\{10^{-4},10^{-3.5},10^{-3},\ldots,10^{1}\}.
\]

\begin{table}[htbp]
\centering
\setlength\tabcolsep{16pt}
\caption{Prediction performance on the NIR spectral dataset over 20 repeated random splits. Each entry reports RMSE, with RMSE(sd) in parentheses.}
\label{tab:realdata-rmse}
\begin{tabular}{lcc}
\toprule
Method & Direction A & Direction B \\
\midrule
Trans-RR     & \textbf{4.6607} (0.2097) & \textbf{4.8161} (0.2664) \\
Pooled-RR    & 5.0963 (0.0814) & 5.4774 (0.1384) \\
Trans-Lasso  & 5.5668 (0.3650) & 5.6803 (0.3131) \\
Single-RR    & 6.3579 (2.1726) & 6.8375 (2.2707) \\
Single-Lasso & 8.0672 (2.8904) & 8.5607 (3.3739) \\
\bottomrule
\end{tabular}
\end{table}

Table~\ref{tab:realdata-rmse} reports the average RMSE and its standard deviation over the $20$ repetitions. 
Trans-RR performs best in both transfer directions, with the smallest RMSE and relatively small variability. 
Pooled-RR is the second-best method but is uniformly dominated by Trans-RR, suggesting that naive pooling fails to account for cross-instrument differences.
The lasso-based methods are less competitive overall, especially the single-domain lasso fit. 
This is consistent with the nonsparse nature of NIR spectra, where relevant information is typically distributed over broad wavelength regions rather than concentrated on a small subset of predictors.
The results therefore support the effectiveness of the proposed Trans-RR method.

\section{Discussion}

This paper introduces Trans-RR, a robust transfer-learning approach for moderate-dimensional linear regression.
It extends transfer-learning ideas of \cite{bastani2021predicting,li2022transfer} to a setting with non-sparse coefficients and heavy-tailed errors, without relying on sparsity assumptions or moment restrictions on the errors.
The theoretical results, simulation studies, and real-data analysis support the effectiveness of Trans-RR.
At the same time, both the theory and the numerical results show that negative transfer can occur when the source study is not sufficiently informative for the target study.

Several directions merit further study.
First, it would be useful to understand which choice of loss function leads to the smallest asymptotic estimation error in our framework.
One natural criterion is to compare candidate losses through the risk quantity characterized in Theorem \ref{thm:step_2}.
Second, in the present paper we restrict attention to a single source study.
Extending the framework to multiple source studies would require introducing a separate source coefficient and a separate source--target discrepancy for each source study.
This would make both the algorithm and the theoretical analysis substantially more involved; see \cite{li2022transfer} for related ideas in sparse settings.
Third, our theoretical analysis relies heavily on the assumption that the predictor covariance matrix is the identity matrix ($\cov(\bfx_i) = \mathbf{I}_p$).
This assumption, while mathematically convenient and commonly used in moderate-dimensional statistical theory \cite{el2018impact, sur2019modern}, is quite restrictive in practice.
It is therefore of interest to extend our analysis to more general covariance structures.
Finally, the possibility of negative transfer suggests the need for data-driven diagnostics or source-screening rules that can decide when transfer should be used and when a single-task fit is preferable.

\bibliographystyle{chicago}
\bibliography{ref}

\setcounter{section}{0}
\renewcommand{\thesection}{S.\arabic{section}}
\setcounter{equation}{0}
\counterwithout{equation}{section}
\renewcommand{\theequation}{S.\arabic{equation}}
\setcounter{theorem}{0}
\counterwithout{theorem}{section}
\renewcommand{\thetheorem}{S.\arabic{theorem}}
\renewcommand{\thefigure}{S.\arabic{figure}}
\renewcommand{\thetable}{S.\arabic{table}}
\setcounter{lemma}{0}
\renewcommand{\thelemma}{S.\arabic{lemma}}
\renewcommand{\theproposition}{S.\arabic{proposition}}
\setcounter{proposition}{0}

\newpage

\begin{center}
\textit{\large Supplementary Material to}
\end{center}
\begin{center}
{\LARGE Transfer Learning for Moderate-Dimensional Ridge-Regularized Robust Linear Regression}
\vskip10pt
\end{center}

\section{Assumptions} \label{Appendix_assumptions}

\textbf{Notation.}
Let $\polyLog(p)$ replace a power of $\log(p)$.
Denote by $\I_m$ the $m \times m$ identity matrix.
When this does not create problems, we also use the standard notation $\I$ for $\I_p$.
Let $\bzero_m \in \mathbb{R}^m$ and $\bone_m \in \mathbb{R}^m$ be the vectors of zeros and ones, respectively.
For an $m \times m$ matrix $\bfA=\{a_{i j}\}_{1 \leq i, j \leq m}$, denote by $\lambda_{\max }(\bfA)$ and $\lambda_{\min }(\bfA)$ the maximum and minimum eigenvalues of $\bfA$, respectively.
The $L_2$ norm of $\bfA$ is defined as $\|\bfA\|=\{\lambda_{\max }(\bfA^{\T} \bfA)\}^{1 / 2}$.
We call $\hat \Sigma = n^{-1} \sum_{i=1}^n \bfx_i \bfx_i^\T$ the sample covariance matrix of $\bfx_i$'s when $\bfx_i$'s are known to have zero mean.
We say that $X \leq Y$ in $L_k$ if $\E(|X|^k) \leq \E(|Y|^k)$.
We use the notation $(a, b)$ to represent either the interval $(a, b)$ or $(b, a)$, as in some cases we need to localize quantities within intervals defined by two values, $a$ and $b$, without knowing in advance whether $a < b$ or $b > a$.
We denote by $\bfX$ the $n \times p$ design matrix whose $i$-th row is $\bfx_i^\T$, and by $\bfX_{(i)}$ the matrix obtained by removing the $i$-th row of $\bfX$.
We write $a \wedge b$ for $\min (a, b)$ and $a \vee b$ for $\max (a, b)$.
For two symmetric matrices $A$ and $B$, $A \succeq B$ means that $A-B$ is positive semi-definite.
For the random variable $W$, we use the definition $\|W\|_{L_k}=\{\E(|W|^k) \}^{1 / k}$.
For sequences of random variables $W_n, Z_n$, we use the notation $W_n=\mO_{L_k}\left(Z_n\right)$ and $W_n=\mo_{L_k}(Z_n)$ when $\|W_n\|_{L_k}=\mathrm{O}(\|Z_n\|_{L_k})$ and $\|W_n\|_{L_k}=\mo(\|Z_n\|_{L_k})$, respectively.
For a vector $\bfv=(v_1, \ldots, v_m)^\T$, the $L_2$ norms are $\|\bfv\|=(\sum_{i=1}^m v_i^2)^{1 / 2}$, whereas $\|v\|_{\infty}=\max _{1 \leq k \leq p}|v(k)|$.
For a function $f$ from $\RR$ to $\RR$, denote $\|f\|_{\infty}=\sup _{x \in \RR}|f(x)|$.

In this appendix, we refine the assumptions in the main text into the forms needed for the proof of Theorem \ref{thm:step_2}.
The proof is divided into three parts, and each part uses a slightly different set of technical conditions.
In particular, the assumptions on $\bbeta_0$ are specified below according to the needs of the different proof subsections, whereas the source-side conditions remain as in Assumption \ref{ass:source}.

\paragraph{Assumptions under which the whole proof goes through}

\begin{itemize}
    \item \textbf{M1.}
    $p/n \to \kappa \in (0,\infty)$.

    \item \textbf{M2.}
    There exists constants $C_{\bbeta}$ and $\me > 1/3$ such that $\|\bbeta_0\| \leq C_{\bbeta}$ and $\|\bbeta_0\|_{\infty}=\mO(n^{-\me})$.

    \item \textbf{M3.}
    Suppose $\rho$ is an even function.
    Assume that $\psi = \rho'$ is bounded and $\psi'$ is Lipschitz and bounded.
    Moreover, we assume that $\sign(\psi(x)) = \sign(x)$ and that $\rho(x) \geq \rho(0) = 0$ for all $x \in \RR$.
    \item \textbf{M4.}
    Assume that there exist independent variables $\lambda_i$'s and $\cX_i$'s such that $\bfx_i=\lambda_i \cX_i$.
    Suppose that $\cX_i$'s are i.i.d. with independent entries, and they have mean $\bzero_p$ and $\cov(\cX_i)=\I_p$.
    Suppose there exist $\mc_n$ and $C_n$ that vary with $n$, where $1/\mc_n = \mO(\polyLog(n))$ and $C_n$ is bounded in $n$ ,such that for any convex $1$-Lipschitz function $G$ of $\cX_i$, $P(|G(\cX_i) - m_G| > t) \leq C_n \exp(-\mc_n t^2)$ holds for all $t > 0$, where $m_G$ is the median of $G(\cX_i)$.
    We require the same assumption to hold for the columns of the $n \times p$ design matrix $\cX$.
    Additionally, we assume that the coordinates of $\cX_i$ have moments of all orders, and the $k$-th moment of the entries of $\cX_i$ is assumed to be uniformly bounded independently of $n$ and $p$ for all $k$.
    Also, for any $1 \leq k \leq p$, the vectors $\Theta_k=(\cX_1(k), \ldots, \cX_n(k))$ in $\RR^n$ satisfy: for any $1$-Lipschitz (with respect to Euclidean norm) convex function $G$, if $m_{G(\Theta_k)}$ is a median of $G(\Theta_k)$, for any $t>0, P(|G(\Theta_k)-m_{G(\Theta_k)}|>t) \leq C_n \exp (-\mc_n t^2)$, $C_n$ and $\mc_n$ can vary with $n$.
    As above, we assume that $1 / \mc_n=\mathrm{O}(\polyLog(n))$.

    \item \textbf{M5.}
    $\lambda_i$'s are independent, with $\E(\lambda_i^2) = 1$, $\E(\lambda_i^4)$ being bounded, and $\sup_{1 \leq i \leq n} |\lambda_i|$ growing at most like $C_\lambda(\log n)^k$ for some $k$.
    $\lambda_i$'s may have finitely many possible distributions.
    \item \textbf{M6.}
    Suppose that $\epsilon_i$'s are independent and independent with $\cX_i$'s and $\lambda_i$'s.
    They may have finitely many possible distributions, each with a density that is differentiable, symmetric, and unimodal.
    Furthermore, for any $r \in \RR$, if $z \sim N(0,1)$, independent of $\epsilon_i$, $\epsilon_i + rz$ has a differentiable density $f_{i,r}$ which is increasing on $(-\infty, 0)$ and decreasing on $(0, \infty)$.
    $\lim _{x \to \infty} x f_i(x) = 0$.

    \item \textbf{M7.}
    $\epsilon_i$'s can have different distributions.
    Similarly, $\lambda_i$'s can have different distributions.
    The fraction of occurrences for each possible combination of distributions for $(\epsilon_i, \lambda_i)$ has a limit as $n \rightarrow \infty$.

\end{itemize}

\paragraph{First part of the proof (Section \ref{sec_p1})}

\begin{itemize}
    \item \textbf{O1.}
    $p/n \to \kappa \in (0, \infty)$.

    \item \textbf{O2.}
    $\rho$ is twice differentiable, convex and non-linear. $\psi=\rho'$.
    Note that $\psi' \geq 0$ since $\rho$ is convex. We assume that $\sign(\psi(x))=\sign(x)$ and $\rho \geq 0=\rho(0)$.

    \item \textbf{O3.}
    $\sup_x |\psi(x)| \leq C \polyLog(n)$ and $\|\psi^2\|_{\infty} \leq C$ for some constant $C$.
    Furthermore, $\psi'$ is assumed to be $\mathrm{L}(n)$-Lipschitz with $\mathrm{L}(n) \leq C n^\alpha, \alpha \geq 0$.
    We also assume that $\|\psi'\|_{\infty} \leq C \polyLog(n)$.

    \item \textbf{O4.}
    Assume that there exist independent variables $\lambda_i$'s and $\cX_i$'s such that $\bfx_i=\lambda_i \cX_i$.
    Suppose that $\cX_i$'s are i.i.d. with independent entries, and they have mean $\bzero_p$ and $\cov(\cX_i)=\I_p$.
    Suppose there exist $\mc_n$ and $C_n$ that vary with $n$, where $1/\mc_n = \mO(\polyLog(n))$ and $C_n$ is bounded in $n$ ,such that for any convex $1$-Lipschitz function $G$ of $\cX_i$, $P(|G(\cX_i) - m_G| > t) \leq C_n \exp(-\mc_n t^2)$ holds for all $t > 0$, where $m_G$ is the median of $G(\cX_i)$.
    We require the same assumption to hold for the columns of the $n \times p$ design matrix $\cX$.
    Additionally, we assume that the coordinates of $\cX_i$ have moments of all orders, and the $k$-th moment of the entries of $\cX_i$ is assumed to be uniformly bounded independently of $n$ and $p$ for all $k$.

    \item \textbf{O5.}
    $\{\cX_i\}_{i=1}^n$ and $\{\lambda_i\}_{i=1}^n$ are independent of $\{\epsilon_i\}_{i=1}^n$.
    $\epsilon_i$'s are independent of each other.

    \item \textbf{O6.}
    $\sup _{1 \leq i \leq n}|\lambda_i| \triangleq \mathcal{L}_n=\mO_{L_k}(\polyLog(n))$ and $\lambda_i$ 's are independent.
    Moreover, $\E(\lambda_i^2)=1$.
    % (Note that this implies that $\cov(X_i)=\cov(\cX_i)$.)

    \item \textbf{O7.}
    $1-2 \alpha>0$ and $\|\bbeta_0\|=\mO(\polyLog(n))$.
\end{itemize}

\paragraph{Second part of the proof (Section \ref{sec_p2})}

\begin{itemize}
    \item \textbf{P1.}
    $\cX_i$ 's have independent entries.
    Furthermore, for any $1 \leq k \leq p$, the vectors $\Theta_k=(\cX_1(k), \ldots, \cX_n(k))$ in $\RR^n$ satisfy: for any $1$-Lipschitz (with respect to Euclidean norm) convex function $G$, if $m_{G(\Theta_k)}$ is a median of $G(\Theta_k)$, for any $t>0, P(|G(\Theta_k)-m_{G(\Theta_k)}|>t) \leq C_n \exp (-\mc_n t^2)$, $C_n$ and $\mc_n$ can vary with $n$.
    As above, we assume that $1 / \mc_n=\mathrm{O}(\polyLog(n))$.

    \item \textbf{P2.}
    $\|\psi'\|_{\infty}=\mO(1)$.

    \item \textbf{P3.}
    $\|\bbeta_0\|_{\infty}=\mO(n^{-\me})$, where $\me>0$.
    Furthermore, $\|\bbeta_0\|_2 \leq C$, where $C$ is a constant independent of $p$ and $n$.
    $\me$ satisfies $\alpha+1 / 4-\me<0$.

    \item \textbf{P4.}
    $1 / 2-2 \alpha>0$ and $\min (1 / 2, \me)-\alpha-1 / 4>0$. The latter implies that $\min (1 / 2, \me )-\alpha>0$
\end{itemize}

\paragraph{Last part of the proof (Section \ref{sec_p3})}

\begin{itemize}
    \item \textbf{F1.}
    $\epsilon_i$'s may have different distributions; however, they may only come from finitely many distributions.
    Furthermore, for any $r \in \RR$, if $z \sim N(0,1)$, independent of $\epsilon_i$, $\epsilon_i + rz$ has a differentiable density $f_{i,r}$ which is increasing on $(-\infty, 0)$ and decreasing on $(0, \infty)$.
    $\lim _{x \to \infty} x f_i(x) = 0$.

    \item \textbf{F2.}
    $\| \psi\|_\infty = \mO(1)$.
    $\psi'$ has Lipschitz constant $\mathrm{L}(n)$.
    Furthermore, $\mathrm{L}(n) \| \psi\|_\infty = \mO(1)$.

    \item \textbf{F3.}
    $\alpha < 1/6$ and $\alpha + 1/3 < 2 \min(1/2, \me)$.

    \item \textbf{F4.}
    there exists constant $C$ such that $\E (\lambda_i^4) \leq C$.

    \item \textbf{F5.}
    $\lambda_i$'s may have different distributions.
    The fraction of occurrences for each possible combination of distributions for $(\epsilon_i, \lambda_i)$ has a limit as $n \rightarrow \infty$.

\end{itemize}

\section{Proof for Theorem \ref{thm:step_2}}

% Suppose $\bfw_0$ is the solution to Equation \eqref{eq:w_solution}, we denote $\bdelta_0 = \bbeta_0 - \bfw_0$.
We call
\begin{equation*}
F(\bdelta)=\frac{1}{n} \sum_{i=1}^{n} \rho \{\epsilon_i + \bfx_i^\T (\bfw_0 - \hat \bfw) +  \bfx_i^\T (\bdelta_0 -  \bdelta) \}+\frac{\tau}{2}\|\bdelta\|^2.
\end{equation*}
$\hat\bdelta$ is defined as the solution of
\begin{equation*}
\begin{aligned}
f(\widehat{\bdelta}) & =0 \ \text { with } \\
\nabla F=f(\bdelta) & =\frac{1}{n} \sum_{i=1}^{n}- \bfx_i \psi\{\epsilon_i + \bfx_i^\T (\bfw_0 - \hat \bfw) +  \bfx_i^\T (\bdelta_0 -  \bdelta) \}+\tau \bdelta .
\end{aligned}
\end{equation*}

We further define
\begin{equation} \label{eq:proof_def}
\tilde\epsilon_i = \epsilon_i + \bfx_i^\T (\bfw_0 - \hat \bfw), \quad
R_i  = \tilde\epsilon_i +  \bfx_i^\T (\bdelta_0 -  \hat\bdelta) ,   \\
\bfS  =\frac{1}{n} \sum_{i=1}^{n} \psi'(R_i) \bfx_i \bfx_i^{\T}, \quad
c_\tau  =\frac{1}{n} \tr(\bfS+\tau \I)^{-1}  .
\end{equation}

\subsection[\appendixname~\thesubsection]{Preliminaries}

\begin{lemma}
\label{lem:w-coordinate-rate}
Under Assumptions \textbf{P3}--\textbf{P4}, for any $l=1,\ldots,p$, where $w(l)$ denotes the $l$th coordinate of a vector $\bfw$,
\[
|\hat w(l)-w_0(l)|
=
O_{L_k}\!\left(\polyLog(n_1)\,n_1^{-1/2}+n_1^{-\me}\right)
=
\mO_{L_k}\!\left(\frac{\polyLog(n_1)}{n_1^{1/2}\wedge n_1^{\me}}\right).
\]
In particular, if we further assume $n = \mO(n_1)$, we have
\begin{equation*}
    |\hat w(l) - w_0(l)| = \mO_{L_k}\Big(\frac{\polyLog(n)}{\sqrt{n} \wedge n^{e}}\Big).
\end{equation*}
\end{lemma}

\begin{proof}
By Proposition~3.12 of \citet{el2018impact},
\[
\left| \mathfrak w_l - w_0(l) \right|
=
O_{L_k}\!\left(
\polyLog(n_1)\,n_1^{-1/2}
+\|w_0\|_{\infty}
\right),
\]
where $\mathfrak w_l$ denotes the analogue of $\mathfrak b_p$ defined in Appendix~4 of \citet{el2018impact}. 
Its explicit construction is rather involved and is not needed here; we only use $\mathfrak w_l$ as an intermediate quantity in the argument.
Under Assumption~P3, $\|w_0\|_{\infty}=O(n_1^{-\me})$, hence
\begin{equation}
\label{eq:oracle-part-final}
\left| \mathfrak w_l - w_0(l) \right|
=
O_{L_k}\!\left(
\polyLog(n_1)\,n_1^{-1/2}
+n_1^{-\me}
\right)
=
\mO_{L_k}\!\left(\frac{\polyLog(n_1)}{n_1^{1/2}\wedge n_1^{\me}}\right).
\end{equation}

Moreover, by Theorem~3.20 of \citet{el2018impact},
\begin{equation}
\label{eq:approx-part-final}
\hat w(l)-\mathfrak w_l
=
O_{L_k}\!\left(
\frac{\polyLog(n_1)\,n_1^{\alpha}}
{[n_1^{1/2}\wedge n_1^{\me}]^2}
\right).
\end{equation}
Let $m:=\min(1/2,\me)$. Then $[n_1^{1/2}\wedge n_1^{\me}]^2=n_1^{2m}$ and
\[
\frac{\polyLog(n_1)\,n_1^{\alpha}}
{[n_1^{1/2}\wedge n_1^{\me}]^2}
=
\polyLog(n_1)\,n_1^{-m} n_1^{\alpha-m}.
\]
Assumption~P4 gives $m > \alpha$, thus
\[
\frac{\polyLog(n_1)\,n_1^{\alpha}}
{[n_1^{1/2}\wedge n_1^{\me}]^2}
= o\!\left(\polyLog(n_1)\,n_1^{-m}\right),
\]
and hence \eqref{eq:approx-part-final} yields
\begin{equation}
\label{eq:approx-negl-final}
\hat w(l)-\mathfrak w_l
=
o_{L_k}\!\left(\polyLog(n_1)\,n_1^{-m}\right).
\end{equation}

By the triangle inequality and Minkowski's inequality in $L_k$,
\[
\|\hat w(l)-w_0(l)\|_{L_k}
\le
\|\hat w(l)-\mathfrak w_l\|_{L_k}
+
\|\mathfrak w_l-w_0(l)\|_{L_k}.
\]
Using \eqref{eq:oracle-part-final} and \eqref{eq:approx-negl-final}, we obtain
\[
\|\hat w(l)-w_0(l)\|_{L_k}
=
\mO\!\left(\frac{\polyLog(n_1)}{n_1^{1/2}\wedge n_1^{\me}}\right),
\]
which is equivalent to
\[
|\hat w(l)-w_0(l)|
=
\mO_{L_k}\!\left(\frac{\polyLog(n_1)}{n_1^{1/2}\wedge n_1^{\me}}\right).
\]

Assume $n=\mO(n_1)$, i.e., $n_1\ge n/C$ for some $C>0$ and all large $n$.
Since $\polyLog(t)$ denotes a power of $\log t$, write
$\polyLog(t)=(\log t)^q$ for some fixed $q\ge 0$ (up to a constant factor).
Consider $f(t):=(\log t)^q\,t^{-m}$.
For $t\ge \exp(q/m)$,
\[
f'(t)=\frac{(\log t)^{q-1}}{t^{m+1}}\big(q-m\log t\big)\le 0,
\]
so $f$ is decreasing for all sufficiently large $t$.
Hence, for all sufficiently large $n$ and all $n_1\ge n/C$,
\[
\frac{\polyLog(n_1)}{\sqrt{n_1}\wedge n_1^{e}}
=\frac{(\log n_1)^q}{n_1^{m}}
=f(n_1)
\le f(n/C)
=\frac{(\log (n/C))^q}{(n/C)^m}
\le C^m\,\frac{(\log n)^q}{n^m}
\asymp \frac{\polyLog(n)}{\sqrt n\wedge n^{e}},
\]
where we use $a_n \asymp b_n$ to denote that $a_n = O(b_n)$ and $b_n = O(a_n)$.
Therefore, we have
\[
|\hat w(l)-w_0(l)|
=\mO_{L_k}\!\Big(\frac{\polyLog(n)}{\sqrt n\wedge n^{e}}\Big).
\]
\end{proof}

\begin{proposition}
\label{prop:1}
    Let $\bdelta_1$ and $\bdelta_2$ be the two vectors in $\RR^p$.
    Then, when $\rho$'s are convex and twice-differentiable,
    \begin{equation}
    \label{eq:prop1}
        \|\bdelta_1 - \bdelta_2\| \leq \frac{1}{\tau} \|f(\bdelta_1) - f(\bdelta_2)\|.
    \end{equation}
\end{proposition}

\begin{proof}
We have by definition
\begin{equation*}
f(\bdelta_1) - f(\bdelta_2) = \tau (\bdelta_1 - \bdelta_2) + \frac{1}{n} \sum_{i=1}^{n} \bfx_i [\psi\{ \tilde\epsilon_i +  \bfx_i^\T (\bdelta_0 -  \bdelta_2) \} - \psi\{\tilde\epsilon_i +  \bfx_i^\T (\bdelta_0 -  \bdelta_1) \}].
\end{equation*}
By the mean value theorem we have
\begin{equation*}
\psi\{ \tilde\epsilon_i +  \bfx_i^\T (\bdelta_0 -  \bdelta_2) \} - \psi\{\tilde\epsilon_i +  \bfx_i^\T (\bdelta_0 -  \bdelta_1) \} = \psi'(\gamma_i^\star) \bfx_i^\T (\bdelta_0 -  \bdelta_1),
\end{equation*}
where $\gamma_i^\star$ is in the interval $\big(\tilde\epsilon_i +  \bfx_i^\T (\bdelta_0 -  \bdelta_1), \tilde\epsilon_i +  \bfx_i^\T (\bdelta_0 -  \bdelta_2)\big)$.

Hence,
\begin{eqnarray*}
f(\bdelta_1) - f(\bdelta_2) &=& \tau (\bdelta_1 - \bdelta_2) + \frac{1}{n} \sum_{i=1}^{n} \psi'(\gamma_i^\star) \bfx_i \bfx_i^\T (\bdelta_0 -  \bdelta_1)\\
&=& (\bfS_{\bdelta_1, \bdelta_2} + \tau \I_p)(\bdelta_1 - \bdelta_2),
\end{eqnarray*}
where
\begin{equation*}
\bfS_{\bdelta_1, \bdelta_2} = \frac{1}{n} \sum_{i=1}^{n} \psi'(\gamma_i^\star) \bfx_i \bfx_i^\T.
\end{equation*}
This shows that
\begin{equation*}
\bdelta_1 - \bdelta_2 = (\bfS_{\bdelta_1, \bdelta_2} + \tau \I_p)^{-1} \{ f(\bdelta_1) - f(\bdelta_2)\}.
\end{equation*}

Since $\rho$ is convex, $\psi' = \rho''$ is non-negative and $\bfS_{\bdelta_1, \bdelta_2}$ is positive semi-definite.
In the semi-definite order, we have $\bfS_{\bdelta_1, \bdelta_2} + \tau \I_p \succeq \tau \I_p$.
In particular,
\begin{equation*}
    \|\bdelta_1 - \bdelta_2\| \leq \frac{1}{\tau} \|f(\bdelta_1) - f(\bdelta_2)\|.
\end{equation*}
\end{proof}

Proposition \ref{prop:1} yields the following lemma.

\begin{lemma}
\label{lemma:delta_ineq_1}
For any $\bdelta_1$,
\begin{equation*}
\| \hat \bdelta - \bdelta_1\| \leq \frac{1}{\tau} \|f(\bdelta_1)\|.
\end{equation*}
\end{lemma}
The lemma is a simple consequence of equation \eqref{eq:prop1} by definition $f(\hat \bdelta) = 0$.

\subsection[\appendixname~\thesubsection]{On \texorpdfstring{$\|\hat \bdelta\|$}{betahat} and \texorpdfstring{$\|\hat \bdelta - \bdelta_0\|$}{hat-0}}

\begin{lemma}
\label{lemma:delta_moment_bounds}
Define $\bfq_n(\bfb) = n^{-1} \sum_{i=1}^{n} \bfx_i \psi\{\tilde\epsilon_i + \bfx_i^\T \bfb \}$, $\bfq_n \in \RR^p$.

If $\bfD_{\psi}$ is the $n \times n$ diagonal matrix with $(i,i)$-entry $\psi\{\tilde\epsilon_i + \bfx_i^\T \bdelta_0 \}$,
%inequality \eqref{eq:delta_ineq_1} can be expressed as
\begin{equation*}
\|\hat\bdelta\| \leq\frac{1}{\tau} \|\bfq_n(\bdelta_0)\| =  \frac{1}{\tau} \sqrt{\frac{1}{n^2} \bone^\T \bfD_{\psi} \bfX \bfX^\T \bfD_{\psi} \bone},
\end{equation*}
and if $\bfD_{\psi(\xi_i)}$ is the $n \times n$ diagonal matrix with $(i,i)$-entry $\psi(\tilde\epsilon_i)$, %inequality \eqref{eq:delta_error_ineq} can be expressed as
\begin{equation*}
\label{eq:delta_error_ineq_2}
\|\hat\bdelta - \bdelta_0\| \leq \|\bdelta_0\| + \frac{1}{\tau} \|\bfq_n(\bzero)\| = \|\bdelta_0\| + \frac{1}{\tau} \sqrt{\frac{1}{n^2} \bone^\T \bfD_{\psi(\xi_i)} \bfX \bfX^\T \bfD_{\psi(\xi_i)} \bone},
\end{equation*}
Also,
\begin{equation*}
\|\bfq_n(\bdelta_0)\|^2 \leq \frac{\bone^\T \bfD_{\psi}^2 \bone}{n} \|\bfX \bfX^\T/n\|_2,
\end{equation*}
where $\|\bfA\|_2$ denotes the largest singular value of the matrix $\bfA$.

Therefore, under assumptions \textbf{O1-O6},
\begin{eqnarray*}
    && \E(\|\hat\bdelta\|^2) \leq \frac{1}{\tau^2} \frac{p}{n} C^2 \polyLog(n), \\
    && \E(\|\hat\bdelta\|^4) \leq \frac{1}{\tau^4}  C \polyLog(n).
\end{eqnarray*}

Similarly, for any finite $k$,
    \begin{equation*}
    \E (\|\hat\bdelta - \bdelta_0\|_2^k) \leq C_k [ \|\bdelta_0\|^k + \polyLog(n) / \tau^k].
    \end{equation*}
    In the case $k=2$, we have the more precise bound
    \begin{equation*}
    \E (\|\hat \bdelta - \bdelta_0\|^2) \leq 2 \Big[ \|\bdelta_0\|^2 + \frac{p/n}{\tau^2} \frac{1}{n} \sum_{i=1}^{n} \E \{\psi^2(\tilde\epsilon_i)\} \Big].
    \end{equation*}
\end{lemma}

Noting that \cite{el2018impact} has shown that $\E (\|\hat \bfw - \bfw_0\|) = \mO(1)$, we have $\E (\|\hat \bdelta + \hat \bfw - \bdelta_0 - \bfw_0\| )$ is bounded by $K \operatorname{polyLog}(n) / \tau^k$.

\begin{proof}
    Recall that $f(\bdelta)  =\frac{1}{n} \sum_{i=1}^{n}- \bfx_i \psi\{\tilde\epsilon_i +  \bfx_i^\T (\bdelta_0 -  \bdelta) \}+\tau \bdelta $. Applying Lemma \ref{lemma:delta_ineq_1} with $\bdelta_1 = \bzero$ we have
\begin{equation*}
\label{eq:delta_ineq_1}
\| \hat \bdelta - \bzero\| \leq \frac{1}{\tau} \|f(\bzero)\| = \frac{1}{\tau} \Big\| \frac{1}{n} \sum_{i=1}^{n}- \bfx_i \psi\{\tilde\epsilon_i +  \bfx_i^\T \bdelta_0 \} \Big\|,
\end{equation*}
which gives the first inequality.
% and using $\bdelta_1 = -\hat\bfw$ we have
% \begin{equation*}
% \| \hat \bdelta + \hat \bfw\| \leq \frac{1}{\tau} \|f(-\hat \bfw)\| = \frac{1}{\tau} \Big\| \frac{1}{n} \sum_{i=1}^{n}- \bfx_i \psi\{\epsilon_i + \bfx_i^\T \bbeta_0 \} - \tau \hat \bfw \Big\|,
% \end{equation*}
% which implies that
% \begin{equation*}
% \|\hat\bdelta\| - \|\hat \bfw\| \leq \frac{1}{\tau} \Big\| \frac{1}{n} \sum_{i=1}^{n}- \bfx_i \psi\{\epsilon_i + \bfx_i^\T \bbeta_0 \} \Big\| + \|\hat \bfw \|.
% \end{equation*}

Using $\bdelta_1 = \bdelta_0$ we have
\begin{equation*}
\label{eq:delta_error_ineq}
\| \hat \bdelta - \bdelta_0\| \leq \frac{1}{\tau} \|f(\bdelta_0)\| = \frac{1}{\tau} \Big\| \frac{1}{n} \sum_{i=1}^{n}- \bfx_i \psi(\tilde\epsilon_i) + \tau \bdelta_0 \Big\|,
\end{equation*}
which gives the second inequality.

We note that under our assumptions, according to Lemma 3.38 from \cite{el2018impact}, we have
\begin{equation*}
\|\bfX \bfX^\T/n\|_2 = \mathrm{O}_{L_k}(\operatorname{polyLog}(n) )
\end{equation*}
and
\begin{equation*}
\frac{1}{n} \sum_{i=1}^{n} \psi^2\{\tilde\epsilon_i + \bfx_i^\T \bdelta_0 \} \leq \frac{1}{n} \sum_{i=1}^{n} \| \psi^2\|_{\infty} = \mathrm{O}(1),
\end{equation*}
which gives all the results about $L_k$ bounds.

For the last result, we note that
\begin{equation*}
    \|\bfq_n(\bzero)\|^2 = \bfq_n(\bzero)^\T \bfq_n(\bzero) = \frac{1}{n^2} \sum_{i,j} \bfx_i^\T \bfx_j \psi(\tilde\epsilon_i) \psi(\tilde\epsilon_j).
\end{equation*}
It implies that
\begin{equation*}
    \E \|\bfq_n(\bzero)\|^2 = \frac{1}{n^2} \sum_{i=1}^n \E (\|\bfx_i\|^2) \E \{\psi^2 (\tilde\epsilon_i)\} .
\end{equation*}
Because $\E (\|\bfx_i\|^2) = p$, we can conclude that
\begin{equation*}
    \E (\| \bfq_n(\bzero)\|^2) = \frac{p}{n}\frac{1}{n} \sum_{i=1}^{n} \E \{\psi^2 (\tilde\epsilon_i)\}.
\end{equation*}
Together with the bound
\begin{equation*}
    \|\hat \bdelta - \bdelta_0\|^2 \leq 2\|\bdelta_0\|^2 + \frac{2}{\tau^2}\|\bfq_n(\bzero)\|^2,
\end{equation*}
it implies the last result about $k=2$.
\end{proof}

% \begin{lemma}[{{\citet[Lemma 3.38]{el2018impact}}}]
% \label{lemma:el_3.38}
% Suppose $\bfx_i$'s are independent random vectors in $\mathbb{R}^p$, satisfying \textbf{O4}, and having mean 0 and covariance $\I_p$. Suppose that $\lambda_i$'s satisfy \textbf{O6}. Let $\widehat{\bSigma}=n^{-1} \sum_{i=1}^{n} \bfx_i \bfx_i^\T$. Then,
% $$
% \| \widehat{\bSigma}  \|_2 = \mathrm{O}_P(\operatorname{polyLog}(n) \mc_n^{-1}) .
% $$
%
% The results hold also in $L_k$.
% \end{lemma}

\subsection[\appendixname~\thesubsection]{Leave-one-observation-out} \label{sec_p1}
In this subsection, we approximate $\hat\bdelta$ by $\hat \bdelta_{(i)}$ via leave-one-observation-out method.

We consider the situation where we leave the $i$-th observation, $(\bfx_i,\epsilon_i)$, out. By definition,
\begin{equation*}
\hat\bdelta_{(i)} = \argmin_{\bdelta \in \RR^p} F_i(\bdelta), \ \text{ where }
F_i(\bdelta) = \frac{1}{n} \sum_{j \neq i} \rho_j\{\tilde\epsilon_j + \bfx_j^\T \bdelta_0 - \bfx_j^\T \bdelta\} + \frac{\tau}{2} \|\bdelta\|^2.
\end{equation*}
We call
\begin{eqnarray*}
f_i(\bdelta) &=& -\frac{1}{n} \sum_{j \neq i} \bfx_j\psi_j\{\tilde\epsilon_j + \bfx_j^\T \bdelta_0 - \bfx_j^\T \bdelta\} + \tau\bdelta \\
&=& f(\bdelta) + \frac{1}{n} \bfx_i\psi\{\tilde\epsilon_i + \bfx_i^\T \bdelta_0 - \bfx_i^\T \bdelta\}.
\end{eqnarray*}
We have
\begin{equation*}
f_i(\hat \bdelta_{(i)}) = 0.
\end{equation*}

We call
\begin{equation*}
\tilde r_{j,(i)} = \tilde\epsilon_j - \bfx_j^\T(\hat \bdelta_{(i)}-\bdelta_0) \text{ and } \bfS_i = \frac{1}{n} \sum_{j \neq i} \psi_j' (\tilde r_{j,(i)}) \bfx_j \bfx_j^\T.
\end{equation*}

Consider
\begin{equation*}
\tilde \bdelta_i = \hat \bdelta_{(i)} + \frac{1}{n} (\bfS_i + \tau\I)^{-1} \bfx_i \psi \{\prox(c_i\rho)(\tilde r_{j,(i)})\} \triangleq \hat  \bdelta_{(i)} + \boldeta_i,
\end{equation*}
where
\begin{eqnarray}
&& c_i = \frac{1}{n} \bfx_i^\T (\bfS_i + \tau\I)^{-1} \bfx_i, \label{eqn:c_i} \\
&& \boldeta_i = \frac{1}{n} (\bfS_i + \tau\I)^{-1} \bfx_i \psi \{\prox(c_i\rho)(\tilde r_{i,(i)})\}. \nonumber
\end{eqnarray}

\subsubsection{Deterministic bounds}
\begin{proposition}
We have
\begin{equation*}
\|\hat\bdelta - \tilde\bdelta_i\| \leq \frac{1}{\tau} \|\cR_i\|,
\end{equation*}
where
\begin{equation*}
\cR_i = \frac{1}{n} \sum_{j\neq i} [\psi_j'\{\gamma^\star (\bfx_j,\hat\bdelta_{(i)},\boldeta_i)\} - \psi_j'(\tilde r_{j, (i)})] \bfx_j \bfx_j^\T \boldeta_i,
\end{equation*}
and $\gamma^\star (\bfx_j,\hat\bdelta_{(i)},\boldeta_i)\}$ is in the (``unordered'') interval $(\tilde r_{j, (i)}, \tilde r_{j, (i)}-\bfx_j^\T \boldeta_i)$.
\end{proposition}

\begin{proof}
Recall that $y_i = \epsilon_i + \bfx_i^\T \bfw_0 + \bfx_i^\T \bdelta_0$.

Since $f_i(\hat \bdelta_{(i)}) = 0$, and $\tilde\bdelta_i = \hat \bdelta_{(i)} + \boldeta_i$, we have
\begin{eqnarray*}
    f(\tilde\bdelta_i) &=& f(\tilde \bdelta_i) - f_i(\hat \bdelta_{(i)})\\
     &=& -\frac{1}{n} \sum_{i=1}^{n} \bfx_i \psi\{\tilde\epsilon_i + \bfx_i^\T \bdelta_0 - \bfx_i^\T \tilde\bdelta_i\} + \tau \tilde\bdelta_i + \frac{1}{n} \sum_{j \neq i} \bfx_j \psi\{\tilde\epsilon_j + \bfx_j^\T \bdelta_0 - \bfx_j^\T \bdelta_{(i)}\} - \tau\bdelta_{(i)} \\
     &=& -\frac{1}{n} \bfx_i \psi\{\tilde\epsilon_i + \bfx_i^\T \bdelta_0 - \bfx_i^\T \tilde\bdelta_i\} + \tau \boldeta_i \\
     &&+ \frac{1}{n} \sum_{j\neq i} \bfx_j \big[\psi_j\{\tilde\epsilon_j + \bfx_j^\T \bdelta_0 - \bfx_j^\T \hat\bdelta_{(i)}\} - \psi_j\{\tilde\epsilon_j + \bfx_j^\T \bdelta_0 - \bfx_j^\T (\hat \bdelta_{(i)} + \boldeta_i)\} \big] .
\end{eqnarray*}

By the mean value theorem, we have
\begin{eqnarray*}
    &&\psi_j\{\tilde\epsilon_j + \bfx_j^\T \bdelta_0 - \bfx_j^\T \hat\bdelta_{(i)}\} - \psi_j\{\tilde\epsilon_j + \bfx_j^\T \bdelta_0 - \bfx_j^\T (\hat \bdelta_{(i)} + \boldeta_i)\}   \\
&&= \psi_j'(\tilde r_{j, (i)})\bfx_j^\T \boldeta_i + [\psi_j'\{\gamma^\star (\bfx_j,\hat\bdelta_{(i)},\boldeta_i)\} - \psi_j'(\tilde r_{j, (i)})] \bfx_j^\T \boldeta_i,
\end{eqnarray*}
where $\gamma^\star (\bfx_j,\hat\bdelta_{(i)},\boldeta_i)$ is in the (``unordered'') interval $(\tilde r_{j, (i)}, \tilde r_{j, (i)}-\bfx_j^\T \boldeta_i)$.

Therefore, if $\cR_i$ is defined as above, we have
\begin{equation*}
    \begin{aligned}
    &\frac{1}{n} \sum_{j\neq i} \bfx_j \big[\psi_j\{\tilde\epsilon_j + \bfx_j^\T \bdelta_0 - \bfx_j^\T \hat\bdelta_{(i)}\} - \psi_j\{\tilde\epsilon_j + \bfx_j^\T \bdelta_0 - \bfx_j^\T (\hat \bdelta_{(i)} + \boldeta_i)\} \big] \\
    =& \frac{1}{n} \sum_{j \neq i} \psi_j' (\tilde r_{j,(i)}) \bfx_j \bfx_j^\T\boldeta_i + \cR_i \\
    =& \bfS_i \boldeta_i + \cR_i.
    \end{aligned}
\end{equation*}

The previous simplicities yield that
\begin{equation*}
    f(\tilde \bdelta_i) = -\frac{1}{n} \bfx_i \psi\{\tilde\epsilon_i + \bfx_i^\T \bdelta_0 - \bfx_i^\T \tilde\bdelta_i\} + (\bfS_i + \tau\I) \boldeta_i + \cR_i .
\end{equation*}
Since by definition, $\boldeta_i = n^{-1} (\bfS_i + \tau\I)^{-1} \bfx_i \psi \{\prox(c_i\rho)(\tilde r_{i,(i)})\}$, we have
\begin{equation*}
    (\bfS_i + \tau\I) \boldeta_i = \frac{1}{n} \bfx_i \psi \{\prox(c_i\rho)(\tilde r_{i,(i)})\}.
\end{equation*}
Also, by definition we have
\begin{equation*}
    \tilde\epsilon_i + \bfx_i^\T \bdelta_0 - \bfx_i^\T \tilde \bdelta_i = \tilde r_{i,(i)} - c_i \psi \{\prox(c_i\rho)(\tilde r_{i,(i)})\}.
\end{equation*}
When $\rho$ is differentiable, $x - c \psi\{\prox(c\rho)(x)\} = \prox(c\rho)(x)$.
Therefore, $\tilde\epsilon_i + \bfx_i^\T \bdelta_0 - \bfx_i^\T \tilde \bdelta_i = \prox(c_i\rho)(\tilde r_{i,(i)})$ and
\begin{equation*}
    -\frac{1}{n} \bfx_i \psi\{\tilde\epsilon_i + \bfx_i^\T \bdelta_0 - \bfx_i^\T \tilde\bdelta_i\} + (\bfS_i + \tau\I) \boldeta_i = - \frac{1}{n} \bfx_i \big[ \psi\{\prox(c_i\rho)(\tilde r_{i,(i)})\} - \psi(\tilde r_{i,(i)})\big] = 0.
\end{equation*}
Therefore, $f(\tilde \bdelta_i) = \cR_i$.
Applying Lemma \ref{lemma:delta_ineq_1} we have
\begin{equation*}
\|\hat\bdelta - \tilde\bdelta_i\| \leq \frac{1}{\tau} \|\cR_i\|.
\end{equation*}
\end{proof}

\paragraph{i. On $\cR_i$}~{}

Next, we provide a bound for $\cR_i$.
\begin{lemma}
We have
\begin{equation*}
\| \boldeta_i\| \leq \frac{1}{\sqrt{n}\tau} \frac{\|\bfx_i\|}{\sqrt{n}} | \psi(\tilde r_{i, (i)})|,
\end{equation*}
and
\begin{equation*}
\|\cR_i\| \leq \|\hat\bSigma\|_2 \sup_{j \neq i} \Big|\psi_j'\{\gamma^\star (\bfx_j,\hat\bdelta_{(i)},\boldeta_i)\} - \psi_j'(\tilde r_{j, (i)})\Big| \frac{1}{\sqrt{n}\tau} \frac{\|\bfx_i\|}{\sqrt{n}} | \psi(\tilde r_{i, (i)})|.
\end{equation*}

\end{lemma}

\begin{proof}
We have
\begin{equation*}
    \cR_i = \frac{1}{n} \sum_{j\neq i} [\psi_j'\{\gamma^\star (\bfx_j,\hat\bdelta_{(i)},\boldeta_i)\} - \psi_j'(\tilde r_{j, (i)})] \bfx_j \bfx_j^\T \boldeta_i.
\end{equation*}
Note that $\mathcal{S} = n^{-1} \sum_{j\neq i} [\psi_j'\{\gamma^\star (\bfx_j,\hat\bdelta_{(i)},\boldeta_i)\} - \psi_j'(\tilde r_{j, (i)})] \bfx_j \bfx_j^\T $ can be written as $\mathcal{S} = n^{-1} \bfX^\T \bfD \bfX$, where $\bfD$ is a diagonal matrix with $(j,j)$-entry $[\psi_j'\{\gamma^\star (\bfx_j,\hat\bdelta_{(i)},\boldeta_i)\} - \psi_j'(\tilde r_{j, (i)})]$.

Using the property of matrix norm $\|\cdot\|_2$, we have $\|\mathcal{S}\|_2 \leq \|\hat \bSigma\|_2 \|\bfD\|_2 $, which implies that
\begin{equation*}
    \|\cR_i\| \leq \|\hat\bSigma\|_2 \sup_{j \neq i} \Big|\psi_j'\{\gamma^\star (\bfx_j,\hat\bdelta_{(i)},\boldeta_i)\} - \psi_j'(\tilde r_{j, (i)})\Big| \| \boldeta_i\|,
\end{equation*}
where $\hat \bSigma = n^{-1} \sum_{j=1}^n \bfx_j \bfx_j^\T$ is the sample covariance matrix.

We now bound $\| \boldeta_i\|$. Note that
\begin{equation*}
    \|\boldeta_i\| \leq \frac{1}{\sqrt{n}\tau} \frac{\|\bfx_i\|}{\sqrt{n}} |\psi \{\prox(c_i\rho)(\tilde r_{i,(i)})\}|.
\end{equation*}
Using Lemma A-1 in \cite{karoui2013asymptotic}, we see that
\begin{equation*}
    |\psi(\operatorname{prox}_{c_i}(\rho)(\tilde{r}_{i,(i)}))| \leq|\psi(\tilde{r}_{i,(i)})|.
\end{equation*}
\end{proof}

Under our assumptions, we have
\begin{equation*}
    |\psi(\tilde{r}_{i,(i)})| \leq\|\psi\|_{\infty} \leq C \polyLog(n)
\end{equation*}
and later in the proof of Lemma \ref{lemma:3.7_2018} we will show that
\begin{equation*}
    \sup_i \frac{\|\bfx_i\|}{\sqrt{n}} = \sup_i \frac{|\lambda_i|\|\cX_i\|}{\sqrt{n}} = \mathrm{O}_{L_k}(\sup_i |\lambda_i|).
\end{equation*}

\paragraph{ii. On $\gamma^\star (\bfx_j,\hat\bdelta_{(i)},\boldeta_i)$ and related quantities}~{}

We now show how to control $n^{-1/2} \sup_{j \neq i} |\psi_j'\{\gamma^\star (\bfx_j,\hat\bdelta_{(i)},\boldeta_i)\} - \psi_j'(\tilde r_{j, (i)})|$.

\begin{lemma}
Suppose that $\psi'$ is $\mathrm{L}(n)$-Lipschitz. Then
\begin{equation*}
\sup_{j \neq i} \Big|\psi_j'\{\gamma^\star (\bfx_j,\hat\bdelta_{(i)},\boldeta_i)\} - \psi_j'(\tilde r_{j, (i)})\Big| \leq  \mathrm{L}(n) \sup_{j\neq i} |\bfx_j^\T \boldeta_i|.
\end{equation*}
\end{lemma}

\begin{proof}
    By definition, we have
    \begin{equation*}
        |\gamma^\star (\bfx_j,\hat\bdelta_{(i)},\boldeta_i) - \tilde r_{j, (i)}| \leq |\bfx_j^\T \boldeta_i|.
    \end{equation*}
    Therefore, the bound follows, using the fact that $\psi_j'$ is $\mathrm{L}(n)$-Lipschitz.
\end{proof}

\subsubsection{Stochastic aspects}
Recall that
\begin{equation*}
\bfx_j^\T \boldeta_i = \psi \{\prox(c_i\rho)(\tilde r_{i,(i)})\} \frac{1}{n} \bfx_j^\T(\bfS_i + \tau\I)^{-1} \bfx_i.
\end{equation*}
Therefore we can bound $\| \cR_i\|$ by
\begin{equation*}
\| \cR_i\| \leq \Big[ \sup_{j \neq i} \frac{|\bfx_j^\T(\bfS_i + \tau\I)^{-1} \bfx_i|}{n} \Big] \frac{ \mathrm{L}(n)}{\sqrt{n}\tau} \frac{\|\bfx_i\|}{\sqrt{n}} \|\hat\bSigma\|_2 \|\psi\|^2_\infty.
\end{equation*}

\paragraph{i. On $\sup_{j \neq i}|\bfx_j^\T(\bfS_i + \tau\I)^{-1} \bfx_i|$}~{}

Now we control $\bfx_j^\T(\bfS_i + \tau\I)^{-1} \bfx_i / n$.

\begin{lemma}
    \label{lemma:3.7_2018}
Suppose $\bfx_i$'s are independent and satisfy \textbf{O4}; suppose that $\lambda_i$'s satisfy \textbf{O6}.
Then
\begin{equation*}
\sup_{j \neq i} |\bfx_j^\T(\bfS_i + \tau\I)^{-1} \bfx_i/n | \leq \frac{1}{\sqrt{n}} \sup_{j \neq i} \frac{\|\mathcal{X}_j\|}{\tau \sqrt{n}} \polyLog(n)
\end{equation*}
in $L_{k}$, for any finite $k$.
Note that under Assumption \textbf{O4}, for any finite $k$,
\begin{equation*}
    \sup _{j \neq i}|\|\mathcal{X}_j\| / \sqrt{n}|=\mathrm{O}_{L_k}(1) .
\end{equation*}
\end{lemma}

\begin{proof}
    Note that
    \begin{equation*}
        |\bfx_j^\T(\bfS_i + \tau\I)^{-1} \bfx_i/n | = |\lambda_i \lambda_j| |\mathcal{X}_j^\T (\bfS_i + \tau\I)^{-1} \mathcal{X}_i/n |.
    \end{equation*}
    Denote $\bm{\cX}_{(i)} = (\cX_1, \ldots, \cX_{i-1}, \cX_{i+1}, \ldots, \cX_n)$,
    and $\bfv_{j,(i)} = (\bfS_i + \tau\I)^{-1} \cX_j$.
    Then we have the map $F_j(\cX_i) = \cX_j^\T (\bfS_i + \tau\I)^{-1} \cX_i = \cX_i^\T \bfv_{j,(i)}$ is linear in $\cX_i$ and it is Lipschitz with Lipschitz constant $\{\cX_j^\T (\bfS_i + \tau\I)^{-2} \cX_j\}^{1/2} \leq \|\cX_j\|/\tau$.
    Therefore, using Lemma Lemma B-2 in \cite{karoui2013asymptotic} and the fact that $\cX_i$ has mean $0$, we have
    \begin{equation*}
        \sup_{j \neq i} |\mathcal{X}_j^\T (\bfS_i + \tau\I)^{-1} \mathcal{X}_i/n | \big| \bm{\cX}_{(i)} \leq \frac{1}{\sqrt{n}} \sup_{j \neq i} \frac{\|\mathcal{X}_j\|}{\tau \sqrt{n}} \polyLog(n) / \mc_n^{1/2}
    \end{equation*}
    in $L_{k}$, when $\sup_{j \neq i} \|\mathcal{X}_j\| / (\tau n^{1/2}) = \mathrm{O}_{L_k}(1)$.

    To prove it, Using the fact that $\cX_j \to \|\cX_j\|/n^{1/2}$ is $n^{-1/2}$-Lipschitz, we have
    \begin{equation*}
        \sup_{j \neq i} | \|\cX_j\|/ \sqrt{n} - m_{\|\cX_j\|/\sqrt{n}}| \leq \polyLog(n) / \sqrt{n \mc_n} \ \text{ in } \sqrt{L_{2k}}.
    \end{equation*}
    Note that $\E (\|\cX_j\|) \leq \{\E (\|\cX_j\|^2)\}^{1/2} = p^{1/2}$, so $m_{\|\cX_j\|/\sqrt{n}}$ is of order $1$.
    Therefore, by Assumption \textbf{O4} on $\mc_n$ we have
    \begin{equation*}
        \sup_{j \neq i} | \|\cX_j\|/ \sqrt{n}| = \mathrm{O}_{L_k}(1).
    \end{equation*}

    Now our assumptions $\mathbf{O 6}$ concerning $\sup _i|\lambda_i|=\mathrm{O}_{L_k}(\polyLog(n))$ guarantee that the bounds we announced are valid.

\end{proof}

\paragraph{Consequences}~{}

 We have the following result.

\begin{proposition}
Under Assumptions \textbf{O1-O6}, we have
\begin{equation*}
\|\mathcal{R}_i\|=\mO_{L_k}\Big(\frac{[ \mathrm{L}(n)]\|\psi\|_{\infty}^2}{n \tau} \polyLog(n)\Big).
\end{equation*}
Furthermore, the same bound hold for $\sup _{1 \leq i \leq n}\|\mathcal{R}_i\|$.
\end{proposition}

\begin{proof}
    By aggregating all the intermediate results, using Holder's inequality and the fact that $\|\hat\bSigma\|_2 = \mathrm{O}_{L_k}(\polyLog(n))$ shown in Lemma 3.38 from \cite{el2018impact}, we finish the proof.
\end{proof}

We can now prove and state the following result.
Recall that
\begin{equation*}
\tilde \bdelta_i = \hat \bdelta_{(i)} + \frac{1}{n} (\bfS_i + \tau\I_p)^{-1} \bfx_i \psi \{\prox(c_i\rho)(\tilde r_{j,(i)})\} \triangleq \hat  \bdelta_{(i)} + \boldeta_i.
\end{equation*}

\begin{theorem}
    \label{th:l1oo}
Under Assumptions \textbf{O1-O7}, we have, for any fixed $k$, when $\tau$ is held fixed and $\mathrm{L}(n) \leq Cn^\alpha$,
\begin{equation*}
\sup_{1\leq i\leq n} \|\hat \bdelta - \tilde \bdelta_i\| = \mO_{L_k}\Big(\frac{\polyLog(n)}{n^{1-\alpha}}\Big).
\end{equation*}
In particular, we have
\begin{equation*}
\forall 1 \leq i \leq n, \E (\|\hat \bdelta - \tilde \bdelta_i\|^2) = \mO \Big(\frac{\polyLog(n)}{n^{2-2\alpha}} \Big).
\end{equation*}

Also,
\begin{equation*}
\sup_{1\leq i \leq n}\sup_{j \neq i} |\tilde r_{j,(i)} - R_j| = \mO_{L_k} \Big(\frac{\polyLog(n)}{n^{1/2-\alpha}} \Big).
\end{equation*}

Finally,
\begin{equation}
    \label{eq:el_27}
\sup_{1 \leq i \leq n} |R_i - \prox (c_i \rho)(\tilde r_{i,(i)})| = \mO_{L_k} \Big(\frac{\polyLog(n)}{n^{1/2-\alpha}} \Big).
\end{equation}
\end{theorem}

% We note that we could state a slightly finer result involving $\mathrm{L}(n)$ and various powers of $\left\|\psi\right\|_{\infty}$. However, we will not need such fine results in what follows, so we opt for slightly coarser but easier-to-state statements.

\begin{proof}
    The first two results are direct consequences of the previous propositions.

    The third result follows from that
    \begin{equation*}
        \begin{aligned}
            \sup_{j \neq i} |\tilde r_{j,(i)} - R_j| &= \sup_{j \neq i} |\bfx_j^\T (\hat \bdelta - \hat \bdelta_{(i)})| \leq \sup_{j \neq i} |\bfx_j^\T (\hat \bdelta - \tilde \bdelta_i)| + \sup_{j \neq i} |\bfx_j^\T (\tilde \bdelta_i - \hat \bdelta_{(i)})| \\
             &= \Big( \sup_{1 \leq j \leq n} \frac{\|\bfx_j\|}{\sqrt{n}} \Big) \sqrt{n} \|\hat \bdelta - \tilde \bdelta_i\| + \sup_{j \neq i} |\bfx_j^\T \boldeta_i|,
        \end{aligned}
    \end{equation*}
and the fact that $\sup_{1 \leq j \leq n} \|\bfx_j\| / n^{1/2} = \mO_{L_k}(\polyLog(n))$ under our assumptions.
Control of the first term follows from the results on $\|\hat \bdelta - \tilde \bdelta_i\|$.
Control of the second term follows from Lemma \ref{lemma:3.7_2018} and the assumption that $\psi$ is bounded by $C \polyLog(n)$.

For the last result, recall that
\begin{equation*}
    R_i = \tilde\epsilon_i +  \bfx_i^\T (\bdelta_0 -  \hat\bdelta) = \tilde\epsilon_i +  \bfx_i^\T \bdelta_0 -  \bfx_i^\T\tilde\bdelta_i - \bfx_i^\T (\hat\bdelta -  \tilde\bdelta_i).
\end{equation*}
Given the definition of $\tilde \bdelta_i$, we have
\begin{equation*}
    \bfx_i^\T \tilde \bdelta_i = \bfx_i^\T \hat \bdelta_{(i)} + c_i \psi \{\prox(c_i\rho)(\tilde r_{i,(i)})\}.
\end{equation*}
By the property of the proximal operator, if $y=\prox(c \rho)(x)$, $y+c \psi(y)=x$, we have
\begin{equation*}
    \tilde\epsilon_i +  \bfx_i^\T \bdelta_0 -  \bfx_i^\T\tilde\bdelta_i = \tilde{r}_{i,(i)}-c_i \psi[\prox(c_i \rho)(\tilde{r}_{i,(i)})]=\prox(c_i \rho)(\tilde{r}_{i,(i)}) .
\end{equation*}

Therefore, we have
\begin{equation*}
    \sup_i |R_i - \prox (c_i \rho)(\tilde r_{i,(i)})| = \sup_i |\bfx_i^\T (\hat\bdelta - \tilde\bdelta_i)|,
\end{equation*}
which is controlled in the previous results.
\end{proof}

\paragraph{On the limiting variance of $\|\hat \bdelta\|^2$ and $\|\hat \bdelta - \bdelta_0\|^2$}~{}

\begin{proposition}
    \label{prop:delta_variance}
Under assumptions \textbf{O1} - \textbf{O7},
\begin{equation*}
\var (\|\hat \bdelta\|^2) \to 0 \text{ as } n \to \infty.
\end{equation*}
Therefore $\|\hat \bdelta\|^2$ has a deterministic equivalent in probability and in $L_2$.

More precisely, we have
\begin{equation*}
\var (\|\hat \bdelta\|^2) = \mO \Big(\frac{\polyLog(n)}{n^{1-2\alpha}} \Big).
\end{equation*}
The same type of results are true for $\var (\|\hat \bdelta - \bdelta_0\|^2)$ and $\var (\|\hat \bbeta - \bbeta_0\|^2)$ provided that $\| \bdelta_0\| = \mO(\polyLog(n))$.
\end{proposition}

\begin{proof}
    We use the Efron-Stein inequality \citep{efron1981jackknife}: if $W$ is a function of $n$ independent random variables, and $W_{(i)}$ is any function of all those random variables except the $i$-th,
    $$
    \var(W) \leq \sum_{i=1}^n \var(W-W_{(i)}) \leq \sum_{i=1}^n \E((W-W_{(i)})^2) .
    $$

    We apply this inequality to $W=\|\hat \bdelta\|^2$ and $W_{(i)}=\|\hat \bdelta_{(i)}\|^2$.
    We first note that
    \begin{equation*}
        \E (| \|\hat \bdelta\|^2 - \|\hat \bdelta_{(i)}\|^2 |^2) = 2 \Big[\E \big(|\|\hat \bdelta\|^2 - \|\tilde \bdelta_i\|^2 |^2 \big) + \E \big(| \|\tilde \bdelta_i\|^2 - \|\hat \bdelta_{(i)}\|^2 |^2 \big) \Big].
    \end{equation*}
    For the first term, we have $|\|\hat \bdelta\|^2 - \|\tilde \bdelta_i\|^2 |^2 = [( \hat \bdelta - \tilde \bdelta_i)^\T ( \hat \bdelta + \tilde \bdelta_i)]^2$ and $( \hat \bdelta - \tilde \bdelta_i)^\T ( \hat \bdelta + \tilde \bdelta_i) = 2 ( \hat \bdelta - \tilde \bdelta_i)^\T \hat \bdelta - \|\hat \bdelta - \tilde \bdelta_i\|^2$.
    Therefore, by the Cauchy–Schwarz inequality we have
    \begin{equation*}
        |\|\hat \bdelta\|^2 - \|\tilde \bdelta_i\|^2 |^2 = \mO_{L_1} (\|\hat \bdelta - \tilde \bdelta_i\|^4) + \sqrt{\mO_{L_1} (\polyLog(n)) \|\hat \bdelta - \tilde \bdelta_i\|^4},
    \end{equation*}
    since $\E (\|\hat \bdelta\|^k)$ exists and is bounded by $k \polyLog(n) / \tau^k$.

    Using the results of Theorem \ref{th:l1oo} we have
    \begin{equation*}
        \E \big(|\|\hat \bdelta\|^2 - \|\tilde \bdelta_i\|^2 |^2 \big) = \mO \Big(\frac{\polyLog(n)}{n^{2-2\alpha}} \Big) = \mo(n^{-1}),
    \end{equation*}
    provided that $\alpha <1/2$.

    For the second term, by definition we have
    \begin{equation*}
        \begin{aligned}
            \|\tilde \bdelta_i\|^2 - \|\hat \bdelta_{(i)}\|^2 =& \frac{2}{n} \hat\bdelta_{(i)}^{\T}(\bfS_i+\tau \I)^{-1} \bfx_i \psi(\prox(c_i \rho)(\tilde{r}_{i,(i)})) \\
            &+ \frac{1}{n^2} \bfx_i^{\T}(\bfS_i+\tau \I)^{-2} \bfx_i \psi^2(\prox(c_i \rho)(\tilde{r}_{i,(i)})).
        \end{aligned}
    \end{equation*}
    Since $\|\hat \bdelta_{(i)}\|^2$ and $\bfS_i$ are independent of $\bfx_i$, and $\|(\bfS_i+\tau \I)^{-1}\|_2 \leq \tau^{-1}$, we have $\hat\bdelta_{(i)}^{\T}(\bfS_i+\tau \I)^{-1} \bfx_i = \mO_{L_2} (|\lambda_i| \|\hat \bdelta_{(i)}\| / \mc_n^{1/2})$.
    Recall also that $\sup_i \|\psi\|_{\infty} = \mO(\polyLog(n))$.
    Therefore, both terms are of order $\mO_{L_2}(\polyLog(n) / n \mc_n^{1/2})$.

    We can now conclude that
    \begin{equation*}
        \E \big(| \|\tilde \bdelta_i\|^2 - \|\hat \bdelta_{(i)}\|^2 |^2 \big) = \mO \Big(\frac{\polyLog(n)}{n^{2}} \Big).
    \end{equation*}
    Therefore, we have
    \begin{equation*}
        \var (\|\hat \bdelta\|^2) = \mO \Big(\frac{\polyLog(n)}{n^{1-2\alpha}} \Big) = \mo(1).
    \end{equation*}
    This shows that $\|\hat \bdelta\|^2$ has a deterministic equivalent in probability and in $L_2$.

    For the second part of the result, similarly we write that
    \begin{equation*}
        \E (| \|\hat \bdelta - \bdelta_0\|^2 - \|\hat \bdelta_{(i)} - \bdelta_0\|^2 |^2) = 2 \Big[\E \big(|\|\hat \bdelta - \bdelta_0\|^2 - \|\tilde \bdelta_i - \bdelta_0\|^2 |^2 \big) + \E \big(| \|\tilde \bdelta_i - \bdelta_0\|^2 - \|\hat \bdelta_{(i)} - \bdelta_0\|^2 |^2 \big) \Big].
    \end{equation*}
    Using the fact that $|\|\hat \bdelta - \bdelta_0\|^2 - \|\tilde \bdelta_i - \bdelta_0\|^2 |^2 = [( \hat \bdelta - \tilde \bdelta_i)^\T ( \hat \bdelta + \tilde \bdelta_i - 2\bdelta_0)]^2$ and $( \hat \bdelta - \tilde \bdelta_i)^\T ( \hat \bdelta + \tilde \bdelta_i - 2\bdelta_0) = 2 ( \hat \bdelta - \tilde \bdelta_i)^\T (\hat \bdelta - \bdelta_0) - \|\hat \bdelta - \tilde \bdelta_i\|^2$, by the Cauchy–Schwarz inequality we have
    \begin{equation*}
        |\|\hat \bdelta - \bdelta_0\|^2 - \|\tilde \bdelta_i - \bdelta_0\|^2 |^2 = \mO_{L_1} (\|\hat \bdelta - \tilde \bdelta_i\|^4) + \sqrt{\mO_{L_1} (\polyLog(n)) \|\hat \bdelta - \tilde \bdelta_i\|^4},
    \end{equation*}
    since $\E (\|\hat \bdelta - \bdelta_0\|^k)$ exists and is bounded by $k \polyLog(n) / \tau^k$ following from assumption \textbf{O7} and Lemma \ref{lemma:delta_moment_bounds}.

    Using the results of Theorem \ref{th:l1oo} we have
    \begin{equation*}
        \E \big(|\|\hat \bdelta - \bdelta_0\|^2 - \|\tilde \bdelta_i - \bdelta_0\|^2 |^2 \big) = \mO \Big(\frac{\polyLog(n)}{n^{2-2\alpha}} \Big) = \mo(n^{-1}),
    \end{equation*}
    provided that $\alpha <1/2$.

    Similarly, by definition we have
    \begin{equation*}
        \begin{aligned}
            \|\tilde \bdelta_i - \bdelta_0\|^2 - \|\hat \bdelta_{(i)} - \bdelta_0\|^2 =& \frac{2}{n} (\hat\bdelta_{(i)} - \bdelta_0)^{\T}(\bfS_i+\tau \I)^{-1} \bfx_i \psi(\prox(c_i \rho)(\tilde{r}_{i,(i)})) \\
            &+ \frac{1}{n^2} \bfx_i^{\T}(\bfS_i+\tau \I)^{-2} \bfx_i \psi^2(\prox(c_i \rho)(\tilde{r}_{i,(i)})).
        \end{aligned}
    \end{equation*}
    Since $(\hat \bdelta_{(i)} - \bdelta_0)^{\T}(\bfS_i+\tau \I)^{-1} \bfx_i = \mO_{L_2} (|\lambda_i| \|\hat \bdelta_{(i)} - \bdelta_0\| / \mc_n^{1/2})$, both terms are of order $\mO_{L_2}(\polyLog(n) / n \mc_n^{1/2})$.

    Similarly, we have
    \begin{equation*}
        \var (\|\hat \bdelta - \bdelta_0\|^2) = \mO \Big(\frac{\polyLog(n)}{n^{1-2\alpha}} \Big) = \mo(1).
    \end{equation*}

    The results for $\var (\|\hat \bbeta - \bbeta_0\|^2)$ are simply followed by the fact that $\var (\|\hat \bfw - \bfw_0\|^2) = \mo(1)$ in Proposition 3.10 of \cite{el2018impact}.

    By assuming $\tilde W = \|\hat \bdelta + \hat \bfw - \bdelta_0 - \bfw_0\|^2$ and $\tilde W_{(i)} = \|\hat \bdelta_{(i)} + \hat \bfw - \bdelta_0 - \bfw_0\|^2$, we can similarly have
    \begin{equation*}
        \E \big(| \|\hat \bdelta + \hat \bfw - \bdelta_0 - \bfw_0\|^2 - \|\hat \bdelta_{(i)} + \hat \bfw - \bdelta_0 - \bfw_0\|^2 |^2 \big) = \mo (1).
    \end{equation*}
    Note that $\E (\|\hat \bfw - \bfw_0\|) = \mO(1)$, we have $\E (\|\hat \bdelta + \hat \bfw - \bdelta_0 - \bfw_0\| )$ is bounded by $K \operatorname{polyLog}(n) / \tau^k$.

\end{proof}

\subsection[\appendixname~\thesubsection]{Leaving out a predictor} \label{sec_p2}

% {\color{red} For simplicity, we call $n$ the sample size of target data instead of $n$.
% Also for $\bfX$, $\bfx$, $epsilon$, etc.}

Let $\bfV$ be the $n \times(p-1)$ matrix corresponding to the first $(p-1)$ columns of the design matrix $\bfX$. We call $\bfv_i$ in $\RR^{p-1}$ the vector corresponding to the first $p-1$ entries of $\bfx_i$, i.e $\bfv_i^{\T}=(x_i(1), \ldots, x_i(p-1))$.
Denote $\cX_i = (\cX_i(1), \ldots, \cX_i(p))^\T$.
We call $\bfX(p)$ the vector in $\RR^n$ with $j$-th entry $x_j(p)$, i.e the $p$-th entry of the vector $\bfx_j$. When this does not create problems, we also use the standard notation $x_{j, p}$ for $x_j(p)$.
We further denote $\bdelta_0 = (\bgamma_0^\T, \delta_0(p))^\T$.

Call $\hat \bgamma$ the solution of
\begin{equation}
    \label{eq:def_gamma_hat}
\hat{\bgamma}=\argmin_{\bgamma \in \mathbb{R}^{p-1}} \frac{1}{n} \sum_{i=1}^n \rho\{\epsilon_i + \bfv_i^{\T}(\bfw_0 - \hat\bfw) -\bfv_i^{\T}(\bgamma-\bgamma_0)\} + \frac{\tau}{2}\|\bgamma\|^2.
\end{equation}
Note that $\Big[\begin{array}{l}\widehat{\gamma} \\ 0\end{array}\Big]$ is the solution of the original optimization problem \eqref{eq:opt_ori} when $x_i(p)$ is replaced by $0$.

% \subsubsection{Approximation to \texorpdfstring{\(\hat \bdelta\)}~{} via leave-one-predictor-out}

\paragraph{Approximation to \(\hat \bdelta\) via leave-one-predictor-out}~{}

We use the notations and partitions
\begin{equation}
    \bfx_i=\left[\begin{array}{c}
    \bfv_i \\
    x_i(p)
    \end{array}\right], \ \hat{\bdelta}=\left[\begin{array}{c}
        \hat{\bdelta}_{-p} \\
    \hat{\delta}(p)
    \end{array}\right] .
\end{equation}
Naturally, $\hat \bgamma$ satisfies
\begin{equation}
    \label{eq:gamma_hat}
    -\frac{1}{n} \sum_{i=1}^n \bfv_i \psi\{\epsilon_i + \bfv_i^{\T}(\bfw_{0,-p} - \hat\bfw_{-p})  -\bfv_i^{\T}(\hat\bgamma-\bgamma_0)\} + \tau \hat \bgamma = \bzero_{p-1}.
\end{equation}
Denote
\begin{equation*}
    r_{i,[p]}=\epsilon_i + \bfv_i^{\T}(\bfw_{0,-p} - \hat\bfw_{-p}) -\bfv_i^{\T}(\hat\bgamma - \bgamma_0).
\end{equation*}
i.e. the residuals based on $p-1$ predictors.

Recall that
\begin{equation}
    \label{eq:delta_hat}
    -\frac{1}{n} \sum_{i=1}^n \bfx_i \psi\{\epsilon_i + \bfx_i^{\T}(\bfw_0 - \hat\bfw) -\bfx_i^{\T}(\hat\bdelta-\bdelta_0)\} + \tau\hat\bdelta = \bzero_p,
\end{equation}
and
\begin{equation*}
    R_i = \epsilon_i + \bfx_i^\T (\bfw_0 - \hat \bfw) -  \bfx_i^\T (\hat\bdelta - \bdelta_0).
\end{equation*}

Taking the difference between \eqref{eq:gamma_hat} and \eqref{eq:delta_hat}, we have
\begin{equation*}
    \frac{1}{n} \sum_i^n \Big\{ \bfx_i \psi(R_i)-\Big[\begin{array}{c}
        \bfv_i \\
        0
        \end{array}\Big] \psi(r_{i,[p]}) \Big\} - \tau\Big(\hat \bdelta - \Big[\begin{array}{c}
            \hat \bgamma \\
            0
            \end{array}\Big] \Big) = \bzero_p.
\end{equation*}
Note that this $p$-dimensional equation can separate into a scalar and a vector equation, namely,
\begin{equation*}
    \begin{aligned}
        &\frac{1}{n} \sum_i^n \Big\{ x_i(p) \psi(R_i) \Big\} - \tau \hat \delta(p) = 0, \\
        &\frac{1}{n} \sum_i^n \bfv_i \Big\{ \psi(R_i)- \psi(r_{i,[p]}) \Big\} - \tau (\hat \bdelta_{-p} - \hat \bgamma) = \bzero_{p-1}.
    \end{aligned}
\end{equation*}
Using a first-order Taylor expansion of $\psi(R_i)$ around $\psi(r_{i,[p]})$ and noting that $R_i - r_{i,[p]} = \bfv_i^{\T}(\hat \bgamma - \hat \bdelta_{-p}) + x_i(p) \{w_0(p) - \hat w(p) + \delta_0(p) - \hat \delta(p)\}$, we can transform the first equation above into
\begin{equation*}
    \frac{1}{n} \sum_i^n x_i(p) \Big( \psi(r_{i,[p]}) + \psi'(r_{i,[p]}) \big[ \bfv_i^{\T}(\hat \bgamma - \hat \bdelta_{-p}) + x_i(p) \{w_0(p) - \hat w(p) + \delta_0(p) - \hat \delta(p)\} \big] \Big) - \tau \hat \delta(p) \simeq 0.
\end{equation*}
This gives the near equivalence
\begin{equation*}
    \hat \delta(p) \simeq \frac{\frac{1}{n} \sum x_i(p) \Big( \psi(r_{i,[p]}) + \psi'(r_{i,[p]}) \big[ \bfv_i^{\T}(\hat \bgamma - \hat \bdelta_{-p}) + x_i(p) \{w_0(p) - \hat w(p) + \delta_0(p)\} \big] \Big)}{ \frac{1}{n} \sum x_i^2(p) \psi'(r_{i,[p]}) + \tau}.
\end{equation*}

Working similarly on the second equation involving $\bfv_i$, we have
\begin{equation*}
    \frac{1}{n} \sum_i^n \psi'(r_{i,[p]}) \bfv_i (R_i - r_{i,[p]}) - \tau (\hat \bdelta_{-p} - \hat \bgamma) \simeq \bzero_{p-1}.
\end{equation*}
Since $R_i - r_{i,[p]} = \bfv_i^{\T}(\hat \bgamma - \hat \bdelta_{-p}) + x_i(p) \{w_0(p) - \hat w(p) + \delta_0(p) - \hat \delta(p)\}$, the above equation can be transformed into
\begin{equation*}
    \Big[ \frac{1}{n} \sum_i^n \psi'(r_{i,[p]}) \bfv_i \bfv_i^{\T} \Big] (\hat \bgamma - \hat \bdelta_{-p}) + \{w_0(p) - \hat w(p) + \delta_0(p) - \hat \delta(p)\} \frac{1}{n} \sum_i^n \psi'(r_{i,[p]}) \bfv_i x_i(p) - \tau (\hat \bdelta_{-p} - \hat \bgamma)
    \simeq \bzero_{p-1}.
\end{equation*}
Denote
\begin{equation*}
    \bfu_p=\frac{1}{n} \sum_{i=1}^n \psi' (r_{i,[p]}) \bfv_i x_i(p), \text { and } \mathfrak{S}_p=\frac{1}{n} \sum_{i=1}^n \psi' (r_{i,[p]}) \bfv_i \bfv_i^{\T},
\end{equation*}
we see that $\hat \bgamma - \hat \bdelta_{-p} \simeq - \{w_0(p) - \hat w(p) + \delta_0(p) - \hat \delta(p)\} (\mathfrak{S}_p + \tau \I)^{-1} \bfu_p$.
Using the above approximation in the equation for $\hat \delta(p)$, we can write

\begin{equation*}
    \hat \delta(p) \simeq \frac{\frac{1}{n} \sum x_i(p) \psi(r_{i,[p]}) + \{w_0(p) - \hat w(p) + \delta_0(p)\} \{\frac{1}{n} \sum x_i^2(p) \psi'(r_{i,[p]}) - \bfu_p^{\T}(\mathfrak{S}_p+\tau \I)^{-1}\} \bfu_p }{ \frac{1}{n} \sum x_i^2(p) \psi'(r_{i,[p]}) - \bfu_p^{\T}(\mathfrak{S}_p+\tau \I)^{-1} \bfu_p + \tau}.
\end{equation*}

Denote
\begin{equation*}
    \xi_n \triangleq \frac{1}{n} \sum_{i=1}^n x_i^2(p) \psi'(r_{i,[p]}) - \bfu_p^{\T}(\mathfrak{S}_p+\tau \I)^{-1} \bfu_p,
\end{equation*}
and
\begin{equation*}
    N_p \triangleq \frac{1}{\sqrt{n}} \sum_{i=1}^n x_i(p) \psi(r_{i,[p]}) .
\end{equation*}
We have
\begin{equation*}
    (\xi_n + \tau) \hat \delta(p) \simeq \sqrt{n} N_p + \{w_0(p) - \hat w(p) + \delta_0(p)\} \xi_n.
\end{equation*}
Also we have
\begin{equation*}
    \hat \bdelta_{-p} \simeq \hat \bgamma + \{w_0(p) - \hat w(p) + \delta_0(p) - \hat\delta(p)\} (\mathfrak{S}_p + \tau \I)^{-1} \bfu_p.
\end{equation*}

Thus, we construct a approximation to $\hat \bdelta$.
As a summary, we introduce the following definitions:

\textbf{Definition.} We call the residuals corresponding to this optimization problem $\{r_{i,[p]}\}_{i=1}^n$, in other words
\begin{equation*}
r_{i,[p]}=\epsilon_i + \bfv_i^{\T}(\bfw_{0,-p} - \hat\bfw_{-p})  -\bfv_i^{\T}(\hat\bgamma - \bgamma_0).
\end{equation*}

We call
\begin{equation*}
\bfu_p=\frac{1}{n} \sum_{i=1}^n \psi' (r_{i,[p]}) \bfv_i x_i(p), \text { and } \mathfrak{S}_p=\frac{1}{n} \sum_{i=1}^n \psi' (r_{i,[p]}) \bfv_i \bfv_i^{\T}.
\end{equation*}

Note that $\bfu_p \in \RR^{p-1}$ and $\mathfrak{S}_p$ is $(p-1) \times(p-1)$. We call
\begin{equation*}
\xi_n \triangleq \frac{1}{n} \sum_{i=1}^n x_i^2(p) \psi'(r_{i,[p]}) - \bfu_p^{\T}(\mathfrak{S}_p+\tau \I)^{-1} \bfu_p,
\end{equation*}
and
\begin{equation*}
N_p \triangleq \frac{1}{\sqrt{n}} \sum_{i=1}^n x_i(p) \psi(r_{i,[p]}) .
\end{equation*}

We consider
\begin{equation*}
\mathfrak{b}_p \triangleq \{w_0(p) - \hat w(p) + \delta_0(p)\} \frac{\xi_n}{\tau+\xi_n}+\frac{1}{\sqrt{n}} \frac{N_p}{\tau+\xi_n} .
\end{equation*}

Note that when $\xi_n>0$, we have
$$
\mathfrak{b}_p-\delta_0(p) = w_0(p) - \hat w(p) + \frac{n^{-1 / 2} N_p-\tau \mathfrak{b}_p}{\xi_n} .
$$

We call
$$
\tilde{\bfb}=\left[\begin{array}{c}
\widehat{\bgamma} \\
\delta_0(p) - \hat w(p) + w_0(p)
\end{array}\right] + \{\mathfrak{b}_p - \delta_0(p) + \hat w(p) - w_0(p)\} \left[\begin{array}{c}
-\left(\mathfrak{S}_p+\tau \I\right)^{-1} \bfu_p \\
1
\end{array}\right] .
$$

\subsubsection{Deterministic aspects}
\begin{proposition}
We have
\begin{eqnarray}
\|\hat\bdelta-\tilde\bfb\| \leq \frac{1}{\tau}|\mathfrak{b}_p - \delta_0(p) + \hat w(p) - w_0(p)| \sup _{1 \leq i \leq n}|\mathrm{d}_{i, p}| \|\widehat{\Sigma}\|_2 \sqrt{\|(\mathfrak{S}_p+\tau \I)^{-1} \bfu_p\|^2+1} \label{eqn:prop311_i}
\end{eqnarray}
where $\mathrm{d}_{i, p}=[\psi' (\gamma_{i, p}^*)-\psi'(r_{i,[p]})]$ and $\gamma_{i, p}^*$ is in the interval $(\epsilon_i + \bfv_i^{\T}(\bfw_{0,-p} - \hat\bfw_{-p})  -\bfv_i^{\T}(\hat\bgamma-\bgamma_0), \epsilon_i + \bfx_i^{\T}(\bfw_0 - \hat\bfw) -\bfx_i^{\T}(\tilde\bfb-\bbeta_0))$. Furthermore,
\begin{equation}
    \label{eqn:prop311_iii}
\|(\mathfrak{S}_p+\tau \I)^{-1} u_p\|^2 \leq \frac{1}{n \tau} \sum_{i=1}^n x_i^2(p) \psi'(r_{i,[p]})=\frac{1}{n \tau} \sum_{i=1}^n \lambda_i^2 \psi'(r_{i,[p]}) \mathcal{X}_i^2(p).
\end{equation}
\end{proposition}

As in Lemma \ref{lemma:delta_ineq_1}, we have
\begin{equation*}
\|\hat{\bdelta}-\tilde{\bfb}\| \leq \frac{1}{\tau}\|f(\tilde{\bfb})\|,
\end{equation*}
where
\begin{equation}
f(\tilde{\bfb})= -\frac{1}{n} \sum_{i=1}^n \bfx_i \psi\{\epsilon_i + \bfx_i^{\T}(\bfw_0 - \hat\bfw) -\bfx_i^{\T}(\tilde\bfb-\bdelta_0)\} + \tau\tilde\bfb.
\end{equation}
We note furthermore that, by deﬁnition of $\hat \bgamma$,
\begin{equation*}
g(\hat{\bgamma})= -\frac{1}{n} \sum_{i=1}^n \bfv_i \psi\{\epsilon_i + \bfv_i^{\T}(\bfw_{0,-p} - \hat\bfw_{-p})  -\bfv_i^{\T}(\hat\bgamma-\bgamma_0)\} + \tau \hat \bgamma = \bzero_{p-1}.
\end{equation*}

\begin{proof}
    \textbf{i. Work on the first $p-1$ coordinates of $f(\tilde{\bfb})$}

    Denote $f_{p-1} (\bdelta)$ the first $p-1$ coordinates of $f(\bdelta)$.
    Denote $\hat \bgamma_{ext}$ the $p$-dimensional vector whose first $p-1$ coordinates are $\hat \bgamma$ and last coordinate is $\bdelta_0(p)$, i.e.
    \begin{equation*}
        \hat \bgamma_{ext} = \Big[\begin{array}{c}
            \hat \bgamma \\
            \bdelta_0(p) - \hat w(p) + w_0(p)
        \end{array}\Big].
    \end{equation*}

    For a vector $\bfv$, we use the notation $\bfv_{-k}$ to denote the vector obtained by removing the $k$-th coordinate of $\bfv$.

    Note that
    \begin{eqnarray*}
        f_{p-1}(\tilde\bfb) &=& f_{p-1}(\tilde{\bfb}) - g(\hat \bgamma) \\
        &=& -\frac{1}{n} \sum_{i=1}^n \bfv_i \Big[ \psi\{ \epsilon_i + \bfx_i^\T (\bfw_0 - \hat \bfw) + \bfx_i^\T (\bdelta_0 - \tilde \bfb) \} - \psi\{ \epsilon_i + \bfv_i^\T (\bfw_{0,-p} - \hat\bfw_{-p})  + \bfv_i^\T (\bgamma_0 - \hat \bgamma) \} \Big]\\
        && + \tau (\tilde \bfb_{-p} - \hat \bgamma).
    \end{eqnarray*}
    By the mean value theorem, for $\gamma_{i, p}^*$ in the interval $(\epsilon_i + \bfv_i^{\T}(\bfw_0 - \hat\bfw) -\bfv_i^{\T}(\bfw_{0,-p} - \hat\bfw_{-p}) , \epsilon_i + \bfx_i^{\T}(\bfw_0 - \hat\bfw) -\bfx_i^{\T}(\tilde\bfb-\bdelta_0))$, we have
    \begin{eqnarray*}
        &&\psi\{ \epsilon_i + \bfx_i^\T (\bfw_0 - \hat \bfw) + \bfx_i^\T (\bdelta_0 - \tilde \bfb) \} - \psi\{ \epsilon_i + \bfv_i^\T (\bfw_{0,-p} - \hat\bfw_{-p})  + \bfv_i^\T (\bgamma_0 - \hat \bgamma) \} \\
        &=& \psi'(\gamma_{i, p}^*) \bfx_i^\T (\hat \bgamma_{ext} - \tilde \bfb) \\
        &=& \psi'(r_{i,[p]}) \bfx_i^\T(\hat \bgamma_{ext} - \tilde \bfb) + \{\psi' (\gamma_{i, p}^*)-\psi'(r_{i,[p]})\} \bfx_i^\T(\hat \bgamma_{ext} - \tilde \bfb).
    \end{eqnarray*}
    Denote
    \begin{equation*}
        \begin{aligned}
            \md_{i,p} &= \psi' (\gamma_{i, p}^*)-\psi'(r_{i,[p]}), \\
            \delta_{i,p} &= \{\psi' (\gamma_{i, p}^*)-\psi'(r_{i,[p]})\} \bfx_i^\T(\hat \bgamma_{ext} - \tilde \bfb),\\
            \mathrm{R}_p &= -\frac{1}{n} \sum_{i=1}^n \bfv_i \Big[ \{\psi' (\gamma_{i, p}^*)-\psi'(r_{i,[p]})\} \bfx_i^\T(\hat \bgamma_{ext} - \tilde \bfb) \Big].
        \end{aligned}
    \end{equation*}

    Therefore, we have
    \begin{equation*}
        f_{p-1}(\tilde\bfb) = -\frac{1}{n} \sum_{i=1}^n \psi'(r_{i,[p]}) \bfv_i  \bfx_i^\T(\hat \bgamma_{ext} - \tilde \bfb) + \tau (\tilde \bfb_{-p} - \hat \bgamma) + \mathrm{R}_p \triangleq \mathrm{A}_p+\mathrm{R}_p.
    \end{equation*}

    Note that by definition,
    \begin{equation*}
        \begin{aligned}
            \hat{\bgamma}_{ext}-\tilde{\bfb} & =\{\mathfrak{b}_p - \delta_0(p) + \hat w(p) - w_0(p)\}\left[\begin{array}{c}
            (\mathfrak{S}_p+\tau \I)^{-1} \bfu_p \\
            -1
            \end{array}\right], \\
            \tilde{\bfb}_{-p}-\hat{\bgamma} & =-\{\mathfrak{b}_p - \delta_0(p) + \hat w(p) - w_0(p)\}(\mathfrak{S}_p+\tau \I)^{-1} \bfu_p .
            \end{aligned}
    \end{equation*}
    Therefore, we have $\bfx_i^\T (\hat \bgamma_{ext} - \tilde \bfb)= \{\mathfrak{b}_p - \delta_0(p) + \hat w(p) - w_0(p)\} \{ \bfv_i^\T (\mathfrak{S}_p + \tau \I)^{-1} \bfu_p - x_i(p) \}$, and
    
    \begin{equation*}
        \mathrm{A}_p = -\{\mathfrak{b}_p - \delta_0(p) + \hat w(p) - w_0(p)\} \Big[ \frac{1}{n} \sum_{i=1}^{n} \psi'(r_{i,[p]}) \bfv_i \big\{ \bfv_i^\T (\mathfrak{S}_p + \tau \I)^{-1} \bfu_p - x_i(p) \big\} + \tau (\mathfrak{S}_p+\tau \I)^{-1} \bfu_p \Big].
    \end{equation*}

    By the definition of $\mathfrak{S}_p$ and $\bfu_p$, we have
    \begin{equation*}
        \mathrm{A}_p = -\{\mathfrak{b}_p - \delta_0(p) + \hat w(p) - w_0(p)\} \{ \mathfrak{S}_p (\mathfrak{S}_p + \tau \I)^{-1} \bfu_p - \bfu_p + \tau (\mathfrak{S}_p+\tau \I)^{-1} \bfu_p \} = \bzero_{p-1},
    \end{equation*}
    since $\mathfrak{S}_p (\mathfrak{S}_p + \tau \I)^{-1}  + \tau (\mathfrak{S}_p+\tau \I)^{-1} = \I$.

    Therefore, we conclude that
    \begin{equation*}
        f_{p-1}(\tilde\bfb) = \mathrm{R}_p.
    \end{equation*}

    \textbf{ii. Work on the last coordinate of $f(\tilde{\bfb})$}

    Denote $[f (\tilde\bdelta)]_p$ the last coordinate of $f(\tilde\bdelta)$.
    We have shown that
    \begin{eqnarray*}
        &&\psi\{ \epsilon_i + \bfx_i^\T (\bfw_0 - \hat \bfw) + \bfx_i^\T (\bdelta_0 - \tilde \bfb) \} - \psi\{ \epsilon_i + \bfv_i^\T (\bfw_{0,-p} - \hat\bfw_{-p}) + \bfv_i^\T (\bgamma_0 - \hat \bgamma) \} \\
        &=& \psi'(r_{i,[p]}) \bfx_i^\T(\hat \bgamma_{ext} - \tilde \bfb) + \{\psi' (\gamma_{i, p}^*)-\psi'(r_{i,[p]})\} \bfx_i^\T(\hat \bgamma_{ext} - \tilde \bfb) .
    \end{eqnarray*}
    Recall that
    \begin{equation*}
        \begin{aligned}
            r_{i,[p]} &= \epsilon_i + \bfv_i^{\T}(\bfw_{0,-p} - \hat\bfw_{-p}) -\bfv_i^{\T}(\hat\bgamma-\bgamma_0), \\
            \delta_{i,p} &= \{\psi' (\gamma_{i, p}^*)-\psi'(r_{i,[p]})\} \bfx_i^\T(\hat \bgamma_{ext} - \tilde \bfb) .
        \end{aligned}
    \end{equation*}
    We note that
    \begin{eqnarray*}
        &&\psi\{ \epsilon_i + \bfx_i^\T (\bfw_0 - \hat \bfw) + \bfx_i^\T (\bdelta_0 - \tilde \bfb) \} \\
        &=& \psi(r_{i,[p]}) + \psi'(r_{i,[p]}) \bfx_i^\T(\hat \bgamma_{ext} - \tilde \bfb) + \delta_{i,p} \\
        &=& \psi(r_{i,[p]}) + \psi'(r_{i,[p]}) \{\mathfrak{b}_p - \delta_0(p) + \hat w(p) - w_0(p)\} \{ \bfv_i^\T (\mathfrak{S}_p + \tau \I)^{-1} \bfu_p - x_i(p) \} + \delta_{i,p}.
    \end{eqnarray*}
    Therefore, we have
    \begin{eqnarray*}
            &&[f (\tilde\bdelta)]_p + \frac{1}{n} \sum_{i=1}^n x_i(p) \delta_{i,p} \\
            &=& -\frac{1}{n} \sum_{i=1}^n x_i(p) \Big[\psi(r_{i,[p]}) + \psi'(r_{i,[p]}) \{\mathfrak{b}_p - \delta_0(p) + \hat w(p) - w_0(p)\} \{ \bfv_i^\T (\mathfrak{S}_p + \tau \I)^{-1} \bfu_p - x_i(p) \} \Big] + \tau \tilde b(p), \\
            &=& -\frac{1}{n} \sum_{i=1}^n x_i(p) \psi(r_{i,[p]}) - \{\mathfrak{b}_p - \delta_0(p) + \hat w(p) - w_0(p)\} \bfu_p^\T (\mathfrak{S}_p + \tau \I)^{-1} \bfu_p \\
            &&+ \{\mathfrak{b}_p - \delta_0(p) + \hat w(p) - w_0(p)\} \frac{1}{n} \sum_{i=1}^n \psi'(r_{i,[p]}) x_i^2(p) + \tau \mathfrak{b}_p, \\
            &=& - \Big\{ \frac{1}{n} \sum_{i=1}^n x_i(p) \psi(r_{i,[p]}) - \tau \mathfrak{b}_p \Big\} \\
            &&+ \{\mathfrak{b}_p - \delta_0(p) + \hat w(p) - w_0(p)\} \Big\{ \frac{1}{n} \sum_{i=1}^n \psi'(r_{i,[p]}) x_i^2(p) - \bfu_p^\T (\mathfrak{S}_p + \tau \I)^{-1} \bfu_p \Big\}, \\
            &=& - \Big( \frac{1}{\sqrt{n}} N_p - \tau \mathfrak{b}_p \Big) + \{\mathfrak{b}_p - \delta_0(p) + \hat w(p) - w_0(p)\} \xi_n, \\
            &=& 0.
    \end{eqnarray*}

    We conclude that
    \begin{eqnarray*}
        [f (\tilde\bdelta)]_p &=& -\frac{1}{n} \sum_{i=1}^n x_i(p) \delta_{i,p} \\
        &=& -\frac{1}{n} \sum_{i=1}^n x_i(p) \{\psi' (\gamma_{i, p}^*)-\psi'(r_{i,[p]})\} \bfx_i^\T(\hat \bgamma_{ext} - \tilde \bfb).
    \end{eqnarray*}

    \textbf{iii. Representation of $f(\tilde\bfb)$}

    Aggregating all the results we have obtained so far, we see that
    \begin{equation}
        \label{eq:f_tilde_b_iii}
        \begin{aligned}
            f(\tilde\bfb) =& - \frac{1}{n} \sum_{i=1}^n \mathrm{d}_{i,p} \bfx_i \bfx_i^\T(\hat \bgamma_{ext} - \tilde \bfb) \\
            =& \{\mathfrak{b}_p - \delta_0(p) + \hat w(p) - w_0(p)\} \Big\{\frac{1}{n} \sum_{i=1}^n \mathrm{d}_{i, p} \bfx_i \bfx_i^{\T}\Big\} \left[\begin{array}{c}
                (\mathfrak{S}_p+\tau \I)^{-1} \bfu_p \\
                -1
                \end{array}\right] ,
        \end{aligned}
    \end{equation}
    which implies \eqref{eqn:prop311_i}.

    % For the first term we have
    % \begin{equation*}
    %     \frac{1}{n} \sum_{i=1}^n \mathrm{d}_{i,p} \bfx_i \bfx_i^\T(\hat \bgamma_{ext} - \tilde \bfb) = \{\mathfrak{b}_p - \delta_0(p) + \hat w(p) - w_0(p)\} \Big\{\frac{1}{n} \sum_{i=1}^n \mathrm{d}_{i, p} \bfx_i \bfx_i^{\T}\Big\} \left[\begin{array}{c}
    %         (\mathfrak{S}_p+\tau \I)^{-1} \bfu_p \\
    %         -1
    %         \end{array}\right] ,
    % \end{equation*}
    % which implies \eqref{eqn:prop311_i}.

    For $\|(\mathfrak{S}_p+\tau \I)^{-1} \bfu_p\|^2$, denote $\bfD_{\psi'(r_{\cdot,[p]})}$ the diagonal matrix with $(i,i)$ entry $\psi'(r_{i,[p]})$. We have
    \begin{equation*}
        \bfu_p = \frac{1}{n} \bfV^\T \bfD_{\psi'(r_{\cdot,[p]})} \bfX(p) .
    \end{equation*}
    Therefore,
    \begin{eqnarray*}
        &&\|(\mathfrak{S}_p+\tau \I)^{-1} \bfu_p\|^2 \\
        &=& \frac{\bfX(p)}{\sqrt{n}} \bfD_{\psi'(r_{\cdot,[p]})}^{1/2}  \frac{\bfD_{\psi'(r_{\cdot,[p]})}^{1/2} \bfV}{\sqrt{n}} \Big( \frac{\bfV^\T \bfD_{\psi'(r_{\cdot,[p]})} \bfV}{n} + \tau \I \Big)^{-1} \frac{\bfV^\T \bfD_{\psi'(r_{\cdot,[p]})}^{1/2}}{\sqrt{n}} \bfD_{\psi'(r_{\cdot,[p]})}^{1/2} \frac{\bfX(p)}{\sqrt{n}}.
    \end{eqnarray*}
    Note that
    \begin{equation*}
        \frac{\bfD_{\psi'(r_{\cdot,[p]})}^{1/2} \bfV}{\sqrt{n}} \Big( \frac{\bfV^\T \bfD_{\psi'(r_{\cdot,[p]})} \bfV}{n} + \tau \I \Big)^{-1} \frac{\bfV^\T \bfD_{\psi'(r_{\cdot,[p]})}^{1/2}}{\sqrt{n}} \preceq \I.
    \end{equation*}
    So we have
    \begin{equation*}
        \|(\mathfrak{S}_p+\tau \I)^{-1} \bfu_p\|^2 \leq \frac{1}{n} \bfX(p)^\T \bfD_{\psi'(r_{\cdot,[p]})} \bfX(p) = \frac{1}{n} \sum_{i=1}^n x_i^2(p) \psi'(r_{i,[p]}) .
    \end{equation*}
    % For the second term in \eqref{eq:f_tilde_b_iii}, we have
    % \begin{equation*}
    %     \E \Big\|\frac{1}{n} \sum_{i=1}^n \bfx_i x_i(p) \Big\|^2 = p.
    % \end{equation*}
    % \begin{equation*}
    %     \frac{1}{n} \sum_{i=1}^n \bfx_i \psi'(\gamma_{i, p}^*) x_i(p) = \Big\{\frac{1}{n} \sum_{i=1}^n \psi'(\gamma_{i, p}^*) \bfx_i \bfx_i^{\T}\Big\} \left[\begin{array}{c}
    %         \bzero_{p-1} \\
    %         1
    %         \end{array}\right].
    % \end{equation*}
    % Thus
    % \begin{equation*}
    %     \Big\|\frac{1}{n} \sum_{i=1}^n \bfx_i \psi'(\gamma_{i, p}^*) x_i(p) \Big\| \leq \sup_i \|\psi'\|_\infty \| \hat \bSigma\|_2.
    % \end{equation*}
    % Note that $\|\hat \bSigma\|_2 = \mO_{L_k}(\polyLog(n))$ by Lemma \ref{lemma:el_3.38}, and $\sup_i \|\psi'\|_\infty \leq C \polyLog(n)$ under assumption \textbf{O3}.
\end{proof}

\subsubsection{Stochastic aspects}

Assume that $\cX(p) = (\cX_1(p), \ldots, \cX_n(p))^\T$ is independent of $\{\cV_i, \epsilon_i\}_{i=1}^n$, where $\cV_i = \bfV_i / \lambda_i$. This is consistent with assumption \textbf{P1}.

To bound $\|\hat \bdelta - \tilde \bfb\|$, using Equation \eqref{eqn:prop311_iii} we have
\begin{equation*}
    \|(\mathfrak{S}_p+\tau \I)^{-1} u_p\|^2 \leq \frac{1}{n \tau} \sum_{i=1}^n \lambda_i^2 \|\psi'\|_\infty \mathcal{X}_i^2(p),
\end{equation*}
and
\begin{equation*}
    \|(\mathfrak{S}_p+\tau \I)^{-1} u_p\|^2 \leq \frac{\sup_i \|\psi'\|_\infty}{\tau} \frac{1}{n } \sum_{i=1}^n \lambda_i^2  \mathcal{X}_i^2(p).
\end{equation*}
Therefore, under assumption \textbf{O3-O4} and \textbf{O6}, we have for any fixed $k$ and at fixed $\tau$,
\begin{equation*}
    \|(\mathfrak{S}_p+\tau \I)^{-1} u_p\|^2 = \mO_{L_k} (\polyLog (n)).
\end{equation*}
It guarantees that
\begin{equation*}
    \left\|\begin{array}{c}
    (\mathfrak{S}_p+\tau \I)^{-1} \bfu_p \\
    -1
    \end{array}\right\|^2 \leq(1+\|(\mathfrak{S}_p+\tau \I)^{-1} \bfu_p\|^2)=\mO_{L_k}(\polyLog(n)) .
\end{equation*}
Thus
\begin{equation*}
    \|\hat \bdelta - \tilde \bfb\| = \mO_{L_k}\Big(\frac{1}{\tau} \polyLog(n) |\mathfrak{b}_p - \delta_0(p) + \hat w(p) - w_0(p)| \sup _{1 \leq i \leq n}|\mathrm{d}_{i, p}| \|\widehat{\Sigma}\|_2 \Big).
\end{equation*}

% {\color{red} Prove that $|\hat w(p) - w_0(p)|$ is far less than $|\mathfrak{b}_p - \delta_0(p)|$.}

Recall that Lemma 3.38 from \cite{el2018impact} guarantees that $\|\hat \bSigma\| = \mO_{L_k}(\polyLog(n))$ under assupmtion \textbf{O1-O7}.
Also we will show in Proposition \ref{prop:bias_coordinate_bound} that $|\mathfrak{b}_p-\delta_0(p) + \hat w(p) - w_0(p)| = \mO_{L_k}\{\polyLog(n)  (n^{-1/2} \vee n^{-e})\}$ and in Proposition \ref{prop:md_bound} that $\sup _{1 \leq i \leq n}|\mathrm{d}_{i, p}| = \{\polyLog(n)  (n^{\alpha-1/2} \vee n^{\alpha-e})\}$ to show that $M_1$ is small.

\paragraph{On $\mathfrak{b}_p-\delta_0(p)$}~{}

Recall the notations
\begin{equation*}
    \begin{aligned}
        N_p &= \frac{1}{\sqrt{n}} \sum_{i=1}^n x_i(p) \psi(r_{i,[p]}) = \frac{1}{\sqrt{n}} \sum_{i=1}^n \lambda_i \psi(r_{i,[p]}) \cX_i(p), \\
        \xi_n &= \frac{1}{n} \sum_{i=1}^n x_i^2(p) \psi'(r_{i,[p]}) - \bfu_p^{\T}(\mathfrak{S}_p+\tau \I)^{-1} \bfu_p.
    \end{aligned}
\end{equation*}

Under our assumptions, we have $\E (\cX_i) = \bzero_{p}$ and $\cov (\cX_i) = \I_p$, and $\cX(p)$ is independent of $\{ r_{i,[p]} \}_{i=1}^n$.

\begin{proposition}
    \label{prop:bias_coordinate_bound}
    We have
    \begin{equation*}
        |\mathfrak{b}_p-\delta_0(p) + \hat w(p) - w_0(p)| \leq \frac{1}{\sqrt{n}\tau} |N_p| + \|\bdelta_0\|_\infty + |\hat w(p) - w_0(p)|.
    \end{equation*}
    Furthermore, under assumptions \textbf{O1-O7} and \textbf{P1}, $N_p = \mO_{L_k}(\polyLog(n))$ and therefore, when $\tau$ is fixed,
    \begin{equation*}
        |\mathfrak{b}_p-\delta_0(p) + \hat w(p) - w_0(p)| = \mO_{L_k}(\polyLog(n)n^{-1/2} + \|\bdelta_0\|_\infty + |\hat w(p) - w_0(p)|).
    \end{equation*}
\end{proposition}

Note that Lemma \ref{lem:w-coordinate-rate} has shown that
\begin{equation}
    \label{eq:el2018_p}
    |\hat w(p) - w_0(p)| = \mO_{L_k}\Big(\frac{\polyLog(n)}{\sqrt{n} \wedge n^{e}}\Big).
\end{equation}
Under Assumption \textbf{P3}, $n\|\bdelta_0\|_\infty^2 \polyLog(n) n^{2\alpha - 1/2} \to 0$, and therefore we have
\begin{equation}
    \label{eq:wbias}
    |\mathfrak{b}_p-\delta_0(p) + \hat w(p) - w_0(p)| = \mO_{L_k} \Big( \frac{\polyLog(n) }{ \sqrt{n} \wedge n^{\me}} \Big).
\end{equation}

\begin{proof}
    From the definition of $\mathfrak{b}_p$, we have, when $\xi_n>0$,
    % \begin{equation*}
    %     \mathfrak{b}_p-\delta_0(p) = \frac{1}{\sqrt{n}} \frac{N_p}{\tau + \xi_n} - \frac{\tau \delta_0(p)}{\tau + \xi_n} + \{w_0(p) - \hat w(p)\} \frac{\xi_n}{\tau + \xi_n}.
    % \end{equation*}
    \begin{equation*}
        \mathfrak{b}_p-\delta_0(p) + \hat w(p) - w_0(p) = \frac{1}{\sqrt{n}} \frac{N_p}{\tau + \xi_n} - \frac{\tau}{\tau + \xi_n} \{ \delta_0(p) + w_0(p) - \hat w(p)\}.
    \end{equation*}
    We will see later that $\xi_n \geq 0$ in Lemma \ref{lemma:xi_nonnegative}. It follows that
    % \begin{equation*}
    %     |\mathfrak{b}_p-\delta_0(p)| \leq \frac{1}{\sqrt{n}\tau} |N_p| + |\delta_0(p)| + |\hat w(p) - w_0(p)| .
    % \end{equation*}
    \begin{equation*}
        |\mathfrak{b}_p-\delta_0(p) + \hat w(p) - w_0(p)| \leq \frac{1}{\sqrt{n}\tau} |N_p| + |\delta_0(p) + w_0(p) - \hat w(p)| .
    \end{equation*}
    Using indenpendence of $\cX(p)$ and $\{ \cV_i, \epsilon_i \}_{i=1}^n$, and the fact that $\E (\cX_i) = \bzero_{p}$, we have
    \begin{equation*}
        \E (N_p^2) = \frac{1}{n} \sum_{i=1}^n \E\{\cX_i^2(p)\} \E\{\lambda_i^2 \psi^2(r_{i,[p]})\} ,
    \end{equation*}
    whether the right hand side is finite or not.
    Using the bounds on $\max \lambda_i^2$ and $\sup_i \| \psi\|_\infty$, we have
    \begin{equation*}
        \E (N_p^2) \leq \frac{1}{\sqrt{n}} \sum_{i=1}^n \E\{\cX_i^2(p)\}  \| \psi\|_\infty \E \{\lambda_i^2\} = \mO(1) = \mO(\polyLog(n)).
    \end{equation*}
    Similarly, it is clear that
    \begin{equation*}
        N_p = \mO_{L_k}(\polyLog(n)).
    \end{equation*}
    Therefore, we have
    \begin{equation*}
        |\mathfrak{b}_p-\delta_0(p) + \hat w(p) - w_0(p)| = \frac{1}{\sqrt{n}\tau}\mO_{L_k}(\polyLog(n)) + \sup_{1\leq k \leq p}|\delta_0(k)| + |\hat w(p) - w_0(p)|.
    \end{equation*}
\end{proof}

\paragraph{On $\xi_n$}~{}

Write $\xi_n$ in matrix form. Let $\bfD_{\psi'(r_{\cdot,[p]})}$ be the diagonal matrix with $(i,i)$ entry $\psi'(r_{i,[p]})$.
Denote $\bfX(p)$ the last column of the design matrix $\bfX$.
Then we have
\begin{equation*}
    \xi_n=\frac{1}{n} \bfX(p)^\T \bfD_{\psi'(r_{\cdot,[p]})}^{1 / 2} \bfM \bfD_{\psi'(r_{\cdot,[p]})}^{1 / 2} \bfX(p),
\end{equation*}
where
\begin{equation}
    \bfM = \I_n - \frac{\bfD_{\psi'(r_{\cdot,[p]})}^{1 / 2} \bfV}{\sqrt{n}} \Big(\frac{1}{n} \bfV^\T \bfD_{\psi'(r_{\cdot,[p]})} \bfV + \tau \I \Big)^{-1} \frac{\bfV^\T \bfD_{\psi'(r_{\cdot,[p]})}^{1 / 2}}{\sqrt{n}}.
\end{equation}

\begin{lemma}
    \label{lemma:xi_nonnegative}
    We have
    \begin{equation*}
        \xi_n \geq 0.
    \end{equation*}

    Furthermore, under assumptions \textbf{O1-O7} and \textbf{P1}, if $\bfD_{\lambda_i}$ is the diagonal matrix with $(i,i)$ entry $\lambda_i$,
    \begin{equation*}
        \Big| \xi_n - \frac{1}{n} \tr \Big( \bfD_{\lambda_i} \bfD_{\psi'(r_{\cdot,[p]})}^{1/2} \bfM \bfD_{\psi'(r_{\cdot,[p]})}^{1/2} \bfD_{\lambda_i} \Big) \Big| = \mO_{L_k}\Big(\sup_{1\leq i \leq n} \lambda_i^2 \psi'(r_{i,[p]}) / \sqrt{n \mc_n}\Big).
    \end{equation*}
\end{lemma}

\begin{proof}
    When $\tau>0$, all the eigenvalues of $\bfM$ are positive.
    Indeed, if the singular values of $n^{1/2} \bfD_{\psi'(r_{\cdot,[p]})}^{1/2} \bfV$ are denoted by $\sigma_1$, the eigenvalues of $\bfM$ are $\tau/(\sigma_i^2 + \tau)$.

    Therefore, since $\xi_n = n^{-1} \bfv^\T \bfM \bfv$ with $\bfv = \bfD_{\psi'(r_{\cdot,[p]})}^{1/2} \bfX(p)$, we have $\xi_n \geq 0$.

    Since $\bfM$ is symmetric and has eigenvalues between $0$ and $1$,
    using Lemma V.1.5 in \cite{bhatia1997matrix}, we have
    \begin{equation*}
        0 \preceq \bfD_{\psi'(r_{\cdot,[p]})}^{1/2} \bfM \bfD_{\psi'(r_{\cdot,[p]})}^{1/2} \preceq \bfD_{\psi'(r_{\cdot,[p]})}.
    \end{equation*}

    Under Assumption \textbf{P1}, the matrix $\bfM$ is independent of $\cX(p)$.
    $\bfD_{\psi'(r_{\cdot,[p]})}$ is also independent of $\cX(p)$.
    By definition, $\bfX(p) = \bfD_{\lambda_i} \cX(p)$.

    Under Assumption \textbf{P1}, $\cX_p$ satisfy the necessary concentration assumptions. Using Lemma 3.37 in \cite{el2018impact}, we have
    \begin{equation*}
        \begin{aligned}
            &\Big| \frac{1}{n} \bfX(p)^\T \bfD_{\psi'(r_{\cdot,[p]})}^{1/2} \bfM \bfD_{\psi'(r_{\cdot,[p]})}^{1/2} \bfX(p) - \frac{1}{n} \tr \Big( \bfD_{\lambda_i} \bfD_{\psi'(r_{\cdot,[p]})}^{1/2} \bfM \bfD_{\psi'(r_{\cdot,[p]})}^{1/2} \bfD_{\lambda_i} \Big) \Big| \\
            &= \mO_{L_k}\Big( \frac{1}{\sqrt{n \mc_n}}\sup_{1\leq i \leq n} \lambda_i^2 \psi'(r_{i,[p]}) \Big).
        \end{aligned}
    \end{equation*}
\end{proof}

\paragraph{About $\frac{1}{n} \tr \Big( \bfD_{\lambda_i} \bfD_{\psi'(r_{\cdot,[p]})}^{1/2} \bfM \bfD_{\psi'(r_{\cdot,[p]})}^{1/2} \bfD_{\lambda_i} \Big)$}~{}

\begin{lemma}%[{{\citet[Lemma 3.14]{el2018impact}}}]
    \label{lemma:trace_approx}
    Denote
    \begin{equation*}
        \mathfrak{S}_p=\frac{1}{n} \sum_{i=1}^n \psi' (r_{i,[p]}) \bfv_i \bfv_i^{\T},
        \text{ and } \ \mathfrak{S}_p (i) = \mathfrak{S}_p - \frac{1}{n} \psi' (r_{i,[p]}) \bfv_i \bfv_i^{\T}.
    \end{equation*}
    Denote
    \begin{equation}
        \begin{aligned}
        \mc_{\tau, p} & =\frac{1}{n} \tr \{(\mathfrak{S}_p+\tau \I)^{-1}\}, \\
        \zeta_i & =\frac{1}{n} \bfv_i^{\T}\{\mathfrak{S}_p(i)+\tau \I\}^{-1} \bfv_i-\lambda_i^2 \mc_{\tau, p} .
        \end{aligned}
    \end{equation}
    Then we have under assumptions \textbf{O1-O7} and \textbf{P1}, if $\bfM$ is the matrix defined in Lemma \ref{lemma:xi_nonnegative},
    \begin{equation*}
        \Big| \frac{1}{n} \tr (\I_n - \bfM) - \Big\{ \frac{1}{n} \tr \Big( \bfD_{\lambda_i} \bfD_{\psi'(r_{\cdot,[p]})}^{1/2} \bfM \bfD_{\psi'(r_{\cdot,[p]})}^{1/2} \bfD_{\lambda_i} \Big) \Big\} \mc_{\tau, p} \Big| \leq \sup_i |\zeta_i| \cdot \frac{1}{n} \sum_{i=1}^n \psi'(r_{i,[p]}) .
    \end{equation*}
    We also have
    \begin{equation*}
        \frac{1}{n} \tr (\I_n - \bfM) = \frac{p-1}{n} - \tau \mc_{\tau, p} .
    \end{equation*}
\end{lemma}

\begin{proof}
    Denote $d_{i,i} = \psi'(r_{i,[p]}) / n$.
    By using the Sherman-MorrisonWoodbury formula (see, e.g., \cite{johnson1985matrix}, p.19),
    \begin{equation*}
        \begin{aligned}
            M_{i, i} & =1-d_{i, i} \bfv_i^{\T} \left(\bfV^{\T} \bfD_{\psi'(r_{\cdot,[p]})} \bfV / n+\tau \I \right)^{-1} \bfv_i , \\
            & =1-d_{i, i} \frac{\bfv_i^{\T} \left(\mathfrak{S}_p(i)+\tau \I \right)^{-1} \bfv_i}{1+d_{i, i} \bfv_i^{\T} \left(\mathfrak{S}_p(i)+\tau \I \right)^{-1} \bfv_i} , \\
            & =\frac{1}{1+d_{i, i} \bfv_i^{\T} \left(\mathfrak{S}_p(i)+\tau \I\right)^{-1} \bfv_i}.
        \end{aligned}
    \end{equation*}
    Recall that we are interested in
    \begin{equation*}
        \frac{1}{n} \sum_i \lambda_i^2 \psi' (r_{i,[p]}) M_{i,i} = \frac{1}{n} \Big( \bfD_{\lambda_i} \bfD_{\psi'(r_{\cdot,[p]})}^{1/2} \bfM \bfD_{\psi'(r_{\cdot,[p]})}^{1/2} \bfD_{\lambda_i} \Big).
    \end{equation*}
    By the property of trace, we have
    \begin{equation*}
        \begin{aligned}
            \tr(\I_n - \bfM) &= \tr\{(\mathfrak{S}_p+\tau \I)^{-1} \mathfrak{S}_p\} \\
            &= p - 1 - \tau \tr\{(\mathfrak{S}_p+\tau \I)^{-1} \}
            = p - 1 - n \tau \mc_{\tau, p}.
        \end{aligned}
    \end{equation*}
    This shows the second result of the lemma.

    The first result follows from the fact that
    \begin{equation}
        \label{eq_tr_I_M}
        \begin{aligned}
            \tr (\I_n - \bfM) =& \sum_i (1 - M_{i,i}) \\
            =&  \sum_i d_{i,i} \frac{\bfv_i^{\T} \left(\mathfrak{S}_p(i)+\tau \I\right)^{-1} \bfv_i}{1 +d_{i, i} \bfv_i^{\T} \left(\mathfrak{S}_p(i)+\tau \I\right)^{-1} \bfv_i}.
        \end{aligned}
    \end{equation}
    With our definitions, we have, since $\lambda_i^2 \mc_{\tau, p} + \zeta_i = n^{-1} \bfv_i^{\T} (\mathfrak{S}_p(i)+\tau \I)^{-1} \bfv_i$,
    \begin{equation*}
        \begin{aligned}
            \frac{1}{n} \tr (\I_n - \bfM) =& \Big( \frac{1}{n} \sum_i \lambda_i^2 \psi' (r_{i,[p]}) M_{i,i} \Big) \mc_{\tau, p}
            + \frac{1}{n} \sum_i \psi' (r_{i,[p]}) \frac{\zeta_i}{1 +d_{i, i} \bfv_i^{\T} \left(\mathfrak{S}_p(i)+\tau \I\right)^{-1} \bfv_i}.
        \end{aligned}
    \end{equation*}
    It immediately follows that
    \begin{equation*}
        \Big| \frac{1}{n} \tr (\I_n - \bfM) - \Big( \frac{1}{n} \sum_i \lambda_i^2 \psi' (r_{i,[p]}) M_{i,i} \Big) \mc_{\tau, p} \Big| \leq \sup_i |\zeta_i| \cdot \frac{1}{n} \sum_{i=1}^n \psi'(r_{i,[p]}).
    \end{equation*}
\end{proof}

\paragraph{Controlling $\zeta_i$}~{}

\begin{lemma}%[{{\citet[Lemma 3.15]{el2018impact}}}]
    \label{lemma:zeta_bound}
    Suppose we can find $\{ \mr_{i,[p]}^{(i)}\}_{j \neq i}$ independent of $(\lambda_i, \cV_i)$ and $K_n$ such that
    \begin{equation*}
        \sup_i \sup_{j \neq i} |\psi_j'(\mr_{j,[p]}^{(i)}) - \psi_j'(r_{j,[p]})| \leq K_n.
    \end{equation*}
    Then
    \begin{equation*}
        \sup_i |\zeta_i| = \mO_{L_k} \Big[ \Big\{ \frac{1}{\tau^2} K_n \|\hat \bSigma\|_2 + \frac{\polyLog(n)}{\tau \sqrt{n \mc_n}} + \frac{1}{n \tau}\Big\} \polyLog(n) \Big],
    \end{equation*}
    provided that $K_n$ has $3k$ uniformly bounded moments.
\end{lemma}

\begin{proof}
    Denote
    \begin{equation*}
        \mathbf{AM}_{i,p} = \frac{1}{n} \sum_{j \neq i} \psi'(r_{j,[p]}^{(i)}) \bfv_j \bfv_j^{\T}.
    \end{equation*}
    Then, using the fact that $\bfA^{-1} - \bfB^{-1} = \bfA^{-1} (\bfB - \bfA) \bfB^{-1}$, we have
    \begin{equation*}
        \| (\mathfrak{S}_p(i) + \tau \I)^{-1} - (\mathbf{AM}_{i,p} + \tau \I)^{-1} \| \leq \frac{1}{\tau^2} K_n \|\hat \bSigma\|_2,
    \end{equation*}
    since $\| n^{-1} \sum_i \bfv_i \bfv_i^{\T} \| \leq \|\hat \bSigma\|_2$.
    In particular, we have
    \begin{equation*}
        \Big| \frac{1}{n} \bfv_i^{\T} (\mathfrak{S}_p(i) + \tau \I)^{-1} \bfv_i - \frac{1}{n} \bfv_i^{\T} (\mathbf{AM}_{i,p} + \tau \I)^{-1} \bfv_i \Big| \leq \frac{\| \bfv_i\|^2}{n} \frac{1}{\tau^2} K_n \|\hat \bSigma\|_2.
    \end{equation*}
    Since $\mathbf{AM}_{i,p}$ is independent of $(\lambda_i, \cV_i)$, we can use Lemma 3.37 in \cite{el2018impact} and since $\bfv = \lambda_i \cV_i$, we have
    \begin{equation*}
        \sup_{1 \leq i \leq n} \Big| \frac{1}{n} \bfv_i^{\T} (\mathbf{AM}_{i,p} + \tau \I)^{-1} \bfv_i - \frac{\lambda_i^2}{n} \tr \{ (\mathbf{AM}_{i,p} + \tau \I)^{-1} \} \Big| = \mO_{L_k} \Big( \frac{\polyLog(n)}{\tau \sqrt{n \mc_n}} \sup_{1 \leq i \leq n} \lambda_i^2 \Big),
    \end{equation*}
    by using the fact that $\lambda_{\max} ((\mathbf{AM}_{i,p} + \tau \I)^{-1}) \leq \tau^{-1}$.

    Using the operator norm bound we gave above, we also have
    \begin{equation*}
        \Big| \frac{1}{n} \tr \{ (\mathbf{AM}_{i,p} + \tau \I)^{-1} \} - \frac{1}{n} \tr \{ (\mathfrak{S}_p(i) + \tau \I)^{-1} \} \Big| \leq \frac{1}{\tau^2} K_n \|\hat \bSigma\|_2 \frac{p}{n}.
    \end{equation*}

    We can now conclude that
    \begin{equation*}
        \begin{aligned}
            &\sup_{1 \leq i \leq n} \Big| \frac{1}{n} \bfv_i^{\T} (\mathfrak{S}_p(i) + \tau \I)^{-1} \bfv_i - \frac{\lambda_i^2}{n} \tr \{ (\mathfrak{S}_p(i) + \tau \I)^{-1} \} \Big| \\
            & \quad = \mO \bigg[ \bigg\{ \frac{1}{\tau^2} K_n \|\hat \bSigma\|_2 \sup_{1 \leq i \leq n} \Big( \frac{p}{n} + \frac{\| \bfv_i\|^2}{n} \Big) + \frac{\polyLog(n)}{\tau \sqrt{n \mc_n}} \bigg\} \Big( \sup_{1 \leq i \leq n} \lambda_i^2 \vee 1 \Big) \bigg].
        \end{aligned}
    \end{equation*}

    So under \textbf{O1} and \textbf{O4}, $\sup_{1 \leq i \leq n} \| \bfv_i\|^2 / n = \mO_{L_k}(1)$ and finally
    \begin{equation*}
        \begin{aligned}
            &\sup_{1 \leq i \leq n} \Big| \frac{1}{n} \bfv_i^{\T} (\mathfrak{S}_p(i) + \tau \I)^{-1} \bfv_i - \frac{\lambda_i^2}{n} \tr \{ (\mathfrak{S}_p(i) + \tau \I)^{-1} \} \Big| \\
            & \quad = \mO_{L_k} \bigg\{ \Big( \frac{1}{\tau^2} K_n \|\hat \bSigma\|_2 + \frac{\polyLog(n)}{\tau \sqrt{n \mc_n}} \Big) \Big( \sup_{1 \leq i \leq n} \lambda_i^2 \vee 1 \Big) \bigg\}.
        \end{aligned}
    \end{equation*}

    \paragraph{Control of $n^{-1} \tr\{ (\mathfrak{S}_p(i) + \tau \I)^{-1} \} - n^{-1} \tr\{ (\mathfrak{S}_p + \tau \I)^{-1} \} $}~{}

    Using the Sherman-Woodbury-Morrison formula, we have
    \begin{equation*}
        (\mathfrak{S}_p(i) + \tau \I)^{-1} - (\mathfrak{S}_p + \tau \I)^{-1} = \frac{\psi'(r_{i, [p]})}{n} \frac{(\mathfrak{S}_p(i) + \tau \I)^{-1} \bfv_i \bfv_i^{\T} (\mathfrak{S}_p(i) + \tau \I)^{-1}}{ 1 + \frac{\psi'(r_{i, [p]})}{n} \bfv_i^{\T} (\mathfrak{S}_p(i) + \tau \I)^{-1} \bfv_i}.
    \end{equation*}
    Take the trace of both sides, we have
    \begin{equation*}
        0 \leq \tr \{ (\mathfrak{S}_p(i) + \tau \I)^{-1} \} - \tr \{ (\mathfrak{S}_p + \tau \I)^{-1} \} \leq \frac{1}{\tau},
    \end{equation*}
    since $\bfv_i^{\T} (\mathfrak{S}_p(i) + \tau \I)^{-2} \bfv_i \leq \tau^{-1} \bfv_i^{\T} (\mathfrak{S}_p(i) + \tau \I)^{-1} \bfv_i$.

    Therefore,
    \begin{equation*}
        0 \leq \frac{1}{n} \tr\{ (\mathfrak{S}_p(i) + \tau \I)^{-1} \} - \frac{1}{n} \tr\{ (\mathfrak{S}_p + \tau \I)^{-1} \} \leq \frac{1}{n \tau}.
    \end{equation*}

    We can now conclude that
    \begin{equation*}
        \sup_i |\zeta_i| = \mO_{L_k} \Big[ \Big\{ \frac{1}{\tau^2} K_n \|\hat \bSigma\|_2 + \frac{\polyLog(n)}{\tau \sqrt{n \mc_n}} + \frac{1}{n \tau}\Big\} \Big( \sup_{1 \leq i \leq n} \lambda_i^2 \vee 1 \Big) \Big].
    \end{equation*}
    provided we can use Holder's inequality.
    This, in turn, requires $K_n$ to have $3k$ moments.
\end{proof}

\paragraph{Control of $K_n$}~{}

A natural chioce for $\{ \mr_{i,[p]}^{(i)}\}_{j \neq i}$ defined in Lemma \ref{lemma:zeta_bound} is to use a leave one out estimator of $\hat \bgamma$, where the $i$-th observation (and hence $\bfv_i$) is removed.
Hence, all the work done in Theorem \ref{th:l1oo} can be used here.

\begin{lemma}% [{{\citet[Lemma 3.16]{el2018impact}}}]
    \label{lemma:kn_bound}
    Suppose we use for $\mr_{i,[p]}^{(i)}$ the residuals we would get by using a leave-one-out estimator of $\hat \bgamma$, i.e. excluding $(\bfv_i, \epsilon_i)$ from problem \eqref{eq:def_gamma_hat}.

    With the notations of Lemma \ref{lemma:zeta_bound}, we have under assumptions \textbf{O1-O7} and \textbf{P1},
    \begin{equation*}
        K_n = \mO_{L_k} (n^{2\alpha - 1/2} \polyLog(n)).
    \end{equation*}
    In particular, for any fixed $\tau$,
    \begin{equation*}
        \sup_i |\zeta_i| = \mO_{L_k} (n^{2\alpha - 1/2} \polyLog(n)).
    \end{equation*}

\end{lemma}

\begin{proof}
    Denote $\delta_n(i)$ random variables such that
    \begin{equation*}
        \sup_{j \neq i} | \mr_{j,[p]}^{(i)} - r_{j,[p]} | \leq \delta_n(i).
    \end{equation*}
    Applying Theorem \ref{th:l1oo} with $R_j = r_{j,[p]}$ and $\tilde r_{j,(i)} = \mr_{j,[p]}^{(i)}$, we have
    \begin{equation*}
        \sup_i \delta_n(i) = \mO_{L_k} (n^{2\alpha - 1/2} \polyLog(n)).
    \end{equation*}

    The control of $K_n$ follows by using the fact that $\psi'$ is $C n^\alpha$-Lipschitz.
\end{proof}

\begin{corollary}
    \label{cor:ci_ctau_gap}
    Recall that in \eqref{eqn:c_i} and \eqref{eq:proof_def}
    \begin{equation*}
        c_i = \frac{1}{n} \bfx_i^\T (\bfS_i + \tau\I)^{-1} \bfx_i , \quad
        c_\tau =\frac{1}{n} \tr(\bfS+\tau \I)^{-1}.
    \end{equation*}
    Then under assumptions \textbf{O1-O7} and \textbf{P1}, we have
    \begin{equation*}
        \sup_i |c_i - \lambda_i^2 c_\tau| = \mO_{L_k} (n^{2\alpha - 1/2} \polyLog(n)).
    \end{equation*}

\end{corollary}

\begin{proof}
    We have showed that
    \begin{equation*}
        \sup_i \Big| \frac{1}{n} \bfv_i^{\T} (\mathfrak{S}_p(i) + \tau \I)^{-1} \bfv_i - \lambda_i^2 \mc_{\tau, p} \Big| = \mO_{L_k} (n^{2\alpha - 1/2} \polyLog(n)).
    \end{equation*}
    Recall that
    \begin{equation*}
        c_\tau = \frac{1}{n} \tr \bigg[ \Big\{ \frac{1}{n} \sum_{i=1}^n \psi'(R_i) \bfx_i \bfx_i^\T + \tau \I \Big\}^{-1} \bigg].
    \end{equation*}
    We see that $c_\tau$ is analogous to $\mc_{\tau, p}$ when we use all the data, rather than $(p-1)$ of them.

    Indeed, $c_i$ in \eqref{eqn:c_i} is defined, in the notation of the proof of Lemma \ref{lemma:zeta_bound} as an analog of $n^{-1} \bfv_i^{\T} (\mathbf{AM}_{i,p} + \tau \I)^{-1} \bfv_i$, with the the role of $\{ \mr_{i,[p]}^{(i)}\}_{j \neq i}$ being played by the residuals obtained by the leave-one-out estimator of $\hat \bgamma$, excluding $(\bfx_i, y_i)$ from the problem.
    Lemma \ref{lemma:zeta_bound} in connection with Theorem \ref{thm:delta_tildeb_bound} shows that $\sup_i | n^{-1} \bfv_i^{\T} (\mathbf{AM}_{i,p} + \tau \I)^{-1} \bfv_i - \lambda_i^2 \mc_{\tau, p} | = \mO_{L_k} (n^{2\alpha - 1/2} \polyLog(n))$.
    Passing from the $p-1$ dimensional version of this result, i.e., Lemma \ref{lemma:zeta_bound}, to the $p$-dimensional version gives the approximation stated in the corollary.

    We therefore have
    \begin{equation*}
        \sup_i |c_i - \lambda_i^2 c_\tau| = \mO_{L_k} (n^{2\alpha - 1/2} \polyLog(n)).
    \end{equation*}
\end{proof}

\paragraph{Further results on $\xi_n$ and $\mathfrak{b}_p$}~{}

\begin{proposition}
    \label{prop:mc_xi_relation}
    Under Assumptions \textbf{O1-O7} and \textbf{P1}, we have
    \begin{equation}
        \label{eq:el_3.18_i}
        \Big| \mc_{\tau, p} (\xi_n + \tau) - \frac{p-1}{n} \Big| = \mO_{L_k} \Big( \frac{\polyLog(n)}{n^{1/2 - 2\alpha}}  \Big).
    \end{equation}
    Furthermore, under Assumptions \textbf{O1-O7} and \textbf{P1-P3}, since $\|\bdelta_0\|_\infty = \mO_{L_k}(n^{-e})$,
    \begin{equation}
        \label{eq:el_3.18_ii}
        \begin{aligned}
            \Big( \frac{p}{n} \Big)^2 n \E [ \{ \mathfrak{b}_p - \delta_0(p) + \hat w(p) - w_0(p)\}^2] =& \frac{1}{n} \sum_{i=1}^n \E \Big[ \{ \mc_{\tau, p} \lambda_i \psi(r_{i,[p]})\}^2 \Big] \\
            &+ n \tau^2 \{ \delta_0(p) -  \hat w(p) + w_0(p) \}^2 \E (\mc_{\tau, p}^2) + \mo(1).
        \end{aligned}
    \end{equation}
\end{proposition}

\begin{proof}
    \textbf{For Equation \eqref{eq:el_3.18_i}}:

    By Lemma \ref{lemma:trace_approx}, we have
    \begin{equation*}
        \frac{p-1}{n} - \tau \mc_{\tau, p} = \frac{1}{n} \tr (\I_n - \bfM) \geq 0.
    \end{equation*}
    The latter quantity was approximated in Lemma \ref{lemma:trace_approx} by
    \begin{equation*}
        \frac{1}{n} \tr \Big( \bfD_{\lambda_i} \bfD_{\psi'(r_{\cdot,[p]})}^{1/2} \bfM \bfD_{\psi'(r_{\cdot,[p]})}^{1/2} \bfD_{\lambda_i} \Big) \mc_{\tau, p},
    \end{equation*}
    which approximate $\xi$ as in Lemma \ref{lemma:xi_nonnegative}.
    This gives the result of Equation \eqref{eq:el_3.18_i}, by simply keeping track of the approximation errors we make at each step.

    \textbf{For Equation \eqref{eq:el_3.18_ii}}:

    Recall that by definition:
    \begin{equation*}
        \sqrt{n} \big[ (\tau + \xi_n) \mathfrak{b}_p - \xi_n \delta_0(p) + \xi_n \{ \hat w(p) - w_0(p) \} \big] = N_p = \frac{1}{\sqrt{n}} \sum_{i=1}^n \lambda_i  \psi(r_{i,[p]}) \cX_i(p).
    \end{equation*}
    Therefore, we have
    \begin{equation*}
        \begin{aligned}
            \mc_{\tau,p} \sqrt{n} (\tau + \xi_n) \{ \mathfrak{b}_p - \delta_0(p) + \hat w(p) - w_0(p) \}  =&  \frac{1}{\sqrt{n}} \sum_{i=1}^n \mc_{\tau,p} \lambda_i  \psi(r_{i,[p]}) \cX_i(p)\\
            & - \mc_{\tau,p} \sqrt{n} \tau \{ \delta_0(p) -  \hat w(p) + w_0(p) \}.
        \end{aligned}
    \end{equation*}
    Note that $\mc_{\tau,p} \lambda_i  \psi(r_{i,[p]})$ is independent of $\cX_i(p)$, and $\cX_i(p)$'s are independent with mean zero and variance one.
    We have
    \begin{equation}
        \label{eq:pf_el_3.18_i}
        \begin{aligned}
            \E \Big[\mc_{\tau,p}^2 n (\tau + \xi_n)^2 \{ \mathfrak{b}_p - \delta_0(p) + \hat w(p) - w_0(p) \}^2 \Big] =& \frac{1}{n} \sum_{i=1}^n \E \Big[ \{ \mc_{\tau,p} \lambda_i  \psi(r_{i,[p]}) \}^2 \Big] \\
            &+ n \tau^2 \{ \delta_0(p) -  \hat w(p) + w_0(p) \}^2 \E (\mc_{\tau,p}^2).
        \end{aligned}
    \end{equation}

    Recall that Proposition \ref{prop:bias_coordinate_bound} gives that
    \begin{equation*}
        |\mathfrak{b}_p-\delta_0(p) + \hat w(p) - w_0(p)| \leq \frac{1}{\sqrt{n}\tau} |N_p| + \|\bdelta_0\|_\infty + |\hat w(p) - w_0(p)|.
    \end{equation*}
    Under Assumption \textbf{P3}, $n\|\bdelta_0\|_\infty^2 \polyLog(n) n^{2\alpha - 1/2} \to 0$.
    Together with Equations \eqref{eq:el_3.18_i} and \eqref{eq:el2018_p}, the LHS of \eqref{eq:pf_el_3.18_i} can be written as
    \begin{equation*}
        \begin{aligned}
            &\E \Big[\mc_{\tau,p}^2 (\tau + \xi_n)^2 n \{ \mathfrak{b}_p - \delta_0(p) + \hat w(p) - w_0(p) \}^2 \Big] \\
            =& \E \Big[ \Big\{\mc_{\tau,p}(\tau + \xi_n) - \frac{p-1}{n} + \frac{p-1}{n} \Big\}^2 n \{ \mathfrak{b}_p - \delta_0(p) + \hat w(p) - w_0(p) \}^2 \Big] \\
            =& \Big(\frac{p}{n} \Big)^2 n \E [ \{ \mathfrak{b}_p - \delta_0(p) + \hat w(p) - w_0(p) \}^2] + \mo(1).
        \end{aligned}
    \end{equation*}
    This implies Equation \eqref{eq:el_3.18_ii}.

\end{proof}

\paragraph{On $\md_{i,p}$}~{}

Recall that
\begin{equation*}
    \md_{i,p} = \psi' (\gamma_{i, p}^*)-\psi'(r_{i,[p]}),
\end{equation*}
where $\gamma_{i, p}^*$ in the interval $(r_{i,[p]}, r_{i,[p]}+\nu_i)$, with
\begin{equation*}
    \begin{aligned}
        \nu_i &= \{\mathfrak{b}_p - \delta_0(p) + \hat w(p) - w_0(p)\} \bfx_i^\T \left[\begin{array}{c}
            (\mathfrak{S}_p+\tau \I)^{-1} \bfu_p \\
            -1
            \end{array}\right] \\
            &= \{\mathfrak{b}_p - \delta_0(p) + \hat w(p) - w_0(p)\} \pi_i.
    \end{aligned}
\end{equation*}

We have the following result.

\begin{proposition}
    \label{prop:md_bound}
    Under Assumptions \textbf{O1-O7} and \textbf{P1-P3}, we have, at fixed $\tau$,
    \begin{equation*}
        \sup_i |\md_{i,p}| = \mO_{L_k} \Big( \frac{\polyLog(n) n^\alpha}{n^{1 / 2} \wedge n^{\me}} \Big).
    \end{equation*}
\end{proposition}

\begin{proof}
    Note that we can write
    \begin{equation*}
        \pi_i = \bfx_i^\T \left[\begin{array}{c}
            (\mathfrak{S}_p+\tau \I)^{-1} \bfu_p \\
            -1
            \end{array}\right] = \bfv_i^\T (\mathfrak{S}_p+\tau \I)^{-1} \bfu_p - x_i(p).
    \end{equation*}
    Recall that $\bfu_p = n^{-1} \bfV^\T \bfD_{\psi'(r_{\cdot,[p]})} \bfX(p)$.
    We can also write it as
    \begin{equation*}
        \bfu_p = \frac{1}{n} \cV^\T \bfD_{\lambda_i^2 \psi'(r_{\cdot,[p]})} \cX(p).
    \end{equation*}

    Using the indenpendence of $\cX(p)$ with $\{(\cV_i, \epsilon_i)\}_{i=1}^n$, and the concentration assumptions on $\cX(p)$ in Assumption \textbf{P1}, according to Lemma 3.36 in \cite{el2018impact}, we have
    \begin{equation*}
        \sup_i | \bfv_i^\T (\mathfrak{S}_p+\tau \I)^{-1} \bfu_p | = \mO_{L_k} \Big( \frac{\polyLog(n) }{\mc_n^{1/2}} \sup_i \big\| \frac{1}{n} \bfD_{\lambda_i^2 \psi'(r_{\cdot,[p]})} \cV (\mathfrak{S}_p+\tau \I)^{-1} \cV_i \big\| \Big),
    \end{equation*}
    where we view $\bfv_i^\T (\mathfrak{S}_p+\tau \I)^{-1} \bfu_p$ as a linear form in $\cV_i$.
    Note that we have absorbed the $\sup_i|\lambda_i|$ in the $\polyLog(n)$ term.

    We can write
    \begin{equation*}
        \big\| \frac{1}{n} \bfD_{\lambda_i^2 \psi'(r_{\cdot,[p]})} \cV (\mathfrak{S}_p+\tau \I)^{-1} \cV_i \big\| = \frac{1}{n} \cV_i^\T (\mathfrak{S}_p+\tau \I)^{-1} \frac{\cV^\T \bfD_{\lambda_i^2 \psi'(r_{\cdot,[p]})}^2 \cV}{n} (\mathfrak{S}_p+\tau \I)^{-1} \cV_i.
    \end{equation*}
    We notice that $\mathfrak{S}_p = n^{-1} \cV^\T \bfD_{\lambda_i^2 \psi'(r_{\cdot,[p]})} \cV$.
    Hence,
    \begin{equation*}
        \frac{\cV^\T \bfD_{\lambda_i^2 \psi'(r_{\cdot,[p]})}^2 \cV}{n} \preceq \|\bfD_{\lambda_i^2 \psi'(r_{\cdot,[p]})}\|_2 \mathfrak{S}_p.
    \end{equation*}
    Therefore we conclude that
    
    \begin{equation*}
        \frac{1}{n} \cV_i^\T (\mathfrak{S}_p+\tau \I)^{-1} \frac{\cV^\T \bfD_{\lambda_i^2 \psi'(r_{\cdot,[p]})}^2 \cV}{n} (\mathfrak{S}_p+\tau \I)^{-1} \cV_i \leq \frac{\|\cV_i\|^2}{n \tau} \|\bfD_{\lambda_i^2 \psi'(r_{\cdot,[p]})}\|_2 = \frac{\|\cV_i\|^2}{n \tau} \sup_i \lambda_i^2 \psi'(r_{i,[p]}).
    \end{equation*}

    Note that $\sup_i | x_i(p)| = \mO_{L_k} (\polyLog(n) / \sqrt{\mc_n})$ under Assumption \textbf{O4}, \textbf{O6} and \textbf{P1}, according to Appendix 7 in \cite{el2018impact}.
    Therefore, we have
    \begin{equation*}
        \begin{aligned}
            \sup_i |\pi_i| &= \mO_{L_k} \bigg[ \frac{\polyLog(n)}{\mc_n^{1/2}} \Big\{ 1 + \sqrt{ \sup_i \lambda_i^2 \psi'(r_{i,[p]}) \sup_i \frac{\|\cV_i\|^2}{n \tau}} \Big\} \bigg] \\
            &= \mO_{L_k} \bigg[ \frac{\polyLog(n)}{\mc_n^{1/2}} \Big\{ 1 + \sqrt{ \sup_i \lambda_i^2 \psi'(r_{i,[p]}) } \Big\} \bigg] \\
            &= \mO_{L_k} \{ \polyLog(n)\}.
        \end{aligned}
    \end{equation*}
    Recall that $|\mathfrak{b}_p-\delta_0(p) + \hat w(p) - w_0(p)| = \mO_{L_k}(\polyLog(n)n^{-1/2} + \|\bdelta_0\|_\infty + |\hat w(p) - w_0(p)|)$ and by Lemma \ref{lem:w-coordinate-rate} we have
    \begin{equation*}
        |\hat w(p) - w_0(p)| = \mO_{L_k}\Big(\frac{\polyLog(n)}{\sqrt{n} \wedge n^{e}}\Big).
    \end{equation*}
    We have
    \begin{equation*}
        \sup_i |\nu_i| = \mO_{L_k} \Big( \frac{\polyLog(n) }{ \sqrt{n} \wedge n^{\me}} \Big).
    \end{equation*}
    Note that \cite{el2018impact} has shown that
    Under our assumption that $\psi'$ is $C n^\alpha$-Lipschitz, we see that
    \begin{equation*}
        \sup _i|\md_{i, p}|=\mO_{L_k}\Big(\frac{\polyLog(n) n^\alpha}{n^{1 / 2} \wedge n^{\me}}\Big) .
    \end{equation*}
\end{proof}

\subsubsection{Final conclusions}

Gathering all the results, we have the following Theorem.

\begin{theorem}
    \label{thm:delta_tildeb_bound}
    Under Assumptions \textbf{O1-O7} and \textbf{P1-P3}, we have, for any fixed $\tau$,
    \begin{equation*}
        \| \hat \bdelta - \tilde \bfb\| \leq \mO_{L_k} \Big( \frac{\polyLog(n) n^\alpha}{(n^{1 / 2} \wedge n^{\me})^2} \Big).
    \end{equation*}
    In particular,
    \begin{equation*}
        \begin{aligned}
            \sqrt{n} (\hat \delta_p - \mathfrak{b}_p) &= \mO_{L_k} \Big( \frac{\polyLog(n) n^{\alpha+1/2}}{(n^{1 / 2} \wedge n^{\me})^2} \Big),\\
            \sup_i |\bfx_i^\T (\hat \bdelta - \tilde \bfb)| &= \mO_{L_k} \Big( \frac{\polyLog(n) n^{\alpha+1/2}}{(n^{1 / 2} \wedge n^{\me})^2} \Big),\\
            \sup_i |R_i - r_{i, [p]}| &= \mO_{L_k}\bigg(\Big\{\frac{\polyLog(n)}{\sqrt{n} \wedge n^{\me}}\Big\} \vee \Big\{ \frac{\polyLog(n) n^{\alpha+1/2}}{(n^{1 / 2} \wedge n^{\me})^2} \Big\} \bigg) .
        \end{aligned}
    \end{equation*}
\end{theorem}

% We note that the index $p$ in the previous theorem plays no particular role and similar results holds when $p$ is replaced by any $k$ in the range $1 \leq k \leq p$.

\begin{proof}
    The first two results are direct consequences of all our results, using the key bound on $\|\hat \bdelta -\tilde \bfb\|$ of Proposition \ref{prop:bias_coordinate_bound}.

    The third result is a direct consequence of the fact that $\sup_i \|\cX_i\| / n^{1/2} = \mO_{L_k} (1)$, which was shown in the proof of Lemma \ref{lemma:3.7_2018}.

    The last result follows from the fact that $R_i - r_{i, [p]} = \bfx_i^\T (\hat \bdelta - \tilde \bfb) - \nu_i$.
    The result on $\nu_i$ is given in the proof of Proposition \ref{prop:md_bound}.
\end{proof}

Recall Equation \eqref{eq:el_3.18_ii}
\begin{equation*}
    \begin{aligned}
        \Big( \frac{p}{n} \Big)^2 n \E [ \{ \mathfrak{b}_p - \delta_0(p) + \hat w(p) - w_0(p)\}^2] =& \frac{1}{n} \sum_{i=1}^n \E \Big[ \{ \mc_{\tau, p} \lambda_i \psi(r_{i,[p]})\}^2 \Big] \\
        &+ n \tau^2 \{ \delta_0(p) -  \hat w(p) + w_0(p) \}^2 \E (\mc_{\tau, p}^2) + \mo(1).
    \end{aligned}
\end{equation*}
and Equation \eqref{eq:wbias}
\begin{equation*}
    |\mathfrak{b}_p-\delta_0(p) + \hat w(p) - w_0(p)| = \mO_{L_k} \Big( \frac{\polyLog(n) }{ \sqrt{n} \wedge n^{\me}} \Big).
\end{equation*}
We have
\begin{equation*}
    \begin{aligned}
        &\Big( \frac{p}{n} \Big)^2 n \E [ \{ \hat\delta_p - \delta_0(p) + \hat w(p) - w_0(p)\}^2] \\
        &\qquad \qquad= \Big( \frac{p}{n} \Big)^2 n \E [ \{ \hat\delta_p - \mathfrak{b}_p + \mathfrak{b}_p - \delta_0(p) + \hat w(p) - w_0(p)\}^2] \\
        &\qquad \qquad = \frac{1}{n} \sum_{i=1}^n \E \Big[ \{ \mc_{\tau, p} \lambda_i \psi(r_{i,[p]})\}^2 \Big] + n \tau^2 \{ \delta_0(p) -  \hat w(p) + w_0(p) \}^2 \E (\mc_{\tau, p}^2) + \mo(A),
    \end{aligned}
\end{equation*}
where $A = \frac{1}{n} \sum_{i=1}^n \E \Big[ \{ \mc_{\tau, p} \lambda_i \psi(r_{i,[p]})\}^2 \Big] + n \tau^2 \{ \delta_0(p) -  \hat w(p) + w_0(p) \}^2 \E (\mc_{\tau, p}^2)$.

We note that the index $p$ in Equation \eqref{eq:el_3.18_ii} and the previous theorem plays no particular role and similar results holds when $p$ is replaced by any $k$ in the range $1 \leq k \leq p$.
Summing over all the coordinates, we have
% \begin{equation*}
%     \Big(\frac{p}{n} \Big)^2 \E ( \| \Vec{\mathfrak{b}} - \bdelta_0 + \hat\bfw - \bfw_0\|^2) = \frac{1}{n} \sum_{k=1}^p \bigg(\frac{1}{n} \sum_{i=1}^n \E \Big[ \{ \mc_{\tau, k} \lambda_i \psi(r_{i,[k]})\}^2 \Big] \bigg) + \tau^2 \sum_{k=1}^p \delta_0^2(k) \E (\mc_{\tau, k}^2) + \mo(1),
% \end{equation*}
% where $\Vec{\mathfrak{b}} = (\mathfrak{b}_1, \ldots, \mathfrak{b}_p)^\T$.
under Assumptions \textbf{O1-O7} and \textbf{P1-P4},
\begin{equation*}
    \begin{aligned}
        \Big(\frac{p}{n} \Big)^2 \E ( \| \hat \bdelta - \bdelta_0 + \hat\bfw - \bfw_0\|^2)
        =& \frac{1}{n} \sum_{k=1}^p \bigg(\frac{1}{n} \sum_{i=1}^n \E \Big[ \{ \mc_{\tau, k} \lambda_i \psi(r_{i,[k]})\}^2 \Big] \bigg) \\
        &+ \tau^2 \sum_{k=1}^p \{ \delta_0(k) -  \hat w(k) + w_0(k) \}^2 \E (\mc_{\tau, k}^2) + \mo(B),
    \end{aligned}
\end{equation*}
where $B = \frac{1}{n} \sum_{k=1}^p \big(\frac{1}{n} \sum_{i=1}^n \E \big[ \{ \mc_{\tau, k} \lambda_i \psi(r_{i,[k]})\}^2 \big] \big) + \tau^2 \sum_{k=1}^p \{ \delta_0(k) -  \hat w(k) + w_0(k) \}^2 \E (\mc_{\tau, k}^2)$.

Note that $0 \leq \mc_{\tau, k} \leq \kappa/\tau$, $n^{-1} \sum_{i=1}^n \|\psi^2\|_{\infty} \leq C$ from Assumption \textbf{O3}, and $\|\bdelta_0\|^2 = \mO(1)$.
We have $B = \mO(1)$.
Therefore, we have
\begin{equation*}
    \begin{aligned}
        \Big(\frac{p}{n} \Big)^2 \E ( \| \hat \bdelta - \bdelta_0 + \hat\bfw - \bfw_0\|^2)
        =& \frac{1}{n} \sum_{k=1}^p \bigg(\frac{1}{n} \sum_{i=1}^n \E \Big[ \{ \mc_{\tau, k} \lambda_i \psi(r_{i,[k]})\}^2 \Big] \bigg) \\
        &+ \tau^2 \sum_{k=1}^p \{ \delta_0(k) -  \hat w(k) + w_0(k) \}^2 \E (\mc_{\tau, k}^2) + \mo(1).
    \end{aligned}
\end{equation*}

 \paragraph{On $\mc_{\tau, k}$ and $c_\tau$}~{}

 We now show that $\mc_{\tau, k}$ and $c_\tau$ are close to each other..

 \begin{proposition}%[{{\citet[Proposition 3.21]{el2018impact}}}]
    \label{prop:mctau_ctau_gap}
    We have, under Assumptions \textbf{O1-O7} and \textbf{P1-P3},
    \begin{equation*}
        \sup_{1 \leq k \leq p} |\mc_{\tau, k} - c_\tau| = \mO_{L_k} \bigg( \Big[\frac{\polyLog(n) n^\alpha}{\sqrt{n} \wedge n^{\me}} \Big] \vee \Big[\frac{\polyLog(n) n^{2 \alpha+1 / 2}}{(n^{1 / 2} \wedge n^{\me})^2}\Big] \vee \frac{\polyLog(n)}{n} \bigg) .
    \end{equation*}
    Of course, we also have $0 \leq c_\tau \leq p/(n\tau)$ and $0 \leq \mc_{\tau, k} \leq p/(n\tau)$.
 \end{proposition}

 \begin{proof}
    Recall that
    \begin{equation*}
        \bfS = \frac{1}{n} \sum_{i=1}^{n} \psi'(R_i) \bfx_i \bfx_i^\T.
    \end{equation*}
    Denote $\boldsymbol{\Gamma} = n^{-1} \sum_{i=1}^{n} \psi'(R_i) \bfv \bfv^\T$ and $a = n^{-1} \sum_{i=1}^{n} \psi'(R_i) x_i^2(p)$, we have
    \begin{equation*}
        \bfS = \left(\begin{array}{ll}
            \boldsymbol{\Gamma} & \mathbf{v} \\
            \mathbf{v} & a
            \end{array}\right) .
    \end{equation*}
    According to Lemma 3.40 in \cite{el2018impact}, we have, since $c_\tau = n^{-1} \tr \{ (\bfS + \tau \I_p)^{-1} \}$,
    \begin{equation*}
        \Big| c_\tau - \frac{1}{n} \tr \{ (\boldsymbol{\Gamma} + \tau \I_{p-1})^{-1} \} \Big| \leq \frac{1}{n} \frac{1 + a/\tau}{\tau}.
    \end{equation*}
    We also have
    \begin{equation*}
        a = \frac{1}{n} \sum_{i=1}^n \lambda_i^2 \cX_i^2(p) \psi'(R_i) \leq \polyLog(n) \frac{1}{n} \sum_{i=1}^n \lambda_i^2 \cX_i^2(p) = \mO_{L_k} (\polyLog(n) ).
    \end{equation*}
    Since $\psi'$ is $C n^\alpha$-Lipschitz and
    \begin{equation*}
        \sup_i |R_i - r_{i, [p]}| = \mO_{L_k}\bigg(\Big\{\frac{\polyLog(n)}{\sqrt{n} \wedge n^{\me}}\Big\} \vee \Big\{ \frac{\polyLog(n) n^{\alpha+1/2}}{(n^{1 / 2} \wedge n^{\me})^2} \Big\} \bigg),
    \end{equation*}
    we have
    \begin{equation*}
        \sup_i | \psi'(R_i) - \psi'(r_{i, [p]}) | = \mO_{L_k}\bigg(\Big\{\frac{\polyLog(n) n^\alpha}{\sqrt{n} \wedge n^{\me}}\Big\} \vee \Big\{ \frac{\polyLog(n) n^{2\alpha+1/2}}{(n^{1 / 2} \wedge n^{\me})^2} \Big\} \bigg),
    \end{equation*}
    Hence, using arguments similar to those in the proof of Lemma \ref{lemma:zeta_bound}, we have
    
    \begin{equation*}
        \Big| \frac{1}{n} \tr \{ (\boldsymbol{\Gamma} + \tau \I_{p-1})^{-1} \} - \frac{1}{n} \tr \{ (\mathfrak{S}_p + \tau \I_{p})^{-1} \} \Big| = \mO_{L_k} \bigg( \Big[\frac{\polyLog(n) n^\alpha}{\sqrt{n} \wedge n^{\me}} \Big] \vee \Big[\frac{\polyLog(n) n^{2 \alpha+1 / 2}}{(n^{1 / 2} \wedge n^{\me})^2}\Big] \bigg).
    \end{equation*}

    Since $\mc_{\tau, k} = n^{-1} \tr \{ (\mathfrak{S}_p + \tau \I_p)^{-1} \}$, the result follows immediately.

    We note that $p$ did not play a particular role in the proof and hence taking the sup over those indices only add a $\polyLog(n)$ term.
    Hence the result holds for $\sup_{1 \leq k \leq n} |\mc_{\tau, k} - c_\tau|$.
 \end{proof}

 We are now ready to prove the last proposition of this section, which will help us get the second equation of our System.

\begin{proposition}
    \label{prop:second_system}
    \begin{equation}
        \label{eq_system_2}
        \begin{aligned}
            \Big(\frac{p}{n} \Big)^2 \E ( \| \hat \bdelta - \bdelta_0 + \hat\bfw - \bfw_0\|^2) =& \frac{p}{n} \frac{1}{n} \sum_{i=1}^n \E \Big[ \{ c_{\tau} \lambda_i \psi(\prox(c_\tau \lambda_i^2 \rho) (\tilde r_{i,(i)}))\}^2 \Big] \\
            &+ \tau^2 \|\bdelta_0 + \bfw_0 - \hat \bfw\|^2 \E (c_{\tau}^2) + \mo(1).
        \end{aligned}
    \end{equation}
    Furthermore, when all $\lambda_i$'s are non-zero,
    \begin{equation*}
        \frac{1}{n} \sum_{i=1}^n \E \Big[ \{ c_{\tau} \lambda_i \psi(\prox(c_\tau \lambda_i^2 \rho) (\tilde r_{i,(i)}))\}^2 \Big] = \frac{1}{n} \sum_{i=1}^n \E \bigg( \frac{\{ \tilde r_{i,(i)} - \prox(c_\tau \lambda_i^2 \rho)(\tilde r_{i,(i)})\}^2}{\lambda_i^2} \bigg).
    \end{equation*}
\end{proposition}

\begin{proof}
    In light of Proposition \ref{prop:mctau_ctau_gap} and Assumption \textbf{P3} which guarantees that $\|\bdelta_0\|^2$ is uniformly bounded in $p$ and $n$, we have
    \begin{equation*}
        \sum_{k=1}^p \{ \delta_0(k) -  \hat w(k) + w_0(k) \}^2 \E (\mc_{\tau, k}^2) =  \|\bdelta_0 + \bfw_0 - \hat \bfw\|^2 \E (c_{\tau}^2) + \mo(1).
    \end{equation*}

    Using Theorem \ref{thm:delta_tildeb_bound} and the bound on $\|\psi'\|_\infty$ in Assumption \textbf{O3},  we have
    \begin{equation*}
        \frac{1}{p} \sum_{k=1}^p \E \Big[ \{ \mc_{\tau, k} \lambda_i \psi(r_{i,[k]})\}^2 \Big] = \frac{1}{p} \sum_{k=1}^p \E \Big[ \{ \mc_{\tau, k} \lambda_i \psi(R_i)\}^2 \Big] + \mo(1).
    \end{equation*}
    With the help of Proposition \ref{prop:mctau_ctau_gap}, we have
    \begin{equation*}
        \frac{1}{p} \sum_{k=1}^p \E \Big[ \{ \mc_{\tau, k} \lambda_i \psi(R_i)\}^2 \Big] = \frac{1}{p} \sum_{k=1}^p \E \Big[ \{ c_{\tau} \lambda_i \psi(R_i)\}^2 \Big] + \mo(1) = \E \Big[ \{ c_{\tau} \lambda_i \psi(R_i)\}^2 \Big] + \mo(1).
    \end{equation*}
    In light of Equation \eqref{eq:el_27}, we have
    \begin{equation*}
        \frac{1}{n} \sum_{i=1}^n \E \{ (c_\tau \lambda_i \psi(R_i))^2 \} = \frac{1}{n} \sum_{i=1}^n \E \Big( [ c_\tau \lambda_i \psi\{ \prox(c_i \rho) (\tilde r_{i,(i)}) \} ]^2 \Big) + \mo(1).
    \end{equation*}

    Lemma \ref{lemma:prox_c_derivative} gives the computation of the derivative of $\prox(c \rho)(x)$ with respect to $c$, which allows us to bound the error $|\psi\{ \prox(c_i \rho)(x)\} - \psi\{ \prox(c_\tau \lambda_i^2 \rho)(x)\}|$.
    In light of this, by using the fact that $\sup_i |c_i - \lambda_i^2 c_\tau| = \mO_{L_k} (n^{2\alpha - 1/2} \polyLog(n))$ in Corollary \ref{cor:ci_ctau_gap}, we can re-express the previous equation as
    \begin{equation*}
        \frac{1}{n} \sum_{i=1}^n \E \{ (c_\tau \lambda_i \psi(R_i))^2 \} = \frac{1}{n} \sum_{i=1}^n \E \Big( [ c_\tau \lambda_i \psi\{ \prox(c_\tau \lambda_i^2 \rho) (\tilde r_{i,(i)}) \} ]^2 \Big) + \mo(1).
    \end{equation*}

    When all $\lambda_i$'s are non-zero, we have
    \begin{equation*}
        \frac{1}{n} \sum_{i=1}^n \E \{ (c_\tau \lambda_i \psi(R_i))^2 \} = \frac{1}{n} \sum_{i=1}^n \E \bigg( \frac{[ c_\tau \lambda_i^2 \psi\{ \prox(c_\tau \lambda_i^2 \rho) (\tilde r_{i,(i)}) \} ]^2}{\lambda_i^2} \bigg) + \mo(1).
    \end{equation*}
    Finally, since by definition,
    \begin{equation*}
        \forall x \in \RR, \quad x = \prox(c \rho)(x) + c\psi(\prox(c \rho)(x)),
    \end{equation*}
    we have
    \begin{equation*}
        \frac{1}{n} \sum_{i=1}^n \E \bigg( \frac{[ c_\tau \lambda_i^2 \psi\{ \prox(c_\tau \lambda_i^2 \rho) (\tilde r_{i,(i)}) \} ]^2}{\lambda_i^2} \bigg) = \frac{1}{n} \sum_{i=1}^n \E \bigg( \frac{[ \tilde r_{i,(i)} - \prox(c_\tau \lambda_i^2 \rho) (\tilde r_{i,(i)}) ]^2}{\lambda_i^2} \bigg).
    \end{equation*}
\end{proof}

\begin{lemma}[{{\citet[Lemma 3.32]{el2018impact}}}]
    \label{lemma:prox_c_derivative}
    Suppose $x$ is a real and $\rho$ is twice differentiable and convex. Then, for $c>0$, we have
$$
\frac{\partial}{\partial c} \operatorname{prox}(c \rho)(x)=-\frac{\psi(\operatorname{prox}(c \rho)(x))}{1+c \psi'(\operatorname{prox}(c \rho)(x))},
$$
and
$$
\frac{\partial}{\partial c} \rho(\operatorname{prox}(c \rho)(x))=-\frac{\psi^2(\operatorname{prox}(c \rho)(x))}{1+c \psi'(\operatorname{prox}(c \rho)(x))} .
$$

In particular, at $x$ given $c \rightarrow \rho(\operatorname{prox}(c \rho)(x))$ is decreasing in $c$.
\end{lemma}

\subsection[\appendixname~\thesubsection]{Last steps of the proof} \label{sec_p3}

\subsubsection{On the asymptotic behavior of \texorpdfstring{$\tilde r_{i,(i)}$}{rii}}

We have the following result.

\begin{lemma}
    \label{lemma:rii_asymptotic}
    Under Assumptions \textbf{O1-O7} and \textbf{P1-P4}, as $n$ and $p$ tend to infinity, $\tilde r_{i,(i)} = \epsilon_i + \bfx_i^\T (\bfw_0 - \hat \bfw) - \bfx_i^\T (\hat \bdelta_{(i)} - \bdelta_0)$ behaves like $\epsilon_i + \lambda_i \sqrt{\E (\| \hat \bbeta - \bbeta_0\|^2)} Z_i$.
    where $Z_i \sim N(0,1)$ is independent of $\epsilon_i$ and $\lambda_i$, in the sense of weak convergence.

    Furthermore, if $i \neq j$, $\tilde r_{i,(i)}$ and $\tilde r_{j,(j)}$ are asymptotically (pairwise) independent.
    The same is true for $(\tilde r_{i,(i)}, \lambda_i)$ and $(\tilde r_{j,(j)}, \lambda_j)$.
\end{lemma}

\begin{proof}
    % Using the fact that $\bfx_i^\T (\bfw_0 - \hat \bfw) = \mo_p(1)$, we have $\tilde r_{i,(i)} = \epsilon_i - \lambda_i \cX_i^\T (\hat \bdelta_{(i)} - \bdelta_0) + \mo_p(1)$.

    \textbf{First part.}

    In this part, we will show that $\cX_i^\T (\hat \bdelta_{(i)} - \bdelta_0 + \hat \bfw - \bfw_0)$ is asymptotically $N(0, \E (\| \hat \bbeta - \bbeta_0\|^2))$.
    Recall that $\hat \bdelta_{(i)} $ is independent of $\cX_i$ and that $\E (\cX_i) = \bzero$, $\cov (\cX_i) = \I$ and that, for any finite $k$, the first $k$ moments of its entries are bounded uniformly in $n$.

    We have shown that in Proposition \ref{prop:delta_variance} that $\var (\| \hat \bbeta - \bbeta_0\|^2) \to 0$.
    In light of Lemma \ref{lemma:delta_moment_bounds}, we also know that $\E (\| \hat \bbeta - \bbeta_0\|^2)$ is uniformly bounded.
    Furthermore, in the proof of Proposition \ref{prop:delta_variance}, we showed that $\E (\|\hat \bdelta + \hat \bfw - \bdelta_0 - \bfw_0 \|^2 - \|\hat\bdelta_{(i)} + \hat \bfw - \bdelta_0 - \bfw_0\|^2) \to 0$.

    We use a simple generalization of the standard Lindeberg-Feller theorem (see e.g. \cite{stroock2010probability}).
    Indeed, if $a_{n,p}(k)$ are random variables with $\sqrt{\sum_{k=1}^{p} a_{n,p}(k)^2} = A_n$, $\E(A_n^2)$ remains bounded in $n$, and $a_{n,p}(k)$'s are independent of $\cX_i$, we see that:
    a) if $\bfz \sim N(\bzero, \I_p)$, independent of $a_{n,p}(k)$, then $\bfa_{n,p}^\T \bfz \sim A_n \mathrm{N}$ where $\mathrm{N} \sim  N(0, 1)$ and independent of $A_n$ (conditionally and unconditionally on $\bfa_{n,p}$);
    b) Theorem 2.1.5 and its proof in \cite{stroock2010probability} hold provided that $\sum_{i=1}^n \E(|a_{n,p}(k)|^3) = \mo(1)$.
    The proof simply needs to be started conditionally on $a_{n,p}(k)$, and the final moment bounds are then taken unconditionally.
    This very mild generalization gives, if $\phi$ is a $C^3$ function, with bounded $2$nd and $3$rd derivatives,
    \begin{equation*}
        \begin{aligned}
        & \forall \epsilon>0,\big|\E\{\phi(\bfa_{n, p}^{\T} \mathcal{X}_i)\}-\E\{\phi(A_n \mathrm{N})\} \big| \\
        & \qquad \leq K \bigg[\epsilon\|\phi^{(3)}\|_{\infty} \E\Big\{\sum_{k=1}^p a_{n, p}(k)^2\Big\} + \frac{\|\phi^{(2)}\|_{\infty}}{\epsilon} \sum_{k=1}^p \E(|a_{n, p}(k)|^3) \bigg],
        \end{aligned}
    \end{equation*}
    where $K$ is a constant depending on the second and third absolute moments of the entries of $\cX_i$.
    It is therefore independent of $n$ and $p$ under our assumptions on $\cX_i$.

    We can apply this result to $a_{n,p}(k) = \hat \delta_{(i)}(k) - \delta_0(k) + \hat w(k) - w_0(k)$.
    Recall that we have shown that
    \begin{equation*}
        \hat \delta(k) - \delta_0(k) + \hat w(k) - w_0(k) = \mO_{L_k} \Big( \frac{\polyLog(n)}{n^{1 / 2} \wedge n^{\me}} \Big)
    \end{equation*}
    for each coordinate $k$. The same arguments apply also to $\hat \delta_{(i)} (k)$, the $k$th coordinate of $\hat \bdelta_{(i)}$.
    We have
    \begin{equation*}
        \E (|\hat \delta_{(i)}(k) - \delta_0(k) + \hat w(k) - w_0(k)|^3) = \mO \Big\{ \frac{\polyLog(n)}{(n^{1 / 2} \wedge n^{\me})^3} \Big\}
    \end{equation*}
    Provided that $\me > 1/3$, we have
    \begin{equation*}
        \E \Big( \sum_{k=1}^p |\hat \delta_{(i)}(k) - \delta_0(k) + \hat w(k) - w_0(k)|^3 \Big) = \mO \Big\{ \frac{\polyLog(n) n}{(n^{1 / 2} \wedge n^{\me})^3} \Big\} = \mo(1).
    \end{equation*}
    Therefore, in connection with Corollary 2.1.9 in \cite{stroock2010probability}, $\cX_i^\T (\hat \bdelta_{(i)} - \bdelta_0 + \hat \bfw - \bfw_0)$ behaves asymptotically like $\|\hat \bdelta_{(i)} - \bdelta_0 + \hat \bfw - \bfw_0\| \mathrm{N}$ in the sense of weak convergence.

    Since $\|\hat \bdelta_{(i)} - \bdelta_0 + \hat \bfw - \bfw_0\| - \E \|\hat \bdelta_{(i)} - \bdelta_0 + \hat \bfw - \bfw_0\| \to 0$ in probability and the expectations are uniformly bounded, Slutsky's lemma gives that
    \begin{equation*}
        \cX_i^\T (\hat \bdelta_{(i)} - \bdelta_0 + \hat \bfw - \bfw_0) \text{ behaves like } \E (\|\hat \bdelta_{(i)} - \bdelta_0 + \hat \bfw - \bfw_0\|) \mathrm{N}
    \end{equation*}
    asymptotically in the sense of weak convergence.
    Using the fact that $\E (\|\hat \bdelta - \bdelta_0 + \hat \bfw - \bfw_0\|^2 - \|\hat\bdelta_{(i)} - \bdelta_0 + \hat \bfw - \bfw_0\|^2) \to 0$, another application of Slutsky's lemma yields
    \begin{equation*}
        \cX_i^\T (\hat \bdelta_{(i)} - \bdelta_0 + \hat \bfw - \bfw_0) \text{ behaves like } \E (\|\hat \bdelta - \bdelta_0 + \hat \bfw - \bfw_0\|) \mathrm{N}
    \end{equation*}
    in the sense of weak convergence.

    We note that the same reasoning applies when replacing $a_{n,p} (k) = \hat \delta_{(i)}(k) - \delta_0(k) + \hat w(k) - w_0(k)$ by $\tilde a_{n,p} (k) = \lambda_i \{\hat \delta_{(i)}(k) - \delta_0(k) + \hat w(k) - w_0(k)\}$, provided that $\lambda_i$ has $3$ moments.
    It shows that
    \begin{equation*}
        \cX_i^\T (\hat \bdelta_{(i)} - \bdelta_0 + \hat \bfw - \bfw_0) \text{ behaves like } \lambda_i \E (\|\hat \bbeta - \bbeta_0\|) \mathrm{N}.
    \end{equation*}
    This shows the first part of the lemma, since $\E (\|\hat \bbeta - \bbeta_0\|) = \sqrt{\E (\|\hat \bbeta - \bbeta_0\|^2)} + \mo(1)$ by Proposition \ref{prop:delta_variance}.

    \textbf{Second part.}

    For the second part, we use a leave-two-out approach. More precisely, we use the approximation
    \[
        \tilde r_{i,(i)} = \epsilon_i + \bfx_i^\T (\bfw_0 - \hat \bfw) - \bfx_i^\T (\hat \bdelta_{(i)} - \bdelta_0)
        = \epsilon_i + \bfx_i^\T (\bfw_0 - \hat \bfw) - \bfx_i^\T (\hat \bdelta_{(ij)} - \bdelta_0) + \mo_{L_k} (1),
    \]
    and similarly for $\tilde r_{j,(j)}$, which follows from Theorem \ref{th:l1oo}.
    Here $\hat \bdelta_{(ij)}$ is computed by solving Problem \eqref{eq:opt_ori} without $(\bfx_i, y_i)$ and $(\bfx_j, y_j)$.
    It is clear that $\tilde r_{i,(i)}$ and $\tilde r_{j,(j)}$ are asymptotically independent conditional on $\bfX_{(ij)}$, i.e., the design matrix without the $i$-th and $j$-th rows.

    Similarly, we have
    \begin{equation*}
        \E \Big( \sum_{k=1}^p |\hat \delta_{(ij)}(k) - \delta_0(k) + \hat w(k) - w_0(k)|^3 \Big) = \mO \Big\{ \frac{\polyLog(n) n}{(n^{1 / 2} \wedge n^{\me})^3} \Big\} = \mo(1).
    \end{equation*}
    Therefore, we also have
    \begin{equation*}
        \sum_{k=1}^p |\hat \delta_{(ij)}(k) - \delta_0(k) + \hat w(k) - w_0(k)|^3 = \mo_p(1).
    \end{equation*}
    Note that $\hat \bdelta_{(ij)}$ depends only on $\{\bfX_{(ij)}, \bepsilon_{(ij)}\}$.
    We call $P_{(ij)}$ the joint probability measure $P_{(ij)} = \prod_{k \neq(i, j)} P_{\bfx_k, \epsilon_k}$, i.e., probability computed with respect to all our random variables except $(\bfx_i, \epsilon_i)$ and $(\bfx_j, \epsilon_j)$.

    So we have found $E_{(ij)}^n$, which depends only on $(\bfX_{(ij)}, \bepsilon_{(ij)})$, such that $P_{(ij)} (E_{(ij)}^n) \to 1$ and $\sum_{k=1}^p |\hat \delta_{(ij)}(k) - \delta_0(k) + \hat w(k) - w_0(k)|^3 = \mo(1)$ when $(\bfX_{(ij)}, \{\epsilon_k\}_{k \neq (i,j)}) \in E_{(ij)}^n$.
    By treating $a_{n,p}$'s as deterministic quantities, the arguments we gave above then imply that, when $(\bfX_{(ij)}, \bepsilon_{(ij)}) \in E_{(ij)}^n$,
    \begin{equation*}
        \cX_i^\T (\hat \bdelta_{(ij)} - \bdelta_0 + \hat \bfw - \bfw_0) | (\bfX_{(ij)}, \bepsilon_{(ij)}) \text{ behaves like } (\|\hat \bdelta_{(ij)} - \bdelta_0\|) \mathrm{N}.
    \end{equation*}

    We now use characteristic function arguments.
    Let $\alpha_i = \cX_i^\T (\hat \bdelta_{(ij)} - \bdelta_0 + \hat \bfw - \bfw_0)$ and $\alpha_j = \cX_j^\T (\hat \bdelta_{(ij)} - \bdelta_0 + \hat \bfw - \bfw_0)$.

    Let $(w_i, w_j) \in \RR^2$ be fixed and
    \begin{equation}
        \chi(w_i, w_j)=\E\big\{e^{i(w_1 \alpha_i+w_2 \alpha_j)}\big\}=\E \Big\{e^{i(w_1 \alpha_i+w_2 \alpha_j)}\Big(1_{E_{(i j)}^n}+1_{[E_{(i j)}^n]^c}\Big)\Big\} .
    \end{equation}
    Since $P (E_{(ij)}^n) = P_{(ij)} (E_{(ij)}^n) \to 1$, we can just focus on $\E \{e^{i(w_1 \alpha_i+w_2 \alpha_j)} 1_{E_{(i j)}^n} \}$, since the modulus of the functions we are integrating is bounded by $1$.

    Now, we have
    \begin{equation*}
        \E \Big\{e^{i(w_1 \alpha_i+w_2 \alpha_j)} 1_{E_{(i j)}^n} \Big\} = \E \Big[1_{E_{(i j)}^n} \E \Big\{e^{i(w_1 \alpha_i+w_2 \alpha_j)} | \bfX_{(ij)}, \bepsilon_{(ij)} \Big\} \Big],
    \end{equation*}
    since $1_{E_{(i j)}^n}$ is a deterministic function of $(\bfX_{(ij)}, \bepsilon_{(ij)})$.
    Independence of $\cX_i$ and $\cX_j$ gives
    \begin{equation*}
        \E \Big\{e^{i(w_1 \alpha_i+w_2 \alpha_j)} | \bfX_{(ij)}, \bepsilon_{(ij)} \Big\} = \E \Big(e^{i w_1 \alpha_i} | \bfX_{(ij)}, \bepsilon_{(ij)} \Big) \E \Big(e^{i w_2 \alpha_j} | \bfX_{(ij)}, \bepsilon_{(ij)} \Big).
    \end{equation*}
    Also, the conditional Gaussian approximation established above implies that
    \begin{equation*}
        1_{E_{(i j)}^n} \Big\{ \E \Big(e^{i(w_1 \alpha_i+w_2 \alpha_j)} | \bfX_{(ij)}, \bepsilon_{(ij)} \Big) - e^{-(w_1^2/2 + w_2^2/2) \| \hat \bdelta_{(ij)} - \bdelta_0\|^2} \Big\} \to 0
    \end{equation*}
    in $P_{(ij)}$-probability.

    So we conclude that
    \begin{equation*}
        \E \Big\{ 1_{E_{(i j)}^n} e^{i(w_1 \alpha_i+w_2 \alpha_j)} \Big\} - \E \Big\{ 1_{E_{(i j)}^n} e^{-(w_1^2/2 + w_2^2/2) \| \hat \bdelta_{(ij)} - \bdelta_0\|^2} \Big\} \to 0.
    \end{equation*}

    Since $P(E_{(ij)}^n) \to 1$ and $\| \hat \bdelta_{(ij)} - \bdelta_0\|^2$ is asymptotically deterministic by arguments similar to those in the proof of Proposition \ref{prop:delta_variance}, we have
    \begin{equation*}
        \E \Big\{ 1_{E_{(i j)}^n} e^{-(w_1^2/2 + w_2^2/2) \| \hat \bdelta_{(ij)} - \bdelta_0\|^2} \Big\} - e^{-(w_1^2/2 + w_2^2/2) \E (\| \hat \bdelta_{(ij)} - \bdelta_0\|^2)} \to 0.
    \end{equation*}

    Therefore,
    \begin{equation*}
        \E\big\{e^{i(w_1 \alpha_i+w_2 \alpha_j)}\big\} - \E \big( e^{iw_1 \alpha_i} \big) \E \big( e^{iw_2 \alpha_j} \big) \to 0.
    \end{equation*}
    This shows that $\alpha_i$ and $\alpha_j$ are asymptotically independent.
    It easily follows that $\tilde r_{i,(i)}$ and $\tilde r_{j,(j)}$ are asymptotically independent.

    The same reasoning applies to $(\tilde r_{i,(i)}, \lambda_i)$ and $(\tilde r_{j,(j)}, \lambda_j)$, since $\hat \bdelta_{(ij)}$ is independent of $\lambda_i$ and $\lambda_j$ under Assumption \textbf{O6}.

\end{proof}

\subsubsection{On the asymptotic behavior of \texorpdfstring{$c_\tau$}{ctau}}

We now show that $c_\tau = n^{-1} \tr\{ (\bfS + \tau \I_p)^{-1} \}$ is asymptotically deterministic.
This require several steps.

\begin{lemma}%{\citet[Lemma 3.24]{el2018impact}}
    \label{lemma:gn_equicontinuity}
    We work under Assumptions \textbf{O1-O7}, \textbf{P1-P4} and \textbf{F2-F4}.
    Consider the random function
    \begin{equation*}
        g_n(x) = \frac{1}{n} \sum_{i=1}^n \frac{1}{1 + x \lambda_i^2 \psi' \{\prox(x \lambda_i^2 \rho) (\tilde r_{i,(i)})\}}, \text{ defined for } x \geq 0.
    \end{equation*}

    Let $B>0$ be in $\RR_+$.
    We have, for any $(x, y) \in \RR_+^2$, and $0 \leq x \leq B$, $0 \leq y \leq B$,
    \begin{equation*}
        \sup_{(x,y): |x-y| \leq \eta, 0 \leq x \leq B, 0 \leq y \leq B} |g_n(x) - g_n(y)| \leq \eta \frac{1}{n} \sum_{i=1}^n (\lambda_i^2 \| \psi'\|_\infty + B \mathrm{L}(n) \lambda_i^4 \|\psi\|_\infty).
    \end{equation*}
    In particular, under Assumption \textbf{P2} and \textbf{F3-F4}, we have for $\mC$ a constant independent of $n$ and $p$,
    \begin{equation}
        \label{eq_bound_gn}
        {\P}^* \Big( \sup_{(x,y): |x-y| \leq \eta, 0 \leq x \leq B, 0 \leq y \leq B} |g_n(x) - g_n(y)| > \delta \Big) \leq \frac{\eta}{\delta} \mC .
    \end{equation}
    Hence, $g_n$ is stochastic equicontinuous on $[0, B]$ for any $B>0$ given.

\end{lemma}

We used the notation ${\P}^*$ above to denote outer probability and avoid a discussion of potential measure theoretic issues associated with taking a supremum over a noncountable collection of random variables (see e.g. \citet[Sect. 18.2]{vd1998asymptotic}).

\begin{proof}
    Consider the function defined for $x \geq 0$,
    \begin{equation*}
        h_u^{(i)} (x) = \frac{1}{1 + x \lambda_i^2 \psi' \{\prox(x \lambda_i^2 \rho) (u)\}} = \frac{\partial}{\partial u} \prox(x \lambda_i^2 \rho)(u).
    \end{equation*}
    The last equality comes from Lemma 3.33 from \citet{el2018impact}.

    Since $\psi'$ is non-negative,
    \begin{equation*}
        \forall u, \quad |h_u^{(i)}(x)-h_u^{(i)}(y)| \leq|x \lambda_i^2 \psi'(\prox(x \lambda_i^2 \rho)(u))-y \lambda_i^2 \psi'(\prox(y \lambda_i^2 \rho)(u))| \wedge 1 .
    \end{equation*}
    Thus, since $x, y \geq 0$, for all u,
    \begin{equation*}
        \begin{aligned}
            |h_u^{(i)}(x)-h_u^{(i)}(y)| \leq& \lambda_i^2|x-y| \psi'(\prox(x \lambda_i^2 \rho)(u))\\
            &+\lambda_i^2 y |\psi'(\prox(x \lambda_i^2 \rho)(u)) -\psi'(\prox(y \lambda_i^2 \rho)(u)) |.
        \end{aligned}
    \end{equation*}
    In particular, if $|x-y| \leq \eta$, and $x \vee y \leq B$, with $x, y \geq 0$, for all $u$,
    
    \begin{equation*}
        \begin{aligned}
            \sup_{y: |x-y| \leq \eta; x \vee y \leq B} \Big| h_u^{(i)}(x) - h_u^{(i)}(y) \Big| \leq& \lambda_i^2 \eta \psi' \{ \prox(x \lambda_i^2 \rho)(u) \} \\
            &+ B \lambda_i^2 \sup_{y: |x-y| \leq \eta, x \vee y \leq B} | \psi' \{ \prox(x \lambda_i^2 \rho)(u) \} - \psi' \{ \prox(y \lambda_i^2 \rho)(u) \} |.
        \end{aligned}
    \end{equation*}

    Under assumption \textbf{O3}, $\psi'$ is $\mathrm{L}(n)$-Lipschitz.
    Therefore, for $x_i = x \lambda_i^2$, $y_i = y \lambda_i^2 \geq 0$, we have
    \begin{equation*}
        \forall u, | \psi' \{ \prox(x_i \rho)(u) \} - \psi' \{ \prox(y_i \rho)(u) \} | \leq \mathrm{L}(n) |\prox(x_i \rho)(u) - \prox(y_i \rho)(u)|.
    \end{equation*}
    Recall that according to Lemma \ref{lemma:prox_c_derivative},
    \begin{equation*}
        \frac{\partial}{\partial x} \prox(x \rho)(u) = -\frac{\psi\{\prox(x \rho)(u)\}}{1 + x \psi'\{\prox(x \rho)(u)\}}.
    \end{equation*}
    Hence we have
    \begin{equation*}
        \sup_x \Big| \frac{\partial}{\partial x} \prox(x \rho)(u) \Big| \leq \|\psi\|_\infty.
    \end{equation*}
    We conclude that
    \begin{equation*}
        \forall u, \quad | \psi' \{ \prox(x_i \rho)(u) \} - \psi' \{ \prox(y_i \rho)(u) \} | \leq \{ \mathrm{L}(n) \|\psi\|_\infty |x_i - y_i| \}\wedge 2 \| \psi\|_\infty.
    \end{equation*}
    We therefore have, for $ x \vee y \leq B$ and $x, y \geq 0$,
    \begin{equation*}
        \forall u, \sup_{y: |x-y| \leq \eta} | h_u^{(i)}(x) - h_u^{(i)}(y) | \leq \lambda_i^2 \eta \psi' \{ \prox(x \lambda_i^2 \rho)(u) \} + B \lambda_i^4 \mathrm{L}(n) \|\psi\|_\infty \eta.
    \end{equation*}

    Thus, for $x,y \geq 0$,
    \begin{equation*}
        \forall u, \sup_{(x,y): |x-y| \leq \eta, x \vee y \leq B} | h_u^{(i)}(x) - h_u^{(i)}(y) | \leq \lambda_i^2 \eta \|\psi'\|_\infty + B \lambda_i^4 \mathrm{L}(n) \|\psi\|_\infty \eta.
    \end{equation*}
    Since the right-hand side does not depend on $u$, we have
    \begin{equation*}
        \sup_u \sup_{(x,y): |x-y| \leq \eta, x \vee y \leq B} | h_u^{(i)}(x) - h_u^{(i)}(y) | \leq \lambda_i^2 \eta \|\psi'\|_\infty + B \lambda_i^4 \mathrm{L}(n) \|\psi\|_\infty \eta.
    \end{equation*}
    Naturally, $g_n(x)$ can be written as
    \begin{equation*}
        g_n(x) = \frac{1}{n} \sum_{i=1}^n h_{\tilde r_{i,(i)}}^{(i)}(x).
    \end{equation*}
    Therefore, for any $x,y \geq 0$,
    \begin{equation*}
        | g_n(x) - g_n(y) | \leq \frac{1}{n} \sum_{i=1}^n | h_{\tilde r_{i,(i)}}^{(i)}(x) - h_{\tilde r_{i,(i)}}^{(i)}(y) |.
    \end{equation*}
    The bound we have obtained above on $\sup_u | h_u^{(i)}(x) - h_u^{(i)}(y) |$ when $x$ and $y$ are sufficiently close to one another can now be used.
    This shows that for $x$ given, if $x,y \geq 0$, $|x-y| \leq \eta$, and $x \vee y \leq B$, we have
    \begin{equation*}
        \sup_{(x,y): |x-y| \leq \eta, 0 \leq x \leq B, 0 \leq y \leq B} |g_n(x) - g_n(y)| \leq \eta \frac{1}{n} \sum_{i=1}^n (\lambda_i^2 \| \psi'\|_\infty + B \mathrm{L}(n) \lambda_i^4 \|\psi\|_\infty).
    \end{equation*}
    Under Assumption \textbf{P2} and \textbf{F3-F4}, all the terms on the right hand side are bounded in $L_1$.
    We can now take expectations and get the result in $L_1$.

\end{proof}

\begin{lemma}%{\citet[Lemma 3.25]{el2018impact}}
    \label{lemma:gn_uniform_lln}
    Let us call $G_n(x) = \E (g_n(x))$.
    Let $B>0$ be given.
    For any given $x_0 \leq B$,
    \begin{equation*}
        g_n(x_0) - G_n(x_0) = \mo_{L_2} (1).
    \end{equation*}

    Under Assumptions \textbf{O1-O7}, \textbf{P1-P4} and \textbf{F1-F5}, we have
    \begin{equation*}
        \E^* \Big( \sup_{0 \leq x \leq B} |g_n(x) - G_n(x)| \Big) \to 0.
    \end{equation*}

\end{lemma}

\begin{proof}
    Under assumptions \textbf{F1} and \textbf{F5}, we can divide the index set $\{1, \ldots, n\}$ into finite $K$ subsets $A_1, \ldots, A_K$, in which $(\bfx_i, \epsilon_i)_{i \in A_j}$ play a symmetric role.
    Hence, $\var(g_n(x_0))$ can be expressed as a sum of variances and covariances of finitely many functions of finitely many random variables $(\lambda_i, \tilde r_{i,(i)})$: for those random  variables, we just need to pick a representative in each subset $\{A_j\}_{j=1}^K$.

    Since $\psi'$ is Lipschitz, $g_n$ is an average of bounded continuous functions of the random variables of interest to us.

    Asymptotic pairwise independence of $(\lambda_i, \tilde r_{i,(i)})$'s implies that
    \begin{equation*}
        \var (g_n(x_0)) = \mo(1).
    \end{equation*}
    and therefore gives the first result.

    Now we pick $\epsilon > 0$.
    By the stochastic equicontinuity of $g_n$ and the bound in \eqref{eq_bound_gn}, we can find $x_1, \ldots, x_K$, independent of $n$, such that for all $x \in [0, B]$, there exists $l$ such that when $n$ is large enough,
    \begin{equation*}
        \E (|g_n(x) - g_n(x_l)|) \leq \epsilon.
    \end{equation*}
    Notice that
    \begin{equation*}
        |g_n(x) - G_n(x)| \leq |g_n(x) - g_n(x_l)| + |g_n(x_l) - G_n(x_l)| + |G_n(x_l) - G_n(x)|.
    \end{equation*}
    We immediately get
    \begin{equation*}
        \E^* \Big( \sup_{0 \leq x \leq B} |g_n(x) - G_n(x)| \Big) \leq 2 \epsilon + \E \Big( \sup_{1 \leq l \leq K} |g_n(x_l) - G_n(x_l)| \Big).
    \end{equation*}
    Because $K$ is finite, the fact that for all $l$, $g_n(x_l) - G_n(x_l) = \mo_{L_2}(1)$ implies that $\sup_{1 \leq l \leq K} |g_n(x_l) - G_n(x_l)| = \mo_{L_2}(1)$.
    In particular, if $n$ is sufficiently large, we have
    \begin{equation*}
        \E \Big(\sup_{1 \leq l \leq K} |g_n(x_l) - G_n(x_l)| \Big) \leq \epsilon.
    \end{equation*}
    This gives the result.
\end{proof}

\begin{lemma}
    \label{lemma:ctau_near_solution}
    Assume \textbf{O1-O7}, \textbf{P1-P4} and \textbf{F1-F5}.
    Recall that $c_\tau = n^{-1} \tr\{ (\bfS + \tau \I_p)^{-1} \}$.
    Call as before
    \begin{equation*}
        g_n(x) = \frac{1}{n} \sum_{i=1}^n \frac{1}{1 + x \lambda_i^2 \psi' \{\prox(x \lambda_i^2 \rho) (\tilde r_{i,(i)})\}}.
    \end{equation*}
    Then $c_\tau$ is a near solution of
    \begin{equation*}
        \frac{p}{n} - \tau x - 1 + g_n(x) = 0,
    \end{equation*}
    i.e. $p/n - \tau c_\tau - 1 + g_n(c_\tau) = \mo_{L_k} (1)$, when $3\alpha - 1/2 < 0$.

    Asymptotically, near solution of
    \begin{equation*}
        \delta_n(x) \triangleq \frac{p}{n}-\tau x-1+g_n(x)=0,
    \end{equation*}
    are close to solutions of
    \begin{equation*}
        \Delta_n(x)=\frac{p}{n}-\tau x-1+\E \{g_n(x)\}=0.
    \end{equation*}
    More precisely, call $T_{n, \epsilon} = \{ x: |\Delta_n(x)| \leq \epsilon \}$.
    We note that $T_{n, \epsilon} \subseteq(0, p /(n \tau)+\epsilon / \tau)$.
    For any given $\epsilon$, as $n \to \infty$, near solutions of $\delta_n(x) = 0$ belong to $T_{n, \epsilon}$ with high probability.

    Our assumptions concerning the possible distributions of $\epsilon_i$'s, specifically Assumption \textbf{F1}, imply that as $n \to \infty$, there is a unique solution to $\Delta_n(x) = 0$.

    Hence $c_\tau$ is asymptotically deterministic.

\end{lemma}

\begin{proof}
    Note that $g_n(x) \leq 1$.

    Let $\delta_n(x)$ be the function
    \begin{equation*}
        \delta_n(x) = \frac{p}{n} - \tau x - 1 + g_n(x),
    \end{equation*}
    and $\Delta_n(x) = \E \{ \delta_n(x)\}$.
    Call $x_n$ a solution of $\delta_n(x_n) = 0$ and $x_{n,0}$ a solution of $\Delta_n(x_{n,0}) = 0$.

    Since $0 \leq g_n(x) \leq 1$, we have $x_n \leq p/(n \tau)$, for otherwise, $\delta_n(x) < 0$.
    The same arguments shows that if $x > (p/n + \epsilon) / \tau$, then $\delta_n(x) < - \epsilon$ and $x \notin T_{n, \epsilon}$.
    Similarly, near solutions of $\delta_n(x) = 0$ must be less or equal to $(p/n + \epsilon) / \tau$.

    \textbf{$\bullet$ Proof of the fact that $c_\tau$ is such that $\delta_n(c_\tau) = \mo(1)$}

    Recall that in the notation of Lemma \ref{lemma:trace_approx}, we have
    \begin{equation*}
        \frac{p-1}{n} - \tau \mc_{\tau,p} = \frac{1}{n} \tr(\I_n - \bfM).
    \end{equation*}
    According to \eqref{eq_tr_I_M},
    \begin{equation*}
        \frac{1}{n} \tr(\I_n - \bfM) = 1 - \frac{1}{n} \sum_{i = 1}^n \frac{1}{1 + \psi'(r_{i,[p]}) \frac{1}{n} \bfv_i^\T (\mathfrak{S}_p(i) + \tau \I_p)^{-1} \bfv_i}.
    \end{equation*}
    According to Lemmas \ref{lemma:trace_approx}, \ref{lemma:zeta_bound} and \ref{lemma:kn_bound}, we have
    \begin{equation*}
        \sup_i \Big| \frac{1}{n} \bfv_i^\T (\mathfrak{S}_p(i) + \tau \I_p)^{-1} \bfv_i - \lambda_i^2 \mc_{\tau,p} \Big| = \mO_{L_k} \Big(\frac{\polyLog(n)}{n^{1/2-2\alpha}}\Big).
    \end{equation*}
    When $x\geq 0$ and $y \geq 0$, $|1/(1+x) - 1/(1+y)| \leq |x-y| \wedge 1$.
    Hence, we have
    \begin{equation*}
        \begin{aligned}
            &\Big| \frac{1}{n} \sum_{i = 1}^n \frac{1}{1 + \psi'(r_{i,[p]}) \frac{1}{n} \bfv_i^\T (\mathfrak{S}_p(i) + \tau \I_p)^{-1} \bfv_i} - \frac{1}{n} \sum_{i=1}^n \frac{1}{1+ \psi'(r_{i,[p]}) \lambda_i^2 \mc_{\tau,p}} \Big| \\
            & \quad \leq \sup_{1 \leq i \leq n} \Big| \frac{1}{n} \bfv_i^\T (\mathfrak{S}_p(i) + \tau \I_p)^{-1} \bfv_i - \lambda_i^2 \mc_{\tau,p} \Big| \|\psi'\|_\infty .
        \end{aligned}
    \end{equation*}
    We conclude that
    \begin{equation*}
        p/n - \tau \mc_{\tau,p} - 1 + \frac{1}{n} \sum_{i = 1}^n \frac{1}{1 + \lambda_i^2 \mc_{\tau,p} \psi' (r_{i,[p]})} = \mO_{L_k} (n^{-1/2+2\alpha} \polyLog(n)).
    \end{equation*}
    Exactly the same computations can be done for $c_\tau$.
    We have established that
    \begin{equation*}
        p/n - \tau c_\tau - 1 + \frac{1}{n} \sum_{i = 1}^n \frac{1}{1 + \lambda_i^2 c_{\tau} \psi' (R_i)} = \mO_{L_k} (n^{-1/2+2\alpha} \polyLog(n)).
    \end{equation*}
    Now we have seen in Theorem \ref{th:l1oo} that
    \begin{equation*}
        \sup_i | R_i - \prox(c_i \rho) (\tilde r_{i,(i)}) | = \mO_{L_k} (n^{-1/2+\alpha} \polyLog(n)).
    \end{equation*}
    By the assumptions on $\psi'$, this implies that
    \begin{equation*}
        \sup_i | \psi'(R_i) - \psi' \{ \prox(c_i \rho) (\tilde r_{i,(i)}) \} | = \mO_{L_k} (n^{-1/2+\alpha} \polyLog(n)).
    \end{equation*}
    We have furthermore noted that $\sup_i |c_i - \lambda_i^2 c_\tau| = \mO_{L_k} (n^{-1/2+2\alpha} \polyLog(n))$ in Corollary \ref{cor:ci_ctau_gap}.
    Using Lemma \ref{lemma:prox_c_derivative}, we can write
    \begin{equation*}
        | \prox(c_i \rho) (\tilde r_{i,(i)}) - \prox(\lambda_i^2 c_\tau \rho) (\tilde r_{i,(i)}) | \leq \|\psi\|_\infty |c_i - \lambda_i^2 c_\tau|
    \end{equation*}
    and hence
    \begin{equation*}
        | \psi' \{ \prox (c_i \rho) (\tilde r_{i,(i)}) \} - \psi' \{ \prox(\lambda_i^2 c_\tau \rho) (\tilde r_{i,(i)}) \} | = \mO_{L_k} (\|\psi\|_\infty n^{-1/2+3\alpha} \polyLog(n)).
    \end{equation*}
    Gathering all these results, we have
    \begin{equation*}
        | \psi'(R_i) - \psi'( \prox(\lambda_i^2 c_\tau \rho) (\tilde r_{i,(i)}) ) | = \mO_{L_k} \{ (\|\psi\|_\infty + 1) n^{-1/2+3\alpha} \polyLog(n) \}.
    \end{equation*}
    So we have shown that $\delta_n(c_\tau) = \mO_{L_k} (n^{-1/2+3\alpha} \polyLog(n))$.

    \textbf{$\bullet$ Final details}

    By Lemma \ref{lemma:gn_uniform_lln}, we have $\delta_n(x) - \Delta_n(x) = \mo_p(1)$ for any given $x$.
    In our case, using the notation of this lemma, $B = p/(n \tau) + \eta/\tau$, for $\eta >0$ given.

    This implies that for any given $\epsilon > 0$,
    \begin{equation*}
        \sup_{x \in (0, p/(n \tau) + \eta/\tau)} | \delta_n(x) - \Delta_n(x) | < \epsilon,
    \end{equation*}
    with high probability when $n$ is large enough.
    Therefore, for any $\epsilon >0$, if $x_n$ is a solution to $\delta_n(x) = 0$,
    \begin{equation*}
        |\Delta_n(x_n)| < \epsilon \quad \text{with high probability}.
    \end{equation*}
    This means that $x_n \in T_{n, \epsilon}$ with high probability.
    The same reasoning applies for near solution of $\delta_n(x) = 0$, which must belong to $T_{n, \epsilon}$ as $n \to \infty$ with high probability for any given $\epsilon > 0$.
    Note that $T_{n, \epsilon}$ is compact because it is bounded and closed, using the fact that $G_n = \E(g_n)$ is continuous.

    If $T_{n,0}$ were reduced to a single point, we would have established the asymptotically deterministic character of $c_\tau$.

    Given our work on the asymptotic behavior of $\tilde r_{i,(i)}$ and our assumptions on $\epsilon_i$'s, we see that Lemma 3.39 from \citet{el2018impact} applies to $\lim_{n \to \infty} \Delta_n(x)$ under assumption \textbf{F1}.
    Therefore, $T_{n,0}$ is reduced to a single point as $n \to \infty$ and $c_\tau$ is asymptotically deterministic.

\end{proof}

\paragraph{Proof of Theorem \ref{thm:step_2}}

\begin{proof}[Proof of Theorem \ref{thm:step_2}]
    Notice that
    \begin{equation*}
        \frac{\partial}{\partial t} \prox(c \rho)(t)=\prox(c \rho)'(t)=\frac{1}{1+c \psi'(\prox(c \rho)(t))} .
    \end{equation*}
    Therefore, $\Delta_n$ can be written as
    \begin{equation*}
        \Delta_n(x)=\frac{p}{n}-\tau x-1+\frac{1}{n} \sum_{i=1}^n \E \big\{ \prox'(x \lambda_i^2 \rho) (\tilde r_{i,(i)}) \big\}.
    \end{equation*}
    Hence, the limiting root of $\Delta_n(x)=0$ satisfies the first fixed-point equation in Theorem \ref{thm:step_2}.
    Since Lemma \ref{lemma:ctau_near_solution} shows that $c_\tau$ is asymptotically arbitrarily close to this root, the first equation follows.
    The second equation comes from \eqref{eq_system_2}.
    Theorem \ref{thm:step_2} is proved, with $c_\rho(\kappa)$ being the limit of $c_\tau$.

\end{proof}
\end{document}